\documentclass[11pt]{article}
\usepackage[margin=1in]{geometry}
\usepackage[utf8]{inputenc}

\usepackage[pdftex]{graphicx}

\usepackage{pdfsync}

\usepackage{wrapfig}
\usepackage{amsmath}
\usepackage{amsfonts}
\usepackage{amssymb}
\usepackage{amsthm}
\usepackage{tikz} %
\usetikzlibrary{patterns}
\usepackage{subcaption}

\usepackage{thmtools}
\usepackage{todonotes}

\usepackage[ruled,noresetcount]{algorithm2e}

\usepackage{silence}
\ActivateWarningFilters[pdftoc]
\usepackage{hyperref}
\usepackage{cleveref}
\usepackage{autonum}
\crefname{thm}{theorem}{theorem}
\crefname{conj}{conjecture}{conjecture}
\crefname{lemma}{lemma}{lemma}

\usepackage{mathtools}
\usepackage{multirow}

\newtheorem{definition}{Definition}
\newtheorem{theorem}{Theorem}

\newtheorem{corollary}{Corollary}
\newtheorem{lemma}{Lemma}
\numberwithin{lemma}{section}

\newcommand{\R}{\mathbb R}
\newcommand{\eps}{\varepsilon}

\DeclareMathOperator*{\E}{E}

\DeclareMathOperator*{\poly}{poly}

\DeclareMathOperator*{\argmin}{arg\,min}
\DeclareMathOperator*{\arginf}{arg\,inf}

\newcommand{\zo}{\{0,1\}}

\newcommand{\abs}[1]{\left| #1 \right|}
\newcommand{\norm}[2]{\left\| #1 \right\|_{#2}}
\newcommand{\ep}[2][]{\E_{#1}\left[ #2 \right]}

\newcommand{\prp}[2][]{\Pr_{#1}\left[ #2 \right]}
\newcommand{\prpcond}[2]{\Pr\left[ #1 \, \middle| \, #2 \right]}

\newcommand{\set}[1]{\left\{ #1 \right\}}
\newcommand{\setbuilder}[2]{\left\{ #1 \, \middle| \, #2 \right\}}
\newcommand{\dd}{\mathrm{d}}
\newcommand{\ceil}[1]{\lceil #1 \rceil}

\DeclarePairedDelimiterX{\infdivx}[2]{(}{)}{#1 \,\delimsize\|\, #2}
\DeclareMathOperator*{\Div}{D}
\newcommand{\D}{\Div\infdivx}
\DeclareMathOperator*{\smallDiv}{d}
\newcommand{\di}{\smallDiv\infdivx}
\newcommand{\dq}{\di{t_q}{w_q}}
\newcommand{\du}{\di{t_u}{w_u}}

\newcommand{\smat}[1]{\left[\begin{smallmatrix}#1\end{smallmatrix}\right]}
\newcommand{\svec}[1]{[\begin{smallmatrix}#1\end{smallmatrix}]}

\begin{document}

\author{
   Thomas D.~Ahle\\ \small BARC, IT University of Copenhagen\\ \small \texttt{thdy@itu.dk}
   \and
   Jakob B.~T.~Knudsen\\ \small BARC, University of Copenhagen\\ \small \texttt{jakn@di.ku.dk}
}
\title{ Subsets and Supermajorities:\\ Optimal Hashing-based Set Similarity Search
   \footnote{A previous version of this manuscript has appeared on arXiv.org under the title ``Subsets and Supermajorities: Unifying Hashing-based Set Similarity Search''.}
}
\date{April 15, 2020}
\hypersetup{pdfauthor={Thomas Dybdahl Ahle, Jakob Tejs Knudsen},pdftitle={Subsets & Supermajorities}}

\maketitle

\begin{abstract}

We formulate and optimally solve a new generalized Set Similarity Search problem,
which assumes the size of the database and query sets are known in advance.
By creating polylog copies of our data-structure, we optimally solve any symmetric Approximate Set Similarity Search problem, including approximate versions of Subset Search, Maximum Inner Product Search (MIPS), Jaccard Similarity Search and Partial Match.

Our algorithm can be seen as a natural generalization of previous work on Set as well as Euclidean Similarity Search, but conceptually it differs by optimally exploiting the information present in the sets as well as their complements, and doing so asymmetrically between queries and stored sets.
Doing so we improve upon the best previous work:
MinHash [J. Discrete Algorithms 1998], SimHash [STOC 2002], Spherical LSF
[SODA 2016, 2017] and Chosen Path [STOC 2017] by as much as a factor $n^{0.14}$ in both time and space; or in the near-constant time regime, in space, by an arbitrarily large polynomial factor.

Turning the geometric concept, based on Boolean supermajority functions, into a practical algorithm requires ideas from branching random walks on $\mathbb Z^2$, for which we give the first non-asymptotic near tight analysis.

Our lower bounds follow from new hypercontractive arguments, which can be seen as characterizing the exact family of similarity search problems for which supermajorities are optimal.
The optimality holds for among all hashing based data structures in the random setting, and by reductions, for 1 cell and 2 cell probe data structures.
As a side effect, we obtain new hypercontractive bounds on the directed noise operator $T^{p_1 \to p_2}_\rho$.
\end{abstract}

\clearpage

\tableofcontents

\section{Introduction}

Set Similarity Search (SSS) is the problem of indexing sets (or sparse boolean data) to allow fast retrieval of sets, similar under a given similarity measure.
The sets may represent one-hot encodings of categorical data, ``bag of words'' representations of documents, or ``visual/neural bag of words'' models, such as the Scale-invariant feature transform (SIFT), that have been discretized.
The applications are ubiquitous across Computer Science, touching everything from recommendation systems to gene sequences comparison.
See~\cite{choi2010survey, jia2018survey} for recent surveys of methods and applications.

Set similarity measures are any function, $s$ that takes two sets and return s value in $[0,1]$.
Unfortunately, most variants of Set Similarity Search, such as Partial Match, are hard to solve assuming popular conjectures around the Orthogonal Vectors Problem~\cite{williams2005new, ahle2015complexity, abboud2017distributed, chen2019equivalence}, which roughly implies that the best possible algorithm is to not build an index, and ``just brute force'' scan through all the data, on every query.
A way to get around this is to study Approximate SSS:
Given a query, $q$, for which the most similar set $y$ has $\text{similarity}(q,y)\ge s_1$, we are allowed to return any set $y'$ with $\text{similarity}(q,y')>s_1$, where $s_2<s_1$.
In practice, even the best \emph{exact} algorithms for similarity search use such an $(s_1,s_2)$-approximate\footnote{By classical reductions~\cite{har2012approximate} we can assume $s_1$ is known in advance.}
solution as a subroutine~\cite{christiani2018confirmation}.

Euclidean Similarity Search, where the data is vectors $x\in\R^d$ and the measure of similarity is ``Cosine'', has recently been solved optimally --- at least in the model of hashing based data structures~\cite{andoni2015optimal, andoni2017optimal}.
Meanwhile, the problem on sets has proven much less tractable.
This is despite that the first solutions date back to the seminal MinHash algorithm (a.k.a. min-wise hashing), introduced by Broder et al.~\cite{broder1997syntactic, broder1997resemblance} in 1997 and by now boasting thousands of citations.
In 2014 MinHash was shown to be near-optimal for set intersection \emph{estimation}~\cite{pagh2014min}, but in a surprising recent development, it was shown not to be optimal for similarity \emph{search}~\cite{tobias2016}.
The question thus remained: What \emph{is} the optimal algorithm for Set Similarity Search?

The question is made harder by the fact that previous algorithms study the problem under different similarity measures, such as Jaccard, Cosine or Braun-Blanquet similarity.
The only thing those measures have in common is that they can be defined as a function $f$ of the sets sizes, the universe size and the intersection size.
In other words, $\text{similarity}(q,y) = f(|q|, |y|, |q \cap y|, |U|)$ where $|U|$ is the size of the universe from which the sets are taken.
In fact, any symmetric measure of similarity for sets must be defined by those four quantities.

Hence, to fully solve Set Similarity Search, we avoid specifying a particular similarity measure, and instead define the problem solely from those four parameters.
This generalized problem is what we solve optimally in this paper, for all values of the four parameters:
\begin{definition}[The $(w_q,w_u,w_1,w_2)$-GapSS problem]
   Given some universe $U$ and a collection $Y\subseteq\binom{U}{w_u|U|}$ of $|Y|=n$ sets of size $w_u|U|$,
   build a data structure that for any query set $q\in\binom{U}{w_q|U|}$:
   either returns $y'\in Y$ with $|y'\cap q| > w_2 |U|$;
   or determines that there is no $y\in Y$ with $|y\cap q| \ge w_1 |U|$.
\end{definition}

For the problem to make sense, we assume that $w_q|U|$ and $w_u|U|$ are integers, that $w_q,w_u\in[0,1]$, and that $0 < w_2 < w_1 \le \min\{w_q,w_u\}$.
Note that $|U|$ may be very large, and as a consequence the values $w_q,w_u,w_1,w_2$ may all be very small.

At first sight, the problem may seem easier than the version where the sizes of sets may vary.
However, the point is that making $\text{polylog}(n)$ data-structures for sets and queries of progressively bigger sizes,\footnote{%
For details, see~\cite{tobias2016} Section 5.
A similar reduction, called ``norm ranging'', was recently shown at NeurIPS to give state of the art results for Maximum Inner Product Search in $\R^d$~\cite{yan2018norm}, suggesting it is very practical.
} immediately yields data structures for the original problem.
Similarly, any algorithm assuming a specific set similarity measure also yields an algorithm for $(w_q,w_u,w_1,w_2)$-GapSS, so our lower bounds too hold for all previously studied SSS problems.

\vspace{-.5em}
\paragraph{Example 1}
As an example, assume we want to solve the Subset Search Problem, in which we, given a query $q$, want to find a set $y$ in the database, such that $y\subseteq q$.
If we allow a two-approximate solution,
GapSS includes this problem by setting $w_1=w_u$ and $w_2=w_1/2$: The overlap between the sets must equal the size of the stored sets;
and we are guaranteed to return a $y'$ such that at least $|q\cap y'|\ge |y|/2$.

\vspace{-.5em}
\paragraph{Example 2}
In the $(j_1,j_2)$-Jaccard Similarity Search Problem, given a query, $q$, we must find $y$ such that the Jaccard Similarity $|q\cap y|/|q\cup y|>j_2$ given that a $y'$ exists with similarity at least $j_1$.
After partitioning the sets by size, we can solve the problem using GapSS by setting $w_1=\frac{j_1(w_q+w_u)}{1+j_1}$ and $w_2=\frac{j_2(w_q+w_u)}{1+j_2}$.
The same reduction works for any other similarity measure with $\text{polylog}(n)$ overhead.

\vspace{1em}

The version of this problem where $w_2=w_qw_u$ is similar to what is in the literature called \emph{``the random instance''}~\cite{panigrahy2006entropy, laarhoven2015tradeoffs, andoni2016optimal}.
To see why, consider generating
$n-1$ sets independently at random with size $w_u|U|$,
and a ``planted'' pair, $(q, y)$, with size respectively $w_q|U|$ and $w_u|U|$ and with intersection $|q\cap y|=w_1|U|$.
Insert the size $w_u |U|$ sets into the database and query with $q$.
Since $q$ is independent from the $n-1$ original sets, its intersection with those is strongly concentrated around the expectation $w_qw_u|U|$.
Thus, if we parametrize GapSS with $w_2=w_qw_u+o(1)$, the query for $q$ is guaranteed to return the planted set $y$.  %

There is a tradition in the Similarity Search literature for studying such this independent case, in part because \emph{it is expected that one can always reduce to the random instance},
for example using the techniques of ``data-dependent hashing''~\cite{andoni2014beyond, andoni2015optimal}.
However, for such a reduction to make sense, we would first need an optimal ``data-independent'' algorithm for the $w_2=w_qw_u$ case, which is what we provide in this paper.
We discuss this further in the Related Work section.

For generality we still define the problem for all $w_2\in(0,w_1)$, our upper bound holds in this general setting and so does the lower bound~\Cref{thm:donnell}.

\vspace{1em}

We give our new results in \Cref{sec:upper_intro} and our new lower bounds in \Cref{sec:lower_intro}, but first we would like to sketch the algorithm and some probabilistic tools used in the theorem statement.

\subsection{Supermajorities}\label{sec:supermajorities}

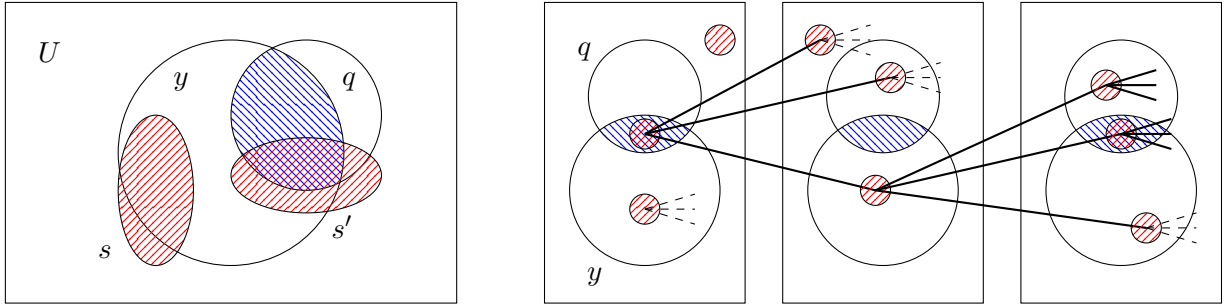
\begin{figure}
   \begin{subfigure}{.38\textwidth} %
      \centering
      \begin{tikzpicture}
         \draw (0,0) rectangle (6,4);
         \node[anchor=north west] at (0.3, 3.6) {$U$};
         \begin{scope}
            \clip (4,2.5) circle (1);
            \fill[pattern=north west lines, pattern color=blue] (3,2) circle (1.5);
         \end{scope}
         \draw (4,2.5) circle (1);
         \node[anchor=north east] at (4.8, 3.2) {$q$};
         \draw (3,2) circle (1.5);
         \node[anchor=north west] at (2.1, 3.2) {$y$};
         \draw[pattern=north east lines, pattern color=red] (4,1.7) ellipse (1cm and .5cm);
         \node[anchor=north west] at (1.1, .9) {$s$};
         \draw[pattern=north east lines, pattern color=red] (2,1.5) ellipse (.5cm and 1cm);
         \node[anchor=north west] at (4.2, 1.3) {$s'$};
      \end{tikzpicture}
      \caption{Two cohorts, $y$ and $q$ with a large intersection (blue).
         The first representative set, $s$, favours $y$, while the second, $s'$, favours both $y$ and $q$.
      \\}
      \label{fig:cohorts}
   \end{subfigure}
   \hfill
   \begin{subfigure}{.57\textwidth} %
      \centering
      \begin{tikzpicture}
         \foreach \x in {0,...,2} {
            \begin{scope}
               \clip (8/3*\x+1/2*\x + 4/3, 1.5) circle (1);
               \fill[pattern=north west lines, pattern color=blue] (8/3*\x+1/2*\x + 4/3, 2.75) circle (3/4);
            \end{scope}
            \draw (8/3*\x+1/2*\x + 4/3, 1.5) circle (1);
            \draw (8/3*\x+1/2*\x + 4/3, 2.75) circle (3/4);
            \draw (8/3*\x+1/2*\x, 0) rectangle (8/3*\x+1/2*\x+8/3,4);
         }
         \node[anchor=north] at (1/3+.2, 3.6) {$q$};
         \node[anchor=south] at (2/3, .1) {$y$};
         \draw[pattern=north east lines, pattern color=red] (7/3,3.5) circle (.2);
         \draw[pattern=north east lines, pattern color=red] (4/3,2.25) circle (.2);
         \draw[pattern=north east lines, pattern color=red] (4/3,1.25) circle (.2);
         \draw[dashed] (4/3,1.25) -- (6/3, 1.05);
         \draw[dashed] (4/3,1.25) -- (6/3, 1.25);
         \draw[dashed] (4/3,1.25) -- (6/3, 1.45);
         \draw[thick] (4/3,2.25) -- (12/3+.4,1.5);
         \draw[thick] (4/3,2.25) -- (12/3+.6,3);
         \draw[thick] (4/3,2.25) -- (11/3,3.5);
         \draw[pattern=north east lines, pattern color=red] (11/3,3.5) circle (.2);
         \draw[dashed] (11/3,3.5) -- (13/3,3.7);
         \draw[dashed] (11/3,3.5) -- (13/3,3.5);
         \draw[dashed] (11/3,3.5) -- (13/3,3.3);
         \draw[pattern=north east lines, pattern color=red] (12/3+.6,3) circle (.2);
         \draw[dashed] (12/3+.6,3) -- (14/3+.6,2.8);
         \draw[dashed] (12/3+.6,3) -- (14/3+.6,3);
         \draw[dashed] (12/3+.6,3) -- (14/3+.6,3.2);
         \draw[pattern=north east lines, pattern color=red] (12/3+.4,1.5) circle (.2);
         \draw[thick] (12/3+.4,1.5) -- (21/3+1,1);
         \draw[thick] (12/3+.4,1.5) -- (20/3+1,2.25);
         \draw[thick] (12/3+.4,1.5) -- (20/3+.8,2.9);
         \draw[pattern=north east lines, pattern color=red] (20/3+.8,2.9) circle (.2);
         \draw[thick] (20/3+.8,2.9) -- (22/3+.8,2.7);
         \draw[thick] (20/3+.8,2.9) -- (22/3+.8,2.9);
         \draw[thick] (20/3+.8,2.9) -- (22/3+.8,3.1);
         \draw[pattern=north east lines, pattern color=red] (20/3+1,2.25) circle (.2);
         \draw[thick] (20/3+1,2.25) -- (22/3+1,2.05);
         \draw[thick] (20/3+1,2.25) -- (22/3+1,2.25);
         \draw[thick] (20/3+1,2.25) -- (22/3+1,2.45);
         \draw[pattern=north east lines, pattern color=red] (21/3+1,1) circle (.2);
         \draw[dashed] (21/3+1,1) -- (23/3+1,.8);
         \draw[dashed] (21/3+1,1) -- (23/3+1,1);
         \draw[dashed] (21/3+1,1) -- (23/3+1,1.2);
      \end{tikzpicture}
      \caption{
         Branching random walk run
         on two cohorts $q$ and $y$.
         The bold lines illustrate paths considered by sets, while the dashed lines adorn paths only considered by only one of $x$ or $y$.
         Here $q$ has a higher threshold ($t_q=2/3$) than $y$ ($t_u=1/2$), so $q$ only considers paths starting with two favourable representatives.
      }
      \label{fig:tree}
   \end{subfigure}
   \caption{The representative sets, coloured in red, are scattered in the universe to provide an efficient space partition for the data.}
\end{figure}

In Social Choice Theory a supermajority is when a fraction strictly greater than $1/2$ of people agree about something.%
\footnote{
``America was founded on majority rule, not supermajority rule. Somehow, over the years, this has morphed into supermajority rule, and that changes things.'' -- Kent Conrad.
}
In the analysis of Boolean functions a $t$-supermajority function $f:\{0,1\}^n\to\{0,1\}$ can be defined as 1, if a fraction $\ge t$ of its arguments are 1, and 0 otherwise.
We will sometimes use the same word for the requirement that a fraction $\le t$ of the arguments are 1.%
\footnote{It turns out that defining everything in terms of having a fraction $t\pm o(1)$ of 1's is also sufficient.
   This is similar to Dubiner~\cite{dubiner2010bucketing}.
}

The main conceptual point of our algorithm is the realization that an optimal algorithm for Set Similarity Search must take advantage of the information present in the given sets, as well as that present in their complement.
A similar idea was leveraged by Cohen et al.~\cite{cohen2009leveraging} for Set Similarity \emph{Estimation}, and we show in~\Cref{sec:minhashdom} that the classical MinHash algorithm can be seen as an average of functions that pull varying amounts of information from the sets and their complements.
In this paper, we show that there is a better way of combining this information, and that doing so results in an optimal hashing based data structure for the entire parameter space of random instance GapSS.

This way of combining this information is by supermajority functions.
While on the surface they will seem similar to the threshold methods applied for time/space trade-offs in Spherical LSF~\cite{andoni2017optimal}, our use of them is very different.
Where~\cite{andoni2017optimal} corresponds to using small $t=1/2+o(1)$ thresholds (essentially simple majorities) our $t$ may be as large as 1 (corresponding to the AND function) or as small as 0 (the NOT AND function).
This way they are a sense as much a requirement on the complement as it is on the sets themselves.

\vspace{.5em}

\emph{The algorithm (idealized):}
While our data structure is technically a tree with a carefully designed pruning rule, the basic concept is very simple.

We start by sampling a large number of ``representative sets'' $R\subseteq\binom{U}{k}$.
Here roughly $|R|\approx n^{\log n}$ and $k\approx \log n$.
Given family $Y\subseteq\binom{U}{w_u|U|}$ of sets to store, which we call ``cohorts'',
we say that $r\in R$ ``$t$-favours'' the cohort $y$ if $|y\cap r|/|r| \ge t$.
Representing sets as vectors in $\{0,1\}^d$, this is equivalent to saying $f_{t}(r\cap y)=1$, where $f_{t}$ is the $t$-supermajority function.
(If $t$ is less than $w_u$, the expected size of the overlap, we instead require $|y\cap r|/|r|\le t$.)

Given the parameters $t_q,t_u\in[0,1]$,
the data-structure is a map from elements of $R$ to the cohorts they $t_u$-favour.
When given a query $q\in\binom{U}{w_q|U|}$, (a $w_q|U|$ sized cohort), we compare it against all cohorts $y$ favoured by representatives $r\in R$ which $t_q$-favour $q$ (that is $|q\cap r|/|r|\ge t_q$).
This set $R_{t_q}(q)$ is much smaller than $|R|$ (we will have $|R_{t_q}(q)|\approx n^\eps$ and $E[|R_{t_u}(y)\cap R_{t_q}(q)|]\approx n^{\eps-1}$),
so the filtering procedure greatly reduces the number of cohorts we need to compare to the query from $n$ to $n^\eps$ (where $\eps=\rho_q<1$ is defined later.)

The intuition is that while it is quite unlikely for a representative to favour a given cohort, and it is \emph{very} unlikely for it to favour two given cohorts ($q$ and $y$).
So if it does, the two cohorts probably have a substantial overlap.
\Cref{fig:cohorts} has a simple illustration of this principle.

\vspace{.5em}

In order to fully understand supermajorities,
we want to understand the probability that a representative set is simultaneously in favour of two distinct cohorts given their overlap and representative sizes.
This paragraph is a bit technical, and may be skipped at first read.
Chernoff bounds in $\R$ are a common tool in the community, and for iid. $X_i\sim \text{Bernoulli}(p)\in\{0,1\}$ the sharpest form (with a matching lower bound) is
$\Pr[\sum X_i \ge tn] \le \exp(-n \di tp)$,\footnote{A special case of Hoeffding's inequality is obtained by $\di{p+\eps}p\ge 2\eps^2$, Pinsker's inequality.} which uses the binary KL-Divergence $\di{t}{p}=t \log\frac{t}{p}+(1-t)\log\frac{1-t}{1-p}$.
The Chernoff bound for $\R^2$ is less common, but likewise has a tight description in terms of the KL-Divergence between two discrete distributions: $\D{P}{Q} = \sum_{\omega\in\Omega} P(\omega)\log\frac{P(\omega)}{Q(\omega)}$ (summing over the possible events).
In our case, we represent the four events that can happen as we sample an element of $U$ as a vector $X_i\in\{0,1\}^2$.
Here $X_i=[\begin{smallmatrix}1\\1\end{smallmatrix}]$ means the $i$th element hit both cohorts, $X_i=[\begin{smallmatrix}1\\0\end{smallmatrix}]$ means it hit only the first and so on.
We represent the distribution of each $X_i$ as a matrix $P=\left[\begin{smallmatrix}w_1 & w_q-w_1\\w_u-w_1 & 1-w_q-w_u+w_1\end{smallmatrix}\right]$,
and say $X_i \sim \text{Bernoulli}(P)$ iid.
such that $\Pr[X_i = \svec{1-j\\1-k}] = P_{j,k}$.
Then $\Pr[\sum X_i \ge [\begin{smallmatrix}t_q\\t_u\end{smallmatrix}]n] \approx \exp(-n \D{T}{P})$
where
$T=\left[\begin{smallmatrix}t_1 & t_q-t_1\\t_u-t_1 & 1-t_q-t_u+t_1\end{smallmatrix}\right]$
and $t_1\in[0,\min\{t_u,t_q\}]$ minimizes $\D{T}{P}$.
(Here the notation $\svec{x\\y} \ge \svec{t_u\\t_q}$ means $x \ge t_u \wedge y \ge t_q$.)

The optimality of Supermajorities for Set Similarity Search is shown using a certain correspondence we show between the Information Theoretical quantities described above, and the hypercontractive inequalities that have been central in all previous lower bounds for similarity search.

These bounds above would immediately allow a cell probe version of our upper bound \Cref{thm:main},
e.g. a query would require $n^{\frac{\D{T_1}{P_1} - \dq}{\D{T_2}{P_2}-\dq}}$ probes,
where $P_i=\smat{w_i&w_q-w_1\\w_u-w_i&1-w_q-w_u+w_i}$ and $T_i$ defined accordingly.
The algorithmic challenge is that, for optimal performance, $|R|$ must be in the order of $\Omega(n^{\log n})$, and so checking which representatives favour a given cohort takes super polynomial time!

The classical approach to designing an oracle to efficiently yield all such representatives,
, is a product-code or ``tensoring trick''.
The idea, (used by~\cite{christiani2016framework, becker2016new}), is to choose a smaller $k'\approx \sqrt{k}$, make $k/k'$ different $R'_i$ sets of size $n^{\sqrt{\log n}}$ and take $R$ as the product $R'_1 \times \dots \times R'_{k/k'}$.
As each $R'$ can now be decoded in $n^{o(1)}$ time, so can $R$.
This approach, however, in the case of Supermajorities, has a big drawback:
Since $t_q k'$ and $t_u k'$ must be integers, $t_q$ and $t_u$ have to be rounded and thus distorted by a factor $1+1/k'$.
Eventually, this ends up costing us a factor $w_1^{-k/k'}$ which can be much larger than $n$.
For this reason, we need a decoding algorithm that allows us to use supermajorities with as large a $k$ as possible!

We instead augment the above representative sampling procedure as follows:
Instead of independent sampling sets, we (implicitly) sample a large, random height $k$ tree, with nodes being elements from the universe.
The representative sets are taken to be each path from the root to a leaf.
Hence, some sets in $R$ share a common prefix, but mostly they are still independent.
We then add the extra constraint that \emph{each of the prefixes of a representative has to be in favour of a cohort}, rather than only having this requirement on the final set.
This is the key to making the tree useful:
Now given a cohort, we walk down tree, pruning any branches that do not consistently favour a supermajority of the cohort.
\Cref{fig:tree} has a simple illustration of this algorithm and \Cref{alg:treesample} has a pseudo-code implementation.
This pruning procedure can be shown to imply that we only spend time on representative sets that end up being in favour of our cohort, while only weakening the geometric properties of the idealized algorithm negligibly.

While conceptually simple and easy to implement (modulo a few tricks to prevent dependency on the size of the universe, $|U|$), the pruning rule introduces dependencies that are quite tricky to analyze sufficiently tight.
The way to handle this will be to consider the tree as a ``branching random walk'' over $\mathbb Z^2_+$ where the value represents the size of the representative's intersection with the query and a given set respectively.
The paths in the random walk at step $i$ must be in the quadrant $[t_q i, i]\times[t_u i,i]$ while only increasing with a bias of $\svec{w_q\\w_u}$ per step.
The branching factor is carefully tuned to just the right number of paths survive to the end.

\begin{center}
\emph{The ``history'' aspect of the pruning is a very important property of our algorithm, and is where it conceptually differs from all previous work.}
\end{center}
Previous Locality Sensitive Filtering, LSF, algorithms~\cite{tobias2016, andoni2016optimal} can be seen as trees with pruning, but their pruning is on the individual node level, rather than on the entire path.
This makes a big difference in which space partitions can be represented, since pruning on node level ends up representing the intersection of simple partitions, which can never represent Supermajorities in an efficient way.
In~\cite{becker2016new} a similar idea was discussed heuristically for Gaussian filters, but ultimately tensoring was sufficient for their needs, and the idea was never analyzed.

\subsection{Upper Bounds}\label{sec:upper_intro}

As discussed, the performance of our algorithm is described in terms of KL-divergences.
To ease understanding, we give a number of special cases, in which the general bound simplifies.
The bounds in this section assume $w_q,w_u,w_1,w_2$ are constants.
See \Cref{sec:full} for a version without this assumption.

\begin{theorem}[Simple Upper Bound]\label{thm:main}
   For any choice of constants $w_q, w_u\ge w_1\ge w_2\ge 0$ and $1\ge t_q, t_u\ge 0$
   we can solve the $(w_q,w_u,w_1,w_2)$-GapSS problem over universe $U$
   with query time $\tilde O(n^{\rho_q} + w_q|U|) + n^{o(1)}$ and auxiliary space usage $\tilde O(n^{1+\rho_u})$,
   where
   \begin{align}
      \rho_q = \frac{\D{T_1}{P_1}-\dq}{\D{T_2}{P_2}-\dq},
      \quad
      \rho_u = \frac{\D{T_1}{P_1}-\du}{\D{T_2}{P_2}-\dq}.
   \end{align}
   and $T_1$, $T_2$ are distributions with expectation $\svec{t_q\\t_u}$ minimizing respectively $\D{T_1}{P_1}$ and $\D{T_2}{P_2}$, as described in \Cref{sec:supermajorities}.
\end{theorem}

The two bounds differ only in the $\dq$ and $\du$ terms in the numerator.
The thresholds $t_q$ and $t_u$ can be chosen freely in $[0,1]^2$.
Varying them compared to each other allows a full space/time trade-off with $\rho_q=0$ in one end and $\rho_u=0$ (and $\rho_q<1$) in the other.
Note that for a given GapSS instance, there are many $(t_q,t_u)$ which are not optimal anywhere on the space/time trade-off.
Using Lagrange's condition $\nabla \rho_q = \lambda \nabla \rho_u$ 
one gets a simple equation that all optimal $(t_q, t_u)$ trade-offs must satisfy.
As we will discuss later, it seems difficult to prove that a solution to this equation is unique, but in practice it is easy to solve and provides an efficient way to optimize $\rho_q$ given a space budget $n^{1+\rho_u}$.
\Cref{fig:sp_comparison} and \Cref{fig:mh_comparison} provides some additional intuition for how the $\rho$ values behave for different settings of GapSS.

\vspace{.5em}

Regarding the other terms in the theorem, we note that the $\tilde O$ hides only $\log n$ factors, and the additive $n^{o(1)}$ term grows as $e^{O(\sqrt{\log n\log\log n})}$, which is negligible unless $\rho_q=0$.
We also note that there is no dependence on $|U|$, other than the need to store the original dataset and the additive $w_q|U|$, which is just the time it takes to receive the query.
The main difference between this theorem and the full version, is that the full theorem does not assume the parameters $(w_q,w_u,w_1,w_2)$ are constants, but consider them potentially very small.
In this more realistic scenario it becomes very important to limit the dependency on factors like $w_1^{-1}$, which is what guides a lot of our algorithmic decisions.

\vspace{-.5em}
\paragraph{Example 1: Near balanced $\rho$ values.}\label{par:example1}
As noted, many pairs $(t_q,t_u)$ are not optimal on the trade-off, in that one can reduce one or both of $\rho_q$, $\rho_u$ by changing them.
The pairs that are optimal are not always simple to express, so it is interesting to study those that are.
One such particularly simple choice on the Lagrangian is $t_q=1-w_u$ and $t_u=1-w_q$.\footnote{To make matters complicated, this \emph{is} a simple choice \emph{and} on the Lagrangian, but that doesn't prove another point on the Lagrangian won't reduce both $\rho_q$ and $\rho_u$ and thus be better. That we have a matching lower bound for the algorithm doesn't help, since it only matches the upper bound for $(t_q,t_u)$ minimal in \Cref{thm:main}.
In the case $w_q=w_u$ we can, however, prove that this $t_q,t_u$ pair is optimal.}
This point is special because the values of $t_q$ and $t_u$ depend only on $w_u$ and $w_q$, while in general they will also depend on $w_1$ and $w_2$.
In this setting we have $T_i=\left[\begin{smallmatrix}1-w_q-w_u+w_i & w_u-w_i\\w_u-w_i & w_i\end{smallmatrix}\right]$, which can be plugged into \Cref{thm:main}.

In the case $w_q=w_u=w$ we get the balanced $\rho$ values
$\rho_q = \rho_u = \log(\frac{w_1}{w}\frac{1-w}{1-2w+w_1})/\log(\frac{w_2}{w}\frac{1-w}{1-2w+w_2})$
in which case it is simple to compare with Chosen Path's $\rho$ value of
$\log(\frac{w_1}{w})\big/\log(\frac{w_2}{w})$.
Chosen Path on balanced sets was shown in~\cite{tobias2016} to be optimal for $w,w_1,w_2$ small enough, and we see that Supermajorities do indeed recover this value for that range.

We give a separate lower bound in \Cref{sec:explicit_hyp} showing that this value is in fact optimal when $w_2 = w_q w_u$.

\vspace{-.5em}
\paragraph{Example 2: Subset/superset queries.}\label{cor:subset}

If $w_1=\min\{w_u,w_q\}$ and $w_2=w_uw_q$ we can take
$t_q = \frac{\alpha}{w_q-w_u} + \frac{w_q(1-w_u)}{w_q-w_u}$
and
$t_u = \frac{w_u(1-w_u)w_q(1-w_q)}{w_q-w_u} \alpha^{-1} - \frac{w_u(1-w_q)}{w_q-w_u}$
for any
$\alpha\in [w_1-w_qw_u$, $\max\{w_u,w_q\}-w_qw_u]$.
\Cref{thm:main} then gives data structures with
\begin{align}
   \rho_q &= \frac{t_q\log\frac{1-t_u}{1-w_u}-t_u\log\frac{1-t_q}{1-w_q}}{\du}
          &\rho_u &= \frac{(1-t_u)\log\frac{t_q}{w_q}-(1-t_q)\log\frac{t_u}{w_u}}{\du}
   \quad &\text{if }w_1=w_u,
   \\
   \rho_q &= \frac{-(1-t_u)\log\frac{t_q}{w_q}+(1-t_q)\log\frac{t_u}{w_u}}{\du}
   &\rho_u &= \frac{-t_q\log\frac{1-t_u}{1-w_u}+t_u\log\frac{1-t_q}{1-w_q}}{\du}
   \quad &\text{if }w_1=w_q.
\end{align}
This represents one of the cases where we can solve the Lagrangian equation to get a complete characterization of the $t_q$, $t_u$ values that give the optimal trade-offs.
Note that when $w_1=w_u$ or $w_1=w_q$, the $P$ matrix as used in the theorem has 0's in it.
The only way the KL-divergence $\D{ T }{ P }$ can then be finite is by having the corresponding elements of $T$ be 0 and use the fact that $0\log\frac{0}{q}$ is defined to be $0$ in this context.

\vspace{-.5em}
\paragraph{Example 3: Linear space/constant time.}
Setting $t_1$ in $T_1=\left[\begin{smallmatrix}t_1 & t_q-t_1\\t_u-t_1 & 1-t_q-t_u+t_1\end{smallmatrix}\right]$ such that either $\frac{t_1}{w_1} = \frac{t_q - t_1}{w_q - w_1}$
or $\frac{t_1}{w_1} = \frac{t_u - t_1}{w_u - w_1}$ we get respectively $\D{ T_1 }{ P_1 } = \dq$ or $\D{ T_1 }{ P_1 } = \du$.
\Cref{thm:main} then yields algorithms with either $\rho_q=0$ or $\rho_u=0$
corresponding to either a data structure with $\approx e^{\tilde O(\sqrt{\log n})}$ query time, or with $\tilde O(n)$ auxiliary space.
Like~\cite{andoni2017optimal} we have $\rho_q<1$ for any parameter choice, even when $\rho_u=0$.
For very small $w_q$ and $w_u < \exp(-\sqrt{\log n})$ there are some extra concerns which are discussed after the main theorem.

\subsection{Lower Bounds}\label{sec:lower_intro}

Results on approximate similarity search are usually phrased in terms of two quantities:
(1) The ``query exponent''  $\rho_q \in [0,1]$ which
determines the query time by bounding it by $O(n^{\rho_q})$;
(2) The ``update exponent'' $\rho_u \in [0,1]$ which
determines the time required to update the data structure when a point is inserted or deleted in  $Y$ and is given by $O(n^{\rho_u})$.
The update exponent also bounds the space usage as $O(n^{1 + \rho_u})$.
Given parameters $(w_q, w_u, w_1, w_2)$, the important question is for which pairs of $(\rho_q, \rho_u)$ there exists data structures.
E.g. given a space budget imposed by $\rho_u$, we ask how small can one make $\rho_q$?

Since the first lower bounds on Locality Sensitive Hashing~\cite{motwani2006lower}, lower bounds for approximate near neighbours have split into two kinds:
(1) Cell probe lower bounds~\cite{panigrahy2008geometric, panigrahy2010lower, andoni2017optimal} and (2) Lower bounds in restricted models~\cite{o2014analysis, andoni2015tight, andoni2017optimal, tobias2016}.
The most general such model for data-independent algorithms was formulated by~\cite{andoni2017optimal} and defines a type of data structure called ``list of points'':
\begin{definition}[List-of-points]\label{def:list-of-points}
   Given some universes, $Q$, $U$, a similarity measure $S : Q\times U\to[0,1]$ and two thresholds $1\ge s_1 > s_2 \ge 0$,
   \begin{enumerate}
      \item We fix (possibly random) sets $A_i\subseteq\{-1,1\}^d$, for $1\le i\le m$; and with each possible query point $q\in\{-1,1\}^d$, we associate a (random) set of indices $I(q)\subseteq [m]$;
      \item For a given dataset $P$, we maintain $m$ lists of points $L_1, L_2, \dots, L_m$, where $L_i=P\cap A_i$.
      \item On query $q$, we scan through each list $L_i$ for $i\in I(q)$ and check whether there exists some $p\in L_i$ with $S(q,p) \ge s_2$. If it exists, return $p$.
   \end{enumerate}
   The data structure succeeds, for a given $q\in Q, p \in P$ with $S(q,p)\ge s_1$,
   if there exists $i\in I(q)$ such that $p\in L_i$. The total space is defined by
   $S = m + \sum_{i \in [m]}\abs{L_i}$ and the query time by $T = \abs{I(q)} + \sum_{i \in I(q)} \abs{L_i}$.
\end{definition}
The List-of-points model contains all known Similarity Search data structures, except for the so-called ``data-dependent algorithms''.
It is however conjectured~\cite{andoni2016optimal} that data-dependency does not help on random instances (recall this corresponds to $w_2=w_qw_u$), which is the setting of \Cref{thm:lower_main}.

We show two main lower bounds:
(1) That requires $w_q=w_u$ and $\rho_q=\rho_u$ and (2) That requires $w_2=w_qw_u$.
The second type is tight everywhere, but quite technical.
The first type meanwhile is quite simple to state, informally:

\begin{theorem}\label{thm:donnell}
   If $w_q=w_u=w$ and $\rho_u=\rho_q=\rho$, any data-independent LSF data structure must use space $n^{1+\rho}$ and have query time $n^{\rho}$ where
      $
      \rho \ge \log(\frac{w_1-w^2}{w(1-w)})
      \big/\log(\frac{w_2-w^2}{w(1-w)})
      $
      .
\end{theorem}
The LSF Model defined in~\cite{becker2016new, tobias2016} generalizes~\cite{motwani2006lower, o2014optimal}, but is slightly stronger than list-of-points.
It is most likely that they are equivalent, so we defer its definition till~\Cref{defn:lsf}.
We will just note that previous bounds of this type~\cite{o2014optimal, tobias2016} were only asymptotic, whereas our lower bound holds over the entire range of $0<w_2<w_1<w<1$.
By comparison with $\rho=\log(\frac{w_1(1-w)}{w(1-2w+w_1)})/\log(\frac{w_2(1-w)}{w(1-2w+w_2)})$ from Example 1 in the Upper Bounds section, we see that the lower bound is sharp when
$w,w_1,w_2\to 0$\footnote{
As $w,w_1,w_2\to 0$ we recover the lower bound $\rho\ge\log(\frac{w_1}{w})\big/\log\left(\frac{w_2}{w}\right)$ obtained for Chosen Path in~\cite{tobias2016}.}
and also for $w_1\to w$, since $w(1-2w+w_1)=w(1-w) - w(w-w_1)$.
However, for $w_2=w^2$ (the random instance), \Cref{thm:donnell} just says $\rho\ge 0$, which means it tells us nothing.

For the random instances, we give an even stronger lower bound, which gets rid of the restrictions $w_q=w_u$ and $\rho_q=\rho_u$.
This lower bound is tight for any $0<w_qw_u<w_1<\min\{w_q,w_u\}$ in the list-of-points model.
\begin{theorem}\label{thm:lower_main}
   Consider any list-of-point data structure for the $(w_q, w_u, w_1, w_q w_u)$-GapSS problem over a universe
   of size $d$ of $n$ points with $w_q w_u d = \omega(\log n)$, which uses expected space $n^{1 + \rho_u}$,
   has expected query time $n^{\rho_q - o_n(1)}$, 
   and succeeds with probability at least $0.99$. Then for every
   $\alpha \in [0, 1]$ we have that
   \begin{align}
        \alpha \rho_q + (1 - \alpha)\rho_u 
        \ge \inf_{\substack{
           t_q, t_u \in [0,1]
           \\ t_u \neq w_u
        }}
        \left(
            \alpha \frac{\D{T}{P}-\dq}{\du}
            + (1 - \alpha) \frac{\D{T}{P}-\du}{\du}
        \right) \;,
   \end{align}
   where
   $P = \smat{w_1 && w_q-w_1 \\ w_u-w_1 && 1-w_q-w_u+w_1}$
   and
   $T = 
        \underset{\substack{
           T \ll P
        , \underset{X\sim T}{E}[X] = \svec{t_q\\t_u}
  }}{\arginf}\D{T}{P}$.
\end{theorem}
Note that for $w_2=w_qw_u$, the term $\D{T_2}{P_2}$, in \Cref{thm:main}, splits into $\dq+\du$, and so the upper and lower bounds perfectly match.
This shows that for any linear combination of $\rho_q$ and $\rho_u$ our algorithm obtains the minimal value.
By continuity of the terms, this equivalently states as saying that no list-of-points algorithm can get a better query time than our \Cref{thm:main}, given a space budget imposed by $\rho_u$.  \footnote{
It is easy to see that $\rho_u = 0$ minimizes $\alpha \rho_q + (1 - \alpha)\rho_u$ when $\alpha = 0$,
and similarly $\rho_u = \rho_{max}$ minimizes $\alpha \rho_q + (1 - \alpha)\rho_u$ when $\alpha = 1$,
where $\rho_{max}$ is the minimal space usage when $\rho_q = 0$. Furthermore, we note that when we change $\alpha$ from 
$0$ to $1$, then $\rho_u$ will continuously and monotonically go from $0$ to $\rho_{max}$. This shows that for every
$\rho_u \in [0, \rho_{max}]$ there exists an $\alpha$ such that $\alpha \rho_q + (1 - \alpha)\rho_u$ is minimized,
where $\rho_q$ is best query time given the space budget imposed by $\rho_u$.}

\vspace{-.5em}
\paragraph{Example 1: Choices for $t_q$ and $t_u$.}
As in the upper bounds, it is not easy to prove that a particular choice of $t_q$ and $t_u$ minimizes the lower bound.
One might hope that having corresponding lower and upper bounds would help in this endeavour, but alas both results have a minimization.
E.g. setting $t_q=1-w_u$ and $t_u=1-w_q$ the expression in \Cref{thm:lower_main} we obtain the same value as in \Cref{thm:main}, however it could be (though we strongly conjecture not) that another set of values would reduce both the upper and lower bound.

The good news is that the hypercontractive inequality by Oleszkiewicz~\cite{oleszkiewicz2003nonsymmetric}, can be used to prove certain optimal choices on the space/time trade-off.\footnote{The generalizations by Wolff~\cite{wolff2007hypercontractivity} could in principle expand this range, but they are only tight up to a constant in the exponent.}
In particular we will show that for $w_q=w_u=w$ the choice $t_q=t_u=1-w$ is optimal in the lower bound, and matches exactly the value $\rho=\log\left(\frac{w_1(1-w)}{w(1-2w+w_1)}\right)/\log(\frac{w_2(1-w)}{w(1-2w+w^2)})$ from Example 1 in the Upper Bounds section.

\vspace{-.5em}
\paragraph{Example 2: Cell probe bounds}
Panigrahy et al.~\cite{panigrahy2008geometric, panigrahy2010lower, kapralov2012nns} created a framework for showing cell probe lower bounds for problems like approximate near-neighbour search and partial match based on a notion of ``robust metric expansion''.
Using the hypercontractive inequalities shown in this paper with this framework, as well (as the extension by~\cite{andoni2017optimal}), we can show, unconditionally, that no data structure, which probes only 1 or 2 memory locations\footnote{For 1 probe, the word size can be $n^{o(1)}$, whereas for the 2 probe argument, the word size can only be $o(\log n)$ for the lower bound to hold.}, can improve upon the space usage of $n^{1+\rho_u}$ obtained by \Cref{thm:main} as we let $\rho_q=0$.
In particular, this shows that the near-constant query time regime from Example 3 in the Upper Bounds is optimal up to $n^{o(1)}$ factors in time and space.

\subsection{Technical Overview}

The contributions of the paper are conceptual as well as technical.
To a large part, what enables tight upper and lower parts is defining the right problem to study.
The second part is realizing which geometry is going to work and proving it in a strong enough model.
Lastly, a number of tricky algorithmic problems arise, requiring a novel algorithm and a new analysis of 2-dimensional branching random walks of exponentially tilted variables.

\paragraph{Supermajorities -- why do they work?}

Representing sets $x\subseteq U$ as a vector $x\in\{0,1\}^{|U|}$
and scaling by $1/\sqrt{|x|}$, we get $\|x\|_2=1$, and it is natural to assume the optimal Similarity Search data structure for data on the unit-sphere --- Spherical LSF --- should be a good choice.
Unfortunately this throws away two key properties of the data: that the vectors are sparse, and that they are non-negative.
Algorithms like MinHash, which are specifically designed for this type of data, take advantage of the sparsity by entirely disregarding the remaining universe, $U$.
This is seen by the fact that adding new elements to $U$ never changes the MinHash of a set.
Meanwhile Spherical LSF takes the inner product between x and a Gaussian vector scaled down by $1/\sqrt{|U|}$, so each new element added to $U$, in a sense, lowers the ``sensitivity'' to $x$.

In an alternative situation we might imagine $|x|$ being nearly as big as $|U|$.
In this case we would clearly prefer to work with $U\setminus x$, since information about an element that is left out, is much more valuable than information about an element contained in $x$.
What Supermajorities does can be seen as balancing how much information to include from $x$ with how much to include from $U\setminus x$.
A very good example of this is in \Cref{sec:minhashdom}, which shows how to view MinHash as an average of simple algorithms that sample a specific amount from each of $x$ and $U\setminus x$.
Supermajorities, however, does this in a more clever way, that turns out to be optimal.
A crucial advantage is the knowledge of the size of $x$, as well as the future queries, which allow us to use different thresholds on the storage and query side, each which is perfectly balanced to the problem instance.

As an interesting side effect, the extra flexibility afforded by our approach allows balancing the time required to perform queries with the size of the database.
It is perhaps surprising that this simple balancing act is enough to be optimal across all hashing algorithms as well as 1 cell and 2 cell probe data structures.

\vspace{.5em}

The results turn out to be best described in terms of the KL-divergences $\D{T}{P}-\dq$ and $\D TP-\du$, which are equivalent to $\D{T_{XY}}{P_{Y|X}T_X}$ and $\D{T_{XY}}{P_{X|Y}T_Y}$.
Here $P_{XY}$ is the distribution of a coordinated sample from both a query and a dataset, $P_X$ and $P_Y$ are the marginals, and $T_{XY}$ is roughly the distribution of samples conditioned on having a shared representative set.
Intuitively these describe the amount of information gained when observing a sample from $T_{XY}$ given a belief that $X$ (resp. $Y$) is distributed as $T$ and $Y$ (resp. $X$) is distributed as $P$.
In this framework, Supermajorities can be seen as a continuation of the Entropy LSH approach by~\cite{panigrahy2006entropy}.

\paragraph{Branching Random Walks}

Making Supermajorities a real algorithm (rather than just cell probe), requires, as discussed in the introduction, an efficient decoding algorithm of which representative sets overlap with a given cohort.
Previous LSF methods can be seen as trees, with independent pruning in each leaf,
going back to the LSH forest in 2005~\cite{bawa2005lsh, andoni2017lsh}.
Our method is the first to significantly depart from this idea:
While still a tree, our pruning is highly dependent across the levels of the tree, carrying a state from the root to the leaf which needs be considered by the pruning as well as the analysis.
In ``branching random walk'', the state is represented in the ``random walk'', while the tree is what makes it branching.
While considered heuristically in~\cite{becker2016new}, such a stateful oracle has not before been analysed, partly because it wasn't necessary.
For Supermajorities, meanwhile, it is crucially important.
The reason is that failure of the ``tensoring trick'' employed previously in the literature, when working with thresholds.

The approach from~\cite{andoni2006near, becker2016new, andoni2017optimal}
when applied to our scheme would correspond to making our representatives have size just $\sqrt{k}$ (so there are only $|R'|\approx e^{\tilde O(\sqrt{\log n})}$ of them,) and then make $R'^{\otimes \sqrt{k}}$ our new $R$.
Since $R'$ can be decoded in $n^{o(1)}$ time, and the second step can be made to take only time proportional to the output, this works well for some cases.
This approach has two main issues: (1) There is a certain overhead that comes from not using the optimal filters, but only an approximation. However, this gives only a factor $e^{\tilde O(\sqrt{\log n})}$, which is usually tolerated.
Worse is (2): Since the thresholds $t_qk$ and $t_uk$ have to be integral, using representative sets of size $\sqrt{k}$ means we have to ``repair'' them by a multiplicative distortion of approximately $1\pm 1/\sqrt{k}$, compared to $1\pm1/k$ for the ``real'' filters.
This turns out to cost as much as $w_1^{-\sqrt{k}}$ which can easily be much larger than the polynomial cost in $n$.
In a sense, this shows that supermajority functions must be applied to measure the entire representative part of a cohort at once!
This makes tensoring not well fit for our purposes.

\vspace{.5em}

A pruned branching random walk on the real line can be described in the following
way.
An initial ancestor is created with value 0 and form the zeroth generation.
The people in the $i$th generation give birth $\Delta$ times each and independently of one another to form the $(i + 1)$th generation.
The people in the $(i + 1)$th generation inherit the value, $v$, of their parent plus an independent random variable $X$.
If ever $v+X<0$, the child doesn't survive.
After $k$ generations, we expect by linearity $\Delta^k \Pr[\forall_{i\le k} \sum_{j\in[i]} X_i \ge 0]$ people to be alive, where $X_i$ are iid. random variables as used in the branching.
A pruned 2d-branching random walk is simply one using values $\in\R^2$.

Branching random walks have been analysed before in the Brownian motion literature~\cite{shi2015branching}.
They are commonly analysed using the second-moment method, however, as noted by Bramson~\cite{biggins1977martingale}: ``an immediate frontal assault using moment estimates, but ignoring the branching structure of the process, will fail.''
The issue is that the probability that a given pair of paths in the branching process survives is too large for standard estimates to succeed.
If the lowest common ancestor of two nodes manages to accumulate much more wealth than expected, its children will have a much too high chance of surviving.
For this reason we have to \emph{counterintuitively add extra pruning when proving the lower bound}
that a representative set survives. More precisely, we prune all the paths that accumulate much more than
the expected value. We show that this does not lower the probability that a representative set
is favour by much, while simultaneously decreasing the variance of the branching random walk a lot.
Unfortunately, this adds further complications, since ideally, we would like to prune every
path that gets below the expectation.
Combined with the upper bound this would trap the random walks in a band to narrow to guarantee the survival of a sufficient number of paths.
Hence instead, we allow the paths to deviate by roughly a standard deviation below the expectation.

\paragraph{Exponential Tilting and Non-asymptotic Central Limit Lemmas for Random Walks}

To analyse our algorithm, we need probability bounds for events such as ``survival of $k$ generations'' that are tight up to polynomial factors.
This contrast with many typical analysis approaches in Computer Science, such as Chernoff bounds, which only need to be tight up to a constant in the exponent.
We also can't use Central Limit type estimates, since they either are asymptotic (which correspond to assuming $w_q$ and $w_u$ are constants) or too weak (such as Berry Esseen) or just don't apply to random walks.

The technical tool we employ is ``Exponential Tilting'', which allows coupling the real pruned branching random walk to one that is much more well behaved.
This can be seen as a nicer way of conditioning the random walk on succeeding.
This nicer random walk then needs to be analysed for properties such as ``probability that the path is always above the mean.''
This is shown using a rearrangement lemma, known as the Truck Driver's Lemma:
Assume a truck driver must drive between locations $l_1, l_2, \dots, l_n, l_1$.
At stop $i$ they pick up $g_i$ gas, and between stop $i$ and $i+1$ they expand $e_i$ gas.
The lemma say, that if the sum of $g_i - e_i$ is non-negative, then there is a starting position $j\in\{1,\dots,n\}$ so that the driver's gas level never goes below 0.

This lemma gives an easy proof that a random walk on $\R_+$ of $n$ identically distributed steps, must be always non-negative with probability at least $1/n$ times the probability that it is eventually non-negative.
That's because, if the location is eventually non-negative, and all arrangements of steps happen with the same probability, then we must hit the ``always non-negative'' rotation with probability $\ge 1/n$.

Extending this argument to two dimensions turns out to require a few extra conditions, such as a positive correlation between the coordinates, but as a surprisingly key result, we manage to show  \Cref{lem:rearrangement}, which says that for 
   $k \in \mathbb Z_+$ and $p, p_1, p_2 \in [0, 1]$, such that, 
   $p k, p_1 k$, and $p_2 k$ are integers and $p \ge p_1 p_2$. Let $X^{(i)} \in \set{0, 1}^2$ be independent
   identically distributed variables. We then get that
   \[
      \prpcond{\forall l \le k : \sum_{i \in [k]} X^{(i)} \ge \smat{p_1\\p_2}l}{\sum_{i \in [k]} X^{(i)} = \smat{p_1\\p_2} k \wedge \sum_{i \in [k]} X^{(i)}_1 X^{(i)}_2 = p k}
       \ge k^{-3} \; .
   \]

\paragraph{Output-sensitive set decoding}
In our algorithm we are careful to not have factors of $|U|$ and $|X|$ (the size of the sets) on our query time and space bounds.
When sampling our tree, at each level we must pick a certain number, $\Delta$, of elements from the universe and check which of them are contained in the set being decoded.
This is an issue, since $\Delta$ may be much bigger than $X\cap \Delta$, and so we need an ``output-sensitive'' sampling procedure.
We do this by substituting random sampling with a two-independent hash function $h:U^k\to [q]$, where $q$ is a prime number close to $|U|$.
The sampling criterion is then $h(r\circ x) \le \Delta$, where $\circ$ is string concatenation.
The function $h(r)$ can be taken to be $\sum_{i=1}^ka_ix_i+b\pmod q$ for random values $a_1,\dots,a_k,b\in[q]$,
so we can expand $h(r\circ x)$ as $h(r) + a_k x \pmod q$.

Now
\begin{align}
   \{x\in X\mid (h(r\circ x) \mod q) < \Delta\}
   &=
   \{x\in X\mid (h(r) + a_k x \mod q) < \Delta\}
 \\&=\cup_{i=0}^{\Delta-1} \{x\in X\mid a_k x \equiv \Delta-h(r) \mod q\}
 \\&=\{x\in X\mid (a_k x\mod q) \in [-h(r), \Delta-h(r)]\mod q\},
\end{align}
where the last equation is adjusted in case $(-h(r)\mod q) > (\Delta-h(r)\mod q)$.
By pre-computing $\{a_k x \mod q \mid x\in X\}$ (just has to be done one for each of roughly $\log n$ levels in the tree), and storing the result in a predecessor data-structure (or just sorting it), the sampling can be done it time proportional to the size of its output.

\paragraph{Lower Bounds and Hypercontractivity}
The structure of our lower bounds is by now standard:
We first reduce our lower bound to random instances by showing that with high probability the random instances are in fact an instance of our problem.
For this to work, we need $w_w \abs{U} = \omega(\log n)$ and in particular $\abs{U} = \omega(\log n)$, so we get concentration around the mean.
This requirement is indeed known to be necessary, since the results of~\cite{becker2016new, chan2017orthogonal}
break the known lower bounds in the ``medium dimension regime'' when $|U|=O(\log n)$.

The main difference compared to previous bounds is that we study Boolean functions on so-called $p$-biased spaces, where the previous lower bounds used Boolean functions on unbiased spaces.
This is necessary for us to lower bound every parameter choice for GapSS.\@
In particular we are interested in tight hypercontractive inequalities on
$p$-biased spaces. We say that a distribution $\mathcal{P}_{XY}$ on a space $\Omega_X \times \Omega_Y$ is
$(r, s)$-hypercontractive if
\[
   \ep[(X,Y) \sim \mathcal{P}_{XY}]{f(X)g(Y)}
      \le \ep[X \sim \mathcal{P}_{X}]{f(X)^r}^{1/r} \ep[Y \sim \mathcal{P}_{Y}]{g(Y)^s}^{1/s} \; ,
\]
for all functions $f : \Omega_X \to \R$ and $g : \Omega_Y \to \R$, where $\mathcal{P}_X$ and $\mathcal{P}_Y$ are the
marginal distributions on the spaces $\Omega_X$ and $\Omega_Y$ respectively.
On unbiased spaces, the classic Bonami-Beckner inequality~\cite{bonami1970etude, beckner1975inequalities} gives a complete understanding of the hypercontractivity.
Unfortunately, this is not the case for $p$-biased spaces where the hypercontractivity
is much less understood, with~\cite{oleszkiewicz2003nonsymmetric} and~\cite{wolff2007hypercontractivity} being state of the art.
We sidestep the issue of finding tight hypercontractive inequalities by instead showing an equivalence between
hypercontractivity and KL-divergence, which is captured in the following lemma:%
\footnote{It appears that one might prove a similar result using~\cite{nair2014equivalent} and~\cite{friedgut2015information}.}
\begin{lemma}\label{lem:correspondence-hypercontractive-divergence}
   Let $\mathcal{P}_{XY}$ be a probability distribution on a space $\Omega_X \times \Omega_Y$ and let $\mathcal{P}_X$ and $\mathcal{P}_Y$ be the
   marginal distributions on the spaces $\Omega_X$ and $\Omega_Y$ respectively. Let $s, r \in [1, \infty)$, then the following is equivalent
   \begin{enumerate}
      \item For all functions $f : \Omega_X \to \R$ and $g : \Omega_Y \to \R$,
      \begin{align}
         \ep[(X, Y) \sim \mathcal{P}_{XY}]{f(X) g(Y)} \le \ep[X \sim \mathcal{P}_X]{f(X)^r}^{1/r} \ep[X \sim \mathcal{P}_Y]{g(Y)^s}^{1/s} \; .
      \end{align}
      \item For all probability distributions $\mathcal{Q}_{XY}\ll\mathcal{P}_{XY}$,
      \begin{align}
         \D{\mathcal{Q}_{XY}}{\mathcal{P}_{XY}} \ge \frac{\D{\mathcal{Q}_X}{\mathcal{P}_X}}{r} + \frac{\D{\mathcal{Q}_Y}{\mathcal{P}_Y}}{s} \; ,
      \end{align}
      where $\mathcal{Q}_X$ and $\mathcal{Q}_Y$ be the
      marginal distributions on the spaces $\Omega_X$ and $\Omega_Y$ respectively
   \end{enumerate}
\end{lemma}
The main technical argument needed for proving \Cref{lem:correspondence-hypercontractive-divergence} is that,
for all probability distributions
$\mathcal{P}, \mathcal{Q}$, where $\mathcal{Q}$ is absolutely continuous with respect to $\mathcal{P}$, and all functions $\phi$,
\begin{align}
   \D{\mathcal{Q}}{\mathcal{P}} + \log \ep[X \sim \mathcal{P}]{\exp(\phi(X))} \ge \ep[X \sim \mathcal{Q}]{\phi(X)}
   .
\end{align}
This can be seen as a version of Fenchel's inequality, which says that $f(x) + f^*(p) \ge xp$ for all convex functions $f, f^*$,
where $f^*$ is convex conjugate of $f$, and all $x, p \in \R$.

We use \Cref{lem:correspondence-hypercontractive-divergence} together with the
``Two-Function Hypercontractivity Induction Theorem''~\cite{o2014analysis}, which shows that
if $\mathcal{P}_{XY}^{\otimes n}$ is $(r, s)$-hypercontractive if and only if $\mathcal{P}_{XY}$ is $(r, s)$-hypercontractive.
This implies that $\ep[(X, Y) \sim \mathcal{P}_{XY}^{\otimes n}]{f(X) g(Y)} \le \ep[X \sim \mathcal{P}_X^{\otimes n}]{f(X)^r}^{1/r} \ep[X \sim \mathcal{P}_Y^{\otimes n}]{g(Y)^s}^{1/s}$
for all functions $f, g$ if and only $\D{\mathcal{Q}_{XY}}{\mathcal{P}_{XY}} \ge \frac{\D{\mathcal{Q}_X}{\mathcal{P}_X}}{r} + \frac{\D{\mathcal{Q}_Y}{\mathcal{P}_Y}}{s}$
for all probability distributions $\mathcal{Q}_{XY}$. In the proof of \Cref{thm:lower_main} we have
$\mathcal{P}_{XY} = \left[\begin{smallmatrix}w_1 & w_q-w_1\\w_u-w_1 & 1-w_q-w_u+w_1\end{smallmatrix}\right]$
and consider all the probability distributions of the form $\mathcal{Q}_{XY} =\underset{\substack{
   \mathcal{Q}_{XY} \ll \mathcal{P}_{XY}
, \underset{X\sim \mathcal{Q}_{XY}}{E}[X] = \svec{t_q\\t_u}
}}{\arginf}\D{\mathcal{Q}_{XY}}{\mathcal{P}_{XY}}$ for $t_q, t_u \in [0, 1]$.

The obtained inequalities can be used directly with the framework by Panigrahy et al.~\cite{panigrahy2008geometric} to obtain bounds on ``Robust Expansion'', which has been shown to give lower bounds for 1-cell and 2-cell probe data structures, with word size $n^{o(1)}$ and $o(\log n)$ respectively.

\paragraph{The Directed Noise Operator}
We extend the range of our lower bounds further, by studying a recently defined generalization of the $p$-biased noise operator~\cite{ahlberg2014noise, abdullah2015directed, lifshitz2018hypergraph, keevash2019hypercontractivity}.
This ``Directed Noise Operator'', $T_\rho^{p_1 \to p_2} : L_2(\set{0, 1}^d, \pi_{p_1}^{\otimes d}) \to L_2(\set{0, 1}^d, \pi_{p_2}^{\otimes d})$ has the property
$   \widehat{T_{\rho}^{p_1 \to p_2} f}^{(p_2)}(S) = \rho^{\abs{S}}\hat{f}^{(p_1)}(S) $
for any $S \subseteq [d]$, where $\hat{f}^{(p)}(S)$ denotes the $p$-biased Fourier coefficient of $f$.
Just like the Ornstein Uhlenbeck operator, we show that $T_{\sigma}^{p_2 \to p_3} T_{\rho}^{p_1 \to p_2} = T_{\rho \sigma}^{p_1 \to p_3}$ and that $T_{\rho}^{p_2 \to p_1}$ is the adjoint of $T_{\rho}^{p_1 \to p_2}$.
By connecting this operator to our hypercontractive theorem, we can integrate the results by Oleszkiewicz and obtain provably optimal points on the $(t_q,t_u)$ trade-off.

We show that for $p$-biased distributions over $\{0,1\}^n$, we can add the following line to the list of equivalent statements in \Cref{lem:correspondence-hypercontractive-divergence}:
\begin{enumerate}
   \setcounter{enumi}{2}
   \item For all functions $f:\{0,1\}^n\to\R$ it holds
      $\|T^{p_1\to p_2}_\rho f\|_{L_{s'}(p_1)} \le \|f\|_{L_r(p_2)}$.
\end{enumerate}
The operator allows us to prove some optimal choices for $r$ and $s$ in \Cref{lem:correspondence-hypercontractive-divergence} (and by effect for $t_q$ and $t_u$.)
Following~\cite{ahlberg2014noise} we use Pareseval's identity, to write
    $\norm{T^{p_1\to p_2}_{\rho} f}{L_2(p_2)}^2$ as
\begin{align}
       \widehat{T^{p_1\to p_2}_{\rho} f}^{(p_2)}(\emptyset)^2 + \widehat{T^{p_1\to p_2}_{\rho} f}^{(p_2)}(\set{1})^2
       = \widehat{f}^{(p_1)}(\emptyset)^2 + \rho^2 \widehat{f}^{(p_1)}(\set{1})^2
       = \norm{T^{p_1\to p_1} f}{L_2(p_1)}^2
       \le \|f\|_{L_r(p_1)}^2
       \; ,
 \end{align}
where $r$ is perfectly determined by Oleszkiewicz in~\cite{oleszkiewicz2003nonsymmetric}.
It is possible to prove further lower bounds using H\"older's inequality on $T$, however the bounds obtained this way turn out to be optimal only in the case $s=2$ or $r=2$ that also follow from Parseval.
A particular simple case is $r = s = $, $w_q=w_u=w$, and $w_2=w^2$, in which case the arguments above gives the lower bound $\rho\ge \log(\frac{w_1(1 - w)}{w(1 - 2w + w_1)})/\log(\frac{1 - w}{w})$ mentioned in Example 1 in the Upper Bounds section.

Another use of $T$ is in proving lower bounds outside of the random instance $w_2=w_qw_u$ regime.
Using the power means inequality over $p$-biased Fourier coefficients,
we show the relation
\begin{align}
   \left( \langle T^{p\to p}_\alpha f, f \rangle_{L_2(p)}/\norm{f}{L_2(p)}^2 \right)^{1/\log(1/\alpha)} 
      \le \left( \langle T^{p\to p}_\beta f, f \rangle_{L_2(p)}/\norm{f}{L_2(p)}^2 \right)^{1/\log(1/\beta)} \; .
\end{align}
which is allows comparing functions under two different noise levels.
This is stronger than hypercontractivity, even though we can prove it in fewer instances.
The proof can been seen as a variation of~\cite{o2014optimal} and we get a lower bound with a similar range, but without asymptotics and for Set Similarity instead of Hamming space Similarity Search.

\subsection{Related Work}\label{sec:related}

\def\arraystretch{1.5}
\begin{table}[p!]
   \center
   \begin{tabular}{|l|c|c|}
      \hline
      \bfseries Method
      &
      \bfseries Balanced $\rho_q=\rho_u$
      &
      \bfseries Space/time trade-offs
      \\ \hline \hline
      \multirow{2}{*}{
         \shortstack[l]{
            Spherical LSF\\
            \cite{terasawa2007spherical, laarhoven2015tradeoffs, christiani2016framework, andoni2017optimal}
      }}
      &
      \multirow{2}{*}{
         $\displaystyle\frac{1-\alpha}{1+\alpha}\frac{1+\beta}{1-\beta}$
      }
      &
      $\rho_q = \frac{(1-\alpha^{1+\lambda})^2}{1-\alpha^2}\frac{1-\beta^2}{(1-\alpha^\lambda\beta)^2}$$^{(***)}$
      \\&&
      $\rho_u = \frac{(1-\alpha^{1+\lambda})^2}{1-\alpha^2}\frac{1-\beta^2}{(1-\alpha^\lambda\beta)^2}$
      \\\hline
      \multirow{2}{*}{
         \shortstack[l]{
            MinHash
            \cite{broder1997syntactic}
         }
      }
      &
      \multirow{2}{*}{
         $\displaystyle\frac{\log\frac{w_1}{w_q+w_u-w_1}}{\log\frac{w_2}{w_q+w_u-w_2}}$
      }
      &
      \multirow{2}{*}{
         \shortstack[l]{
            Same as above$^{(*)}$ with\\
         $\alpha=\frac{w_1}{w_q+w_u-w_1}, \beta=\frac{w_2}{w_q+w_u-w_2}$
      }}
      \\ && \\\hline
      \multirow{2}{*}{
         \shortstack[l]{
            Chosen Path
            \cite{tobias2016}
         }
      }
      &
      \multirow{2}{*}{
         $\displaystyle\frac{\log\frac{w_1}{\max\{w_q,w_u\}}}{\log\frac{w_2}{\max\{w_q,w_u\}}}$
      }
      & \multirow{2}{*}{N/A}
       \\&& \\\hline
      \multirow{2}{*}{
         \shortstack[l]{
            \bfseries Supermajorities \\[.2em] (This paper)
         }
      }
      & \multirow{2}{*}{\shortstack{\Cref{thm:main},\\[.2em] Example 1}}
      & \multirow{2}{*}{\Cref{thm:main} }
      \\&&\\\hline\hline
      \multirow{2}{*}{
         \shortstack[l]{
            Data-Dependent LSF\\
            \cite{andoni2015optimal, andoni2017optimal}
      }}
      &
      \multirow{2}{*}{
         $\displaystyle\frac{1-\alpha}{1+\alpha-2\beta}$
      }
      &
      \multirow{2}{*}{
         \shortstack[l]{
            $\sqrt{\rho_q}+\alpha'\sqrt{\rho_u}=\sqrt{1-\alpha'^2}$\\
            where $\alpha' = 1-\frac{1-\alpha}{1-\beta}$
      }}
      \\&&\\\hline
      \multirow{2}{*}{
         \shortstack[l]{
            SimHash
            \cite{charikar2002similarity}
         }
      }
      &
      \multirow{2}{*}{
         $\displaystyle\frac{\log(1-\arccos(\alpha)/\pi)}{\log(1-\arccos(\alpha)/\pi)}$
      }
      & \multirow{2}{*}{N/A$^{(**)}$}
    \\&&\\\hline
      \multirow{2}{*}{
         \shortstack[l]{
            Bit Sampling
            \cite{indyk1998approximate}
         }
      }
      &
      \multirow{2}{*}{
         $\displaystyle\frac{\log(1-w_q-w_u+2w_1)}{\log(1-w_q-w_u+2w_2)}$
      }
      & \multirow{2}{*}{N/A$^{(**)}$}
    \\&&\\\hline
   \end{tabular}
   \caption{
      Time and space exponents for the best similarity search data-structures.
      For Spherical LSF and SimHash, $\alpha$ and $\beta$ are the inner products between sets represented as vectors, and can by \Cref{lem:embedding} be taken to be
      $\alpha=\frac{w_1-w_qw_u}{\sqrt{w_q(1-w_q)w_u(1-w_u)}}$ and
      $\beta=\frac{w_2-w_qw_u}{\sqrt{w_q(1-w_q)w_u(1-w_u)}}$.\\
      $(*)$:~Space/time trade-offs for MinHash can be obtained using MinHash as an embedding for Spherical LSF.
      ${(**)}$:~Some space/time trade-offs can be obtained for LSH using Multi-probing~\cite{lv2007multi}.
      $(*\!*\!*)$:~$\lambda\in[-1,1]$ controls the space/time trade-off.
   }
   \label{table:related}
\end{table}

For the reasons laid out in the introduction, we will compare primarily against approximate solutions.
The best of those are all able to solve GapSS, thus making it easy to draw comparisons.
The guarantees of these algorithms are listed in \Cref{table:related} and we provide plots in \Cref{fig:sp_comparison} and \Cref{fig:mh_comparison} for concreteness.

The methods known as Bit Sampling~\cite{indyk1998approximate} and SimHash (Hyperplane rounding)~\cite{charikar2002similarity},
while sometimes better than MinHash\cite{broder1997syntactic} and Chosen Path~\cite{tobias2016}
are always worse (theoretically) that Spherical LSF, so we won't perform a direct comparison to those.

It should be noted that both Chosen Path and Spherical LSF both have proofs of optimality in the restricted models.
However these proofs translated to only a certain region of the $(w_q,w_u,w_1,w_2)$ space, and so they may nearly always be improved.

Arguably the largest break-through in Locality Sensitive Hashing, LSH, based data structures was the introduction of \emph{data-dependent} LSH~\cite{andoni2014beyond, andoni2015optimal, andoni2017lsh}.
It was shown how to reduce the general case of $\alpha, \beta$ similarity search as described above, to the case $(\alpha, \beta) \mapsto (\frac{\alpha-\beta}{1-\beta}, 0)$, in which many LSH schemes work better.
Using those data structures on GapSS with $w_2>w_qw_u$ will often yield better performance than the algorithms described in this paper.
However, in the ``random instance'' case $w_2=w_qw_u$, which is the main focus of this paper, data-dependency has no effect, and so this issue won't show up much in our comparisons.

We note that even without a reduction to the random instance, for many practical uses, it is natural to assume such ``independence'' between the query and most of the dataset.
Arguably this is the main reason why approximate similarity search algorithms have gained popularity in the first place.
In practice, some algorithms for Set Similarity Search take special care to handle ``skew'' data distributions~\cite{rashtchian2020lsf, zhang2017efficient,mccauley2018set}, in which some elements of the Universe are heavily over or under-represented.
By special casing those elements, those algorithms can be seen as reducing the remaining dataset to the random instance.
Curiously, even the early research on Partial Match by Ronald Rivest in his PhD thesis~\cite{rivest1976partial}, studied the problem on random data.

Many of the algorithms, based on the LSH framework, all had space usage roughly $n^{1+\rho}$ and query time $n^\rho$ for the same constant $\rho$.
This is known as the ``balanced regime'' or the ``LSH regime''.
Time/space trade-offs are important, since $n^{1+\rho}$ can sometimes be too much space, even for relatively small $\rho$.
Early work on this was done by Panigrahy~\cite{panigrahy2006entropy} and Kapralov~\cite{kapralov2015smooth} who gave smooth trade-offs ranging from space $n^{1+o(1)}$ to query time $n^{o(1)}$.
A breakthrough was the use of LSF (rather than LSH), which allowed time/space trade-offs with sublinear query time even for near linear space and small approximation~\cite{laarhoven2015tradeoffs, christiani2016framework, andoni2016optimal}.

We finally compare our results to the classical literature on Partial Match and Super-/Subset search, which has some intriguing parallels to the work presented here.

\paragraph{Comparison to Spherical LSF}
\begin{figure}
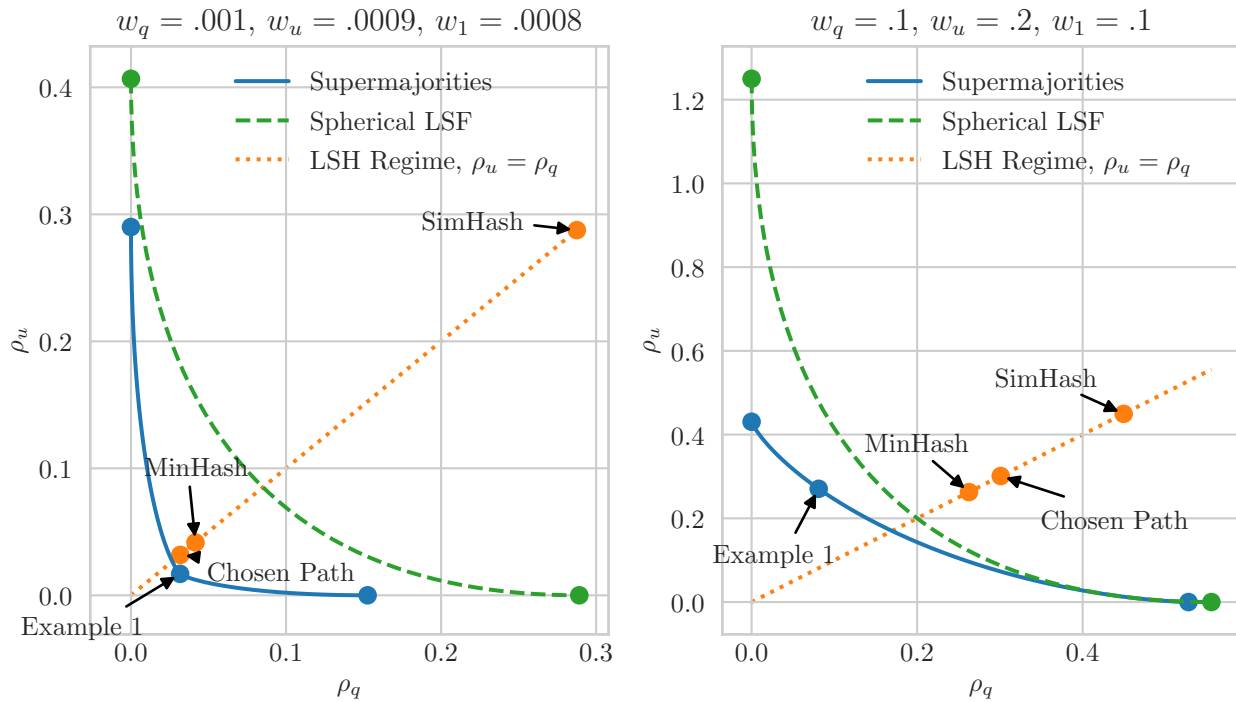

   \vspace*{-1cm}
   \centering
     \begin{subfigure}[t]{.49\textwidth}
      \vspace*{0pt}
      \centering\input{figures/sp1.pgf}
      \caption{Example of search with very small sets.}
      \label{fig:sp_comparison_1}
   \end{subfigure}\hspace{0.019\textwidth}%
   \begin{subfigure}[t]{.49\textwidth}
      \vspace*{0pt}
      \centering\input{figures/sp4.pgf}
      \caption{Example with larger sets of different sizes.}
      \label{fig:sp_comparison_2}
   \end{subfigure}
   \caption{Comparison to Spherical LSF: Plots of the achievable $\rho_q$ (time exponent) and $\rho_u$ (space exponent) achievable with \Cref{thm:main}.
      Note that using our optimal spherical embedding from Lemma~\ref{lem:embedding} is critical to achieve the exponents shown for Spherical LSF\@.
      The plots are drawn in the ``random setting'', $w_2=w_qw_u$ where Spherical LSF and Data-Dependent LSH coincide.
   }
   \label{fig:sp_comparison}
\end{figure}

We use ``Spherical LSF'' as a term for the algorithms~\cite{becker2016new} and~\cite{laarhoven2015tradeoffs}, but in particular section 3 of~\cite{andoni2017optimal}, which has the most recent version.
The algorithm solves the $(r,cr)$-Approximate Near Neighbour problem, in which we, given a dataset $Y\subseteq\R^d$ and a query $q\in\R^d$ must return $y\in Y$ such that $\|q-y\|< cr$ or determinate that there is no $y'\in Y$ with $\|y-q\|\le r$.

The algorithm is a tree over the points, $P$.
At each node they sample $T$ i.i.d.~Gaussian $d$-dimensional vectors $z_1,\dots,z_T$ and  split the dataset up into (not necessarily disjoint) ``caps'' $P_i = \{p\in P \mid \langle z_i,p\rangle \ge t_u\}$.
They continue recursively and independently until the expected number of leaves shared between two points at distance $\ge cr$ is $\approx n^{-1+\eps}$.

The real algorithm also samples includes some caps that are dependent on an analysis of the dataset.
This allows obtaining a query time of $n^{1/(2c^2-1)}$, for all values of $r$, rather than only in the ``random instance'', which, for data on the sphere, corresponds to $r=1/(\sqrt{2}c)$.
(To see this, notice that $rc=1/\sqrt 2$, which is the expected distance between two orthogonal points on a sphere.)

Whether we analyse the data-independent algorithm or not, however, a key property of Spherical LSF is that each node in the tree is independent of the remaining nodes.
This allows a nice inductive analysis.
In comparison, in our algorithm, the nodes are not independent.
Whether a certain node gets pruned, depends on which elements from the universe were sampled at all the previous nodes along the path from the root.
One could imagine doing Spherical LSF with a running total of inner products along each path, which would make the space partition more smooth, and possible better in practice.
Something along these lines was indeed suggested in~\cite{becker2016new}, however it wasn't analysed, as for Spherical LSF \emph{the inner products at each node are continuous, and the thresholds can be set at any precision.}

\vspace{.5em}

It is clear that Spherical LSF can solve GapSS -- one simply needs an embedding of the sets onto the sphere.
An obvious choice is $x\mapsto x/\|x\|_2$.
This was used in~\cite{tobias2016} when comparing Chosen Path to Spherical LSF\@.
However it is also clear that the choice of embedding matters on the performance one gets out of Spherical LSF\@.
Other authors have considered $x\mapsto (2x-1)/\sqrt{d}$ and various asymmetric embeddings~\cite{shrivastava2014asymmetric}.

We would like to find the most efficient embedding to get a fair comparison.
However, we don't know how to do this optimally over all possible embeddings, which include using MinHash and possibly somehow emulating Supermajorities.\footnote{We would also need some sort of limit on how much time the embedding takes to perform.}
We instead find the most efficient \emph{affine} embedding, which turns out to be surprisingly simple, and which encompasses all previously suggested approaches.
In Lemma~\ref{sec:embed} we prove a general result, implying that the embedding is optimal for Spherical LSF as well as other spherical data structures like SimHash.
In~\Cref{fig:sp_comparison} and \Cref{fig:mh_comparison} the $\rho$-values of Spherical LSF are obtained using this optimal embedding.

From the figures, we see the two main cases in which Spherical LSF is suboptimal.
As the sets get very small ($w_q,w_u,w_1\to 0$) the $\rho$ value in the LSH regime goes to 1, whereas Supermajorities (as well as MinHash and Chosen Path) still obtain good performance.
Similarly in the asymmetric case $w_q\neq w_u$, as we make $\rho_q$ very small, the performance gap between Supermajorities and Spherical LSF can grow to arbitrarily large polynomial factors.

\paragraph{Comparison to MinHash}

\begin{figure}
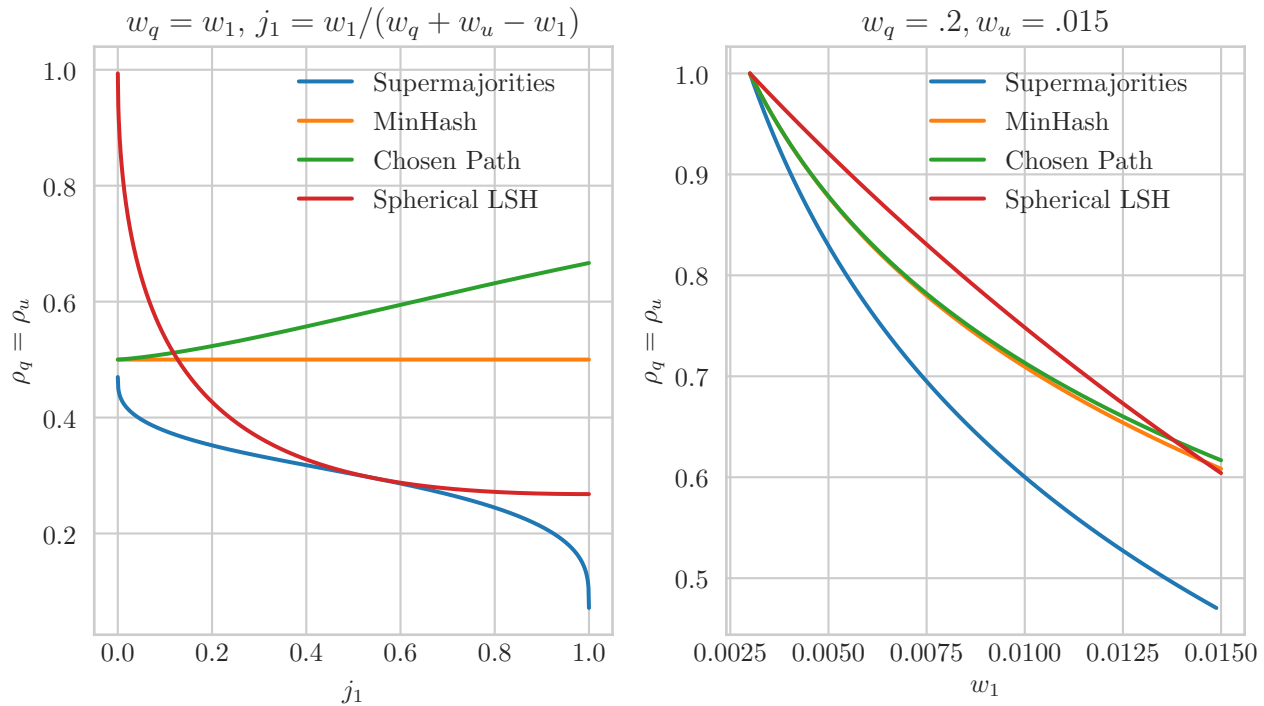

      \vspace*{-1cm}
      \centering
     \begin{subfigure}[t]{.49\textwidth}
      \vspace*{0pt}
      \centering\input{figures/mh.pgf}
     \caption{
        The plot shows the effect on the balanced exponent of Supermajorities as we the Jaccard similarity while fixing the exponent of MinHash at $\rho=.5$.
     }
     \label{fig:mh_comparison_1}
   \end{subfigure}\hspace{0.019\textwidth}%
   \begin{subfigure}[t]{.49\textwidth}
      \vspace*{0pt}
      \centering\input{figures/mh2.pgf}
     \caption{
        Varying the overlap $w_1$ among close sets while fixing the query and database set sizes.
        At $w_1=.002$ there is no gap between close and far sets.
      }
      \label{fig:mh_comparison_2}
   \end{subfigure}
   \caption{Comparison to MinHash: Varying different parameters while searching on a background of random sets $(w_2=w_qw_u)$, Supermajorities regularly get substantially better time and space exponents.
      The plots are drawn in the ``random setting'', $w_2=w_qw_u$ and use the optimal embedding for Spherical LSF.
   }
   \label{fig:mh_comparison}
\end{figure}

Given a random function $h:\mathcal P(\{1,\dots,d\})\to[0,1]$, the MinHash algorithm hashes a set $x\subseteq\{1,\dots,d\}$ to $m_h(x) = \argmin_{i\in x}h(i)$.
One can show that $Pr[m_h(x) = m_h(y)] = J(x, y) = \frac{|x\cap y|}{|x\cup y|}$.
Using the LSH framework by Indyk and Motwani~\cite{indyk1998approximate} this yields a data structure for Approximate Set Similarity Search over Jaccard similarity, $J$, with query time $dn^{\rho}$ and space usage $n^{1+\rho}+dn$, where $\rho=\frac{\log j_1}{\log j_2}$ and $j_1$ and $j_2$ define the gap between ``good'' and ``bad'' search results.
As Jaccard similarity is a set similarity measures, it is clear that MinHash yields a solution to the GapSS problem
with $\rho_q=\rho_u = \log\frac{w_1}{w_q+w_u-w_1}\big/\log\frac{w_2}{w_q+w_u-w_2}$.
Similarly, and that any solution to GapSS can yield a solution to Approximate SSS over Jaccard similarity.

MinHash has been very popular, since it gives a good, all-round algorithm for Set Similarity Search, that is easy to implement.
In \Cref{fig:mh_comparison} we see how MinHash performant for different settings of GapSS.
In particular we see that when solving the Superset Search problem, which is a common use case for MinHash, our new algorithm obtains quite a large polynomial improvement, except when the Jaccard similarity between the query and the sought after superset is nearly 0 (which is hardly an interesting situation.)

\vspace{.5em}

It is possible to use MinHash as an embedding (or densification) of sets into Hamming space or onto the Sphere.
We can then use Spherical LSF to get space/time trade-offs.
We have not plotted those, but we can notice that in the balanced case, $\rho_q=\rho_u$, this would give $\rho=\frac{1-j_1}{1+j_1}\frac{1+j_2}{1-j_2}$, which is worse than $\rho=\log j_1/\log j_2$ obtained by the direct algorithm.

MinHash is quite different from the other algorithms considered in this section.
For some more intuition of why MinHash is not optimal for Approximate Set Similarity Search, we show in \Cref{sec:minhashdom} that MinHash can be seen as an average of a family of Chosen Path like algorithms.
We also show that an average is always worse than simply using the best family member, which implies that MinHash is never optimal.

\paragraph{Comparison to Chosen Path}
\begin{figure}
   \vspace*{-1cm}
   \centering
     \begin{subfigure}[t]{.49\textwidth}
      \vspace*{0pt}
      \centering\input{figures/cp1.pgf}
      \caption{In the plot, Chosen Path never matches Supermajorities, even at $w_q=w_u$ since the sets are relatively large.}
      \label{fig:cp_comparison_1}
   \end{subfigure}\hspace{0.019\textwidth}%
   \begin{subfigure}[t]{.49\textwidth}
      \vspace*{0pt}
      \centering\input{figures/cp3.pgf}
      \caption{In the plot,
      Chosen Path nearly matches Supermajorities when $w_u=w_q$ as the sets are relatively small.}
      \label{fig:cp_comparison_2}
   \end{subfigure}
   \caption{Comparison to Chosen Path: Fixing $w_q$ and the Jaccard similarity so $w_1=\frac{j_1}{1+j_1}(w_q+w_u)$, we vary $w_u$ to see the performance of different algorithms at different levels of asymmetry in the set sizes.
      The plots are drawn in the ``random setting'', $w_2=w_qw_u$ and use the optimal embedding for Spherical LSF.
   }
   \label{fig:cp_comparison}
\end{figure}

The Chosen Path algorithm of~\cite{tobias2016}, is virtually identical to Supermajorities, when parametrized with $t_q=t_u=1$.
Similar to Spherical LSF and our decoding algorithm, they build a tree on the datasets.
For each node they sample iid. Elements $x_1, x_2, \dots \in U$ from the universe, and split the data into (not necessarily disjoint) subsets $P_i = \{p\in P \mid x_i \in p \}$.
They again continue recursively and independently until the expected number of leaves shared between two dissimilar points is sufficiently small.

The case $t_q=t_u=1$ however, turns out to be a very special case of our algorithm, because one can decide which leaves of the tree to prune, without knowledge of what happened previously on the path from the root to the node.
This allows a nice inductive analysis of Chosen Path based on second moments, which is a classic example literature on branching processes.
Meanwhile, for our general algorithm, we need to analyse the resulting branching random walk, a conceptually much different beast.

Doing the analysis, one gets a data structure for Approximate Set Similarity Search over Braun-Blanquet similarity, $B(x,y)=\frac{|x\cap y|}{\max\{|x|,|y|\}}$, with query time $|q|n^{\rho}$ and auxiliary space usage $n^{1+\rho}$, where $\rho=\frac{\log b_1}{\log b_2}$ and $b_1$ and $b_2$ define the gap between ``good'' and ``bad'' search results.
Since $t_q=t_u=1$ is sometimes the optimal choice for Supermajorities, it is clear that we must sometimes coincide in performance with Chosen Path.
In particular, this happens as $w_q=w_u$ and $w_q,w_u,w_1\to 0$.
This is also one of the case where our lower bound~\Cref{thm:donnell} is sharp, which confirms, in addition to the lower bound in~\cite{tobias2016} that both algorithms are sharp for LSF data structures in this setting.
\Cref{fig:sp_comparison_1} shows how Chosen Path does nearly as well as Supermajorities on very small sets.

In the case $w_q=w_u$ the $\rho$ value of Chosen Path can be equivalently written in terms of Jaccard similarities as $\log\frac{2j_1}{1+j_1}\big/\log\frac{2j_2}{1+j_2}$, which is always smaller than the $\log j_1\big/\log j_2$ obtained by MinHash.
(This value, $2j/(1+j)$, is also known as the Sørensen-Dice coefficient of two sets.)
However, in the case $w_q\neq w_u$ Chosen Path can be much worse than MinHash, as seen in \Cref{fig:sp_comparison_2} and \Cref{fig:mh_comparison_1}.
In~\cite{tobias2016} it was left as an open problem whether MinHash could be improved upon in general.
It is a nice result that the balanced $\rho$ value of Supermajorities (when $\rho_q=\rho_u$) can be shown (numerically) to always be less than or equal to $\log\frac{2j_1}{1+j_1}\big/\log\frac{2j_2}{1+j_2}$, even when $w_q\neq w_u$.
It is a curious problem for which similarity measure, $S$, so the balanced $\rho$ value of Supermajorities equal $\log s_1/\log s_2$.

\paragraph{Partial Match (PM) and Super-/Subset queries (SQ)}
Partial Match asks to pre-process a database $D$ of $n$ points in $\{0,1\}^d$
such that, for all query of the form $q \in \{0, 1, \ast\}^d$, either report a point $x \in D$ matching all non-$\ast$ characters in $q$ or report that no such $x$ exists.
A related problem is Super-/Subset queries, 
in which queries are on the form
$q \in \{0, 1\}^d$, and we must either report a point $x \in D$ such that $x\subseteq q$ (resp. $q\subseteq x$) or report that no such $x$ exists.

The problems are equivalent to the subset query problem by the following folklore reductions:
(PM $\to$ SQ) Replace each $x \in D$ by the set $\{(i,p_i) : i \in [d]\}$.
Then replace each query $q$ by $\{(i,q_i) : q_i = \ast\}$.
(SQ $\to$ PM) Keep the sets in the database as vectors and
replace in each query each $0$ by an $\ast$.

The classic approach, studied by Rivest~\cite{rivest1976partial}, is to split up database strings like \texttt{supermajority}
and file them under \texttt{s}, \texttt{u}, \texttt{p} etc.
Then when given query like \texttt{set} we take the intersection of the lists \texttt{s}, \texttt{e}, \texttt{t}.
Sometimes this can be done faster than brute force searching each list.
He also considered the space heavy solution of storing all subsets, and showed
that when $d \le 2\log n$, the trivial space bound of $2^d$ can be somewhat improved.
Rivest finally studied approaches based on tries and in particular the case where most of the database was random strings.
The latter case is in some ways similar to the LSH based methods we will describe below.

\begin{figure}
   \vspace{-1cm}
   \centering
    \includegraphics[scale=0.3]{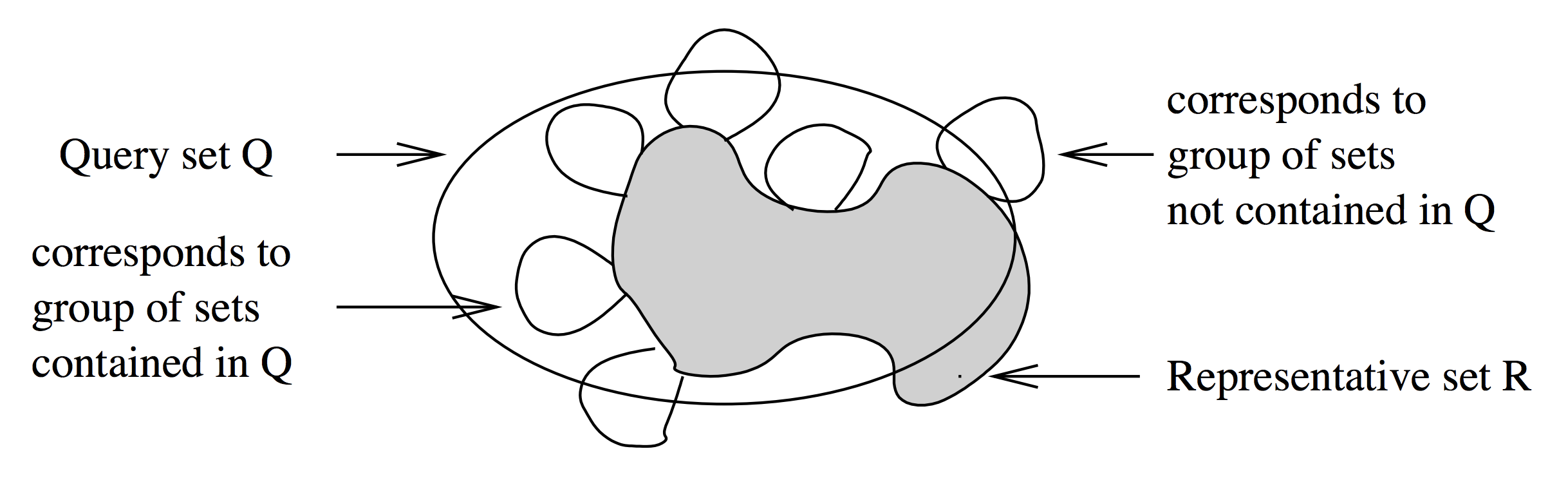}
   \caption{
      This figure from the Partial Match algorithm of~\cite{charikar2002new} shares some of the same geometrical intuition visible in our own figure~\ref{fig:cohorts}.
   }
   \label{fig:partial-match}
\end{figure}

Indyk, Charikar and Panigrahy~\cite{charikar2002new} also studied the exact version of the problem, and gave, for each $c\in[n]$, an algorithm with
$O(n/2^c)$ time and $n 2^{(O(d\log^2 d\sqrt{c/\log n})}$ space,
and another with $O(dn/c)$ query time and $nd^c$ space.
Their approach was a mix between the shingling method of Rivest,
building a look-up table of size $\approx 2^{\Omega(d)}$,
and a brute force search.
These bounds manage to be non-trivial for $d=\omega(\log n)$, however only slightly. (e.g. $n/\poly(\log n)$ time with polynomial space.)

There has also been a large number of practical papers written on Partial Match / Subset queries or the equivalent batch problem of subset joins~\cite{Ramasamy2000SetCJ, melnik2003adaptive, goel2010small, agrawal2010indexing, lazo2019}.
Most of these use similar methods to the above, but save time and space in various places by using bloom filters and sketches such as MinHash~\cite{broder1997syntactic} and HyperLogLog~\cite{flajolet2007hyperloglog}.

\paragraph{Maximum Inner Product}
(MIPS) is the Similarity Search problem with $S(x,y)=\langle x,y\rangle$ --- the Euclidean inner product.
For exact algorithms, most work has been done in the batch version ($n$ data points, $n$ queries).
Here Alman et al.~\cite{alman2016polynomial} gave an $n^{2-1/\tilde O(\sqrt k)}$ algorithm, when $d=k \log n$.

An approximative version can be defined as:
Given $c>1$, pre-process a database $D$ of $n$ points in $\{0,1\}^d$ such that, for all query of the form $q \in \{0, 1\}^d$ return a point $x \in D$ such that $\langle q,x\rangle \ge \frac{1}{c} \max_{x'\in D} \langle q,x'\rangle$.
Here~\cite{ahle2015complexity} gives a data structure with query time $\approx \tilde O(n/c^2)$, and \cite{chen2019equivalence} solves the batch problem in time $n^{2-1/O(\log c)}$
(both when $d$ is $n^{o(1)}$.)

There are a large number of practical papers on this problem as well.
Many are based on the Locality Sensitive Hashing framework (discussed below) and have names such as
SIMPLE-LSH~\cite{neyshabur2015symmetric} and L2-ALSH~\cite{shrivastava2014asymmetric}.
The main problem for these algorithms is usually that no hash family of functions $h : \{0,1\}^d \times \{0,1\}^d \to [m]$ such that $\Pr[h(q)=h(x)]=\langle q,x\rangle/d$ ~\cite{ahle2015complexity} and various embeddings and asymmetries are suggested as solutions.

The state of the art is a paper from NeurIPS 2018~\cite{yan2018norm} which suggests partitioning data by the vector norm, such that the inner product can be more easily estimated by LSH-able similarities such as Jaccard.
This is curiously very similar to what we suggest in this paper.

We will not discuss these approaches further since, for GapSS, they all have higher exponents than the three LSH approaches we study next.

\section{The Algorithm}

\label{sec:upper}

We now describe the full algorithm that gives \Cref{thm:main}.
We state the full version of the theorem, discuss it and prove it.
The section ends with an involved analysis of the survival probabilities of the branching random walk.

Notationally we define $[n]=\{1,\dots,n\}$ and let $(\cdot\circ\cdot):A^{l_1} \times A^{l_2} \to A^{l_1+l_2}$ be the concatenation operator for any set $A$ and integers $l_1,l_2$.
We will use the Iversonian bracket, defined by $[P]=1$ if $P$ and $0$ otherwise.
For $R$ and $U$ sets, we have $R\times U = \{r\circ u \mid r\in R, u\in U\}$
$\mathcal P(U)$ is the power set of $U$.

\vspace{.5em}

The first step is to set up our assumptions.
For $w_q, w_u, w_1, w_2, t_q, t_u \in [0, 1]$ given, we can assume $\min\{w_q,w_u\}\ge w_1 > w_2$ and $t_q\neq w_q, t_u\neq w_u$.
We are also given a universe $U$ and a family $Y\subseteq\binom{U}{w_u|U|}$ of size $|Y|=n$.

It will be nice to assume $|U|=q$ where $q$ is a prime number.
This can always be achieved by adding at most $|U|^{0.525}$ elements to $U$ large enough\footnote{It is an open conjecture by Harald Cramér that $(\log |U|)^2$ suffices as well.~\cite{cramer1936order}}~\cite{baker2001difference}.
Hence we only distort each of $w_q, w_u, w_1, w_2$ by roughly a factor $1+O(|U|^{-1/2})$, which is insignificant for $|U|=\Omega(\log n)^2$, and we can always increase $|U|$ without changing the problem parameters by duplicating the set elements.

Let $k\in \mathbb Z_+$ be defined later.
For all $i\in[k]$ we define $h_i(r) : [q]^i \to [q]$ by $h_i(r) = \sum_{j\in[i]} a_{i,j}r_j + b_i \mod q$ for some sequences of random numbers $a_{i,j} \in [q]\setminus\{0\}, b_i\in[q]$,
such that each $h_i$ is a 2-independent random function.
(That means $\Pr[h_i(r)=h_i(r')]\le 1/q$ for $r\neq r'$.)

Finally two sequences $(\Delta_i\in\mathbb Z_+)_{i\in[k]}$ and $(c_\ell\in\R^2)_{\ell\in[k]}$ to be specified later.
We can now define the sets
$R_i = \setbuilder{r \circ x \in R_{i - 1} \times U}{h_i(r \circ x) < \Delta_i},$
as well as the decoding functions
\begin{align}
   \mathcal R_i(X, t) = \bigg\{ r \in R_i \biggm|
      \forall \ell \le i : \sum_{j \in [\ell]} [r_j \in X] \ge t\ell - c_\ell
   \bigg\}
\end{align}
Intuitively $R_i$ are our representative sets at level $i$ in the tree, such that $R_k$ is a close to iid.\ uniform sample from $U^k$.
The decoding function takes a set $X\subseteq U$ and a value $t\in[0,1]$, and returns all $r \in R_i$ such that all prefixes $r'$ of $r$ ``$(t-\eps)$-favours'' $X$ (as defined by $|r\cap X|/|r|\ge t-\eps$ in the introduction), where $\eps=\frac{c_{|r'|}}{|r'|}$ is some slack that helps ensure survival of at least one representative set.
The slack won't be the same on each coordinate, but scaled by their variance.
The algorithm is shown below as pseudo-code in \Cref{alg:treesample}.

\begin{algorithm}
      \caption{Pseudocode for the decoding function $\mathcal R$.}
      \label{alg:treesample}
      \SetKwInOut{Input}{Input}
      \Input{Universe $U$, Set $X\subseteq U$, Threshold $t\in[0,1]$}
      \KwResult{Set $P_k\subseteq U^{k}$ of paths}
      $R_0 \leftarrow \{((), 0)\}$
      \tcp*[f]{These $R_i$ values contain the paths and scores}\\
      \For{$i=1$ \KwTo $k$}{
         $R_i \leftarrow \{\}$\\
         \For{$(r, s) \in R_{i-1}$}{
            \For(\tcp*[f]{Sample the universe})
               {$x \in U$ st. $h_i(r\circ x) < \Delta_i$}{
               $s' \leftarrow s + [x \in X]$\\
               \If(\tcp*[f]{Trim to promising paths})
                  {$s' \ge it - c_i$}{
                  $R_i \leftarrow R_i\cup \{(r \circ x, s')\}$\\
               }
            }
         }
      }
   \end{algorithm}

Our data structure now builds a hash-table $M$ of lists of pointers and store each set $y\in Y$ in $M[r]$ for every $r\in R_k(y,t_u)$.
One can think of this as storing the elements at the leafs of the tree represented by the sets $R_i$.
On a new query $q\in\binom{U}{w_q|U|}$ we look at every list $M[r]$ for $r\in R_k(q,t_q)$.
For each $y$ in such a list, we compute the intersection with $q$ and return $y$ if $|q\cap y|/|U|\ge w_2$.
This takes time $\min\{w_u,w_q\}|U|$, which would be a large multiplicative factor on our query time, so we may instead choose to sample just 
\begin{align}O(\min\{w_q,w_u\} w_2^{-1}\log n)\label{eq:compare-cost}\end{align}
elements, which suffices as a test with high probability.

This describes the entire algorithm, exception for an optimization for the ``Sample the universe'' step above, which naively implemented would take time $|X|$.
This optimization is the reason $|U|$ was chosen to be a prime number.

\vspace{-.5em}
\paragraph{An optimization}
In the ``Sample the universe'' step of \Cref{alg:treesample} a naive implementation spends time $|X|$ hashing all possible elements and comparing their value to $\Delta_i$.
We now show how to make this step output sensitive, using only time equal to the number of values for which the condition is true.
\footnote{The subroutine is inspired by personal communications with Rasmus Pagh and Tobias Christiani.}

The requirement $s' \ge it-c_i$ we call the ``trimming condition''.
This allows us to trim away most prefix paths which would be very unlikely to ever reach our requirement for the final path.
To speed up finding all $x\in U$ such that $h_i(r\circ x) < \Delta_i$ we
note that there are two cases relevant to the trimming condition, depending on $s$ in the algorithm:
(1) $s'$ has to be $s+1$ or (2) $s'=s$ suffices.
In the first case we are only interested in $x$ values in $X$, while in the second case, all $x\in U$ values are relevant.

We have $h_i(r\circ x) = \eta + a x \mod q$ for some values $\eta$, $a$ and $b$ where $a>0$.
In case (2) the relevant $x$ are simple $\{a^{-1}(v - \eta) \mod q \mid v \in [\Delta_i]\}$, where $a^{-1}$ exists because $q$ is prime.
For the case (1) where $x$ must be in $X$, we pre-process $X$ by storing $ax \mod q$ for $x\in X$ in a sorted list.
Using a single binary search, we can then find the relevant values with a time overhead of just $\lg|X|$.
Using a more advanced predecessor data structure, this overhead can be reduced.
See \Cref{alg:sample} for a pseudocode version of this idea.
\begin{algorithm}
   \caption{Output sensitive sample}
   \label{alg:sample}
   \SetKwInOut{Input}{Input}
   \SetKwInOut{Preprocess}{Pre-process}
   \Input{$r\in[q]$, $\Delta\in[q]$}
   \Preprocess{$s = \text{sorted}\{h(x)\mid x\in X\}\in[q]^{|X|}$
               and $\kappa\in X^{|X|}$ st. $h(\kappa[i]) = s[i]$.}
   \KwResult{$R = \{x \in X \mid (h(x)+r\mod q) < \Delta\}$}
   $i \leftarrow \min\{i\in[|X|] \mid s[i-1] < q-r \le s[i]\}$
   \tcp*[f]{We assume $s[i]=-\infty$ for $i<0$}\\
   $R \leftarrow \{\}$\\
   \While{$(s[i\mod |X|] + r \mod q) < \Delta$}{
      $R \leftarrow R\cup\{\kappa[i\mod |X|]\}$\\
      $i \leftarrow i+1$
   }
\end{algorithm}

\subsection{Full Theorem}\label{sec:full}
We state the full version of \Cref{thm:main} and a discussion of the differences between it and the idealized version in the introduction.
\begingroup
\def\thetheorem{\ref{thm:main}}
\begin{theorem}[Full version]
   Let $w_q, w_u \ge w_1 \ge w_2 \ge 0$ be given with $w_1 \ge w_q w_u$ and $1 \le t_q, t_u \le 0$.
   Set $k$ to be the smallest even integer greater than or equal to
   $\frac{\log n}{\D{ T_2 }{ P_2 } - \dq}$ and assume that $t_q k/2$ and
   $t_u k/2$ are integers. The $(w_q, w_u, w_1, w_2)$-GapSS problem over a universe $U$ can be solved with expected query time
   \begin{align}
      \text{query time}\quad
      &O\big(\varsigma_q\, k^{28} \, n^{\rho_q} + k w_q \abs{U}
      + \big( \tfrac{t_q(1 - w_q)}{(1-t_q)w_q} \big)^{\sqrt{t_q(1-t_q) k\,\cdot\, 6.5 \log(3k)})}\big),
      \\
      \text{space usage}\quad
      &O(\varsigma_u\, k^{28}\, n^{1 + \rho_u} + n w_u \abs{U})
      \\
      \text{and update time}\quad
      &O\big(\varsigma_u\, k^{28}\, n^{\rho_u} + k w_u \abs{U}
      + \big( \tfrac{t_u(1-w_u)}{(1-t_u)w_u} \big)^{\sqrt{t_u(1-t_u) k\,\cdot\, 6.5 \log(3k)})}\big),
   \end{align}
   \vspace{-1.5em}
   \begin{align}
      \text{where}\quad
      \rho_q = \frac{\D{ T_1 }{ P_1 }-\dq}{\D{ T_2 }{ P_2 }-\dq}
      \quad\text{and}\quad
      \rho_u = \frac{\D{ T_1 }{ P_1 }-\du}{\D{ T_2 }{ P_2 }-\dq},
      \\[10pt]
      \text{and}\quad
      \varsigma_q = \tfrac{\min\{w_q,w_u\}}{w_2}e^{2(\D{ T_1 }{ P_1 } - \dq)},
      \quad\varsigma_u = e^{2(\D{ T_1 }{ P_1 } - \du)}.
   \end{align}
\end{theorem}
\addtocounter{theorem}{-1}
\endgroup

We stress that all previous Locality Sensitive algorithms with time/space trade-offs had $n^{o(1)}$ factors on $n^{\rho_q}$ and $n^{\rho_u}$.
These could be as large as $\exp(\sqrt{\log n})$ or even $\exp((\log n)/(\log \log n))$.
In contrast, our algorithm is the first that only loses $k\approx\log(n)$ multiplicative factors!

In the statement of \Cref{thm:main} we have taken great effort to make sure that any dependence on $w_q,w_u,w_1,w_2,t_q,t_u$ is visible and only truly universal constants, like 4, are hidden in the $O(\cdot)$.

The main thing we do lose is the additive $(\tfrac{t_q(1-w_q)}{(1-t_q)w_q})^{\tilde O(\sqrt{t_q(1-t_q)k})}$.
We may note the bound $(\tfrac{t_q}{1-t_q})^{\sqrt{t_q(1-t_q)}}\le 2$, so the main eyesore is the $1/w_q$.
For $w_q > e^{-\tilde O(\sqrt{\log n})}$ this is dominated by the main term, but for very small sets it could potentially be an issue.
However, it turns out that as $w_q$ and $w_u$ get small, the optimal choices of $t_q$ and $t_u$ move towards 0 or 1.
Since this effect is exponentially stronger we get that $(1/w_q)^{\sqrt{t_q(1-t_q)}}$ is usually never more than a small constant.
It also means that we recover the performance of Chosen Path in the case $t_q=1$, $t_u=1$, which has no $\Omega(e^{\sqrt{\log n}})$ terms.
\footnote{
   The authors know of a way to reduce the error term further, so it only appears in the $\rho_q=0$ case, and only as $\exp((\log1/w_q)^{2/3}k^{1/3})$ which is $o(n)$ for any $w_q=\omega(1/n)$.
}

In case $w_q^{-1}$ is large, but $w_2$ is not too small, we can reduce $w_q^{-1}$ to  $\frac{w_u}{w_2}k$ by hashing!
Sketch: Define a hash function $h:U\to [m]$ where $m=O(\frac{w_qw_u}{w_2}|U|k)$ and map each set $y$ to $\{i\in [m] \mid \exists e\in y : h(e)=i\}$, that is the OR of the hashed values.
With high probability this only distorts the size of the sets and their inner products by a factor $(1+1/k)$ which doesn't change $\rho$.

The constants of the size $\varsigma_q$ and $\varsigma_u$ are standard in all other similar algorithms since~\cite{indyk1998approximate}, as they come from the requirement that $k$ is an integer.
The terms $\D{T_1}{P_1}-\dq$ and $\D{T_1}{P_1}-\du$ in $\varsigma_q$ and $\varsigma_u$ may be bounded by $\log\frac{w_u}{w_1}$ and $\log\frac{w_q}{w_1}$ respectively.
The factor of $2$ on those terms come from the tensoring step done on paths of length $k/2$.
This can be removed at the cost of making the ratio-of-odds term multiplicative in the bounds above.
The factor $\min\{w_q,w_u\}/w_2$ in $\varsigma_q$ comes from equation \eqref{eq:compare-cost} and is the time it takes to verify a candidate identified by the filtering.
Note that this factor would exist even in a brute force $O(n)$ algorithm
and exists in any data structures known for similar problems.
In fact, for small $n$, it is necessary due to communication complexity bounds.

\begin{proof}[Proof of \Cref{thm:main}]

Let $\mathcal T_q$ and $\mathcal T_u$ be the time it takes to compute $R_k(x, t_q)$ and $R_k(y, t_u)$ on given sets.
When creating the data structure, decoding each $y\in Y$ takes time $n\mathcal T_u$ and uses $n\ep{\abs{R_k(Y, t_u)}}$ words of memory for space equivalent.
When querying the data structure we first use time $\mathcal T_q$ to decode $q$,
then $\ep{\abs{R_k(X, t_q)}}$ time to look in the buckets,
and finally $\frac{\min\{w_q,w_u\}\log n}{w_2}$ time on each of $\ep{\abs{R_k(X, t_q) \cap R_k(Y, t_u)}} n$ expected collisions with far sets (the worst case is that we never find any $y$ with $y\cap q>w_2|U|$ so we can't return early.)

The key to proving the theorem is thus bounding the above quantities.
We do this using the following lemma, which we prove at the end of the section:

\begin{lemma}\label{lem:tree-sample-properties}
   In \Cref{alg:treesample} let $k \in \mathbb Z_+$ and let $w_q, w_u, w_1, w_2 \in [0, 1]$ be the Gap-SS parameters
   such that $w_1 \ge w_q w_u$. Now let $t_q, t_u \in [0, 1]$ be the thresholds such that $t_q k$ and
   $t_u k$ are integers, and let $\Delta > 0$ be the branching factor. Given a query set $X$, with
   $\abs{X} = w_q \abs{U}$, and data set $Y \subseteq U$, with $\abs{Y} = w_u \abs{U}$, then running
   \Cref{alg:treesample} with $c_\ell = \smat{\sqrt{t_q(1-t_q)}\\\sqrt{t_u(1-t_u)}}\cdot\sqrt{6.5 \ell \log(3 k)}$ for $\ell < k$ and $c_k = \svec{0\\0}$, gives that
   \begin{align}
      \ep{\abs{R_k(X, t_q)}} &\le 2\Delta^{k} \exp(-k \dq) \; .
      \label{eq:q-bound}
      \\
      \ep{\abs{R_k(Y, t_u)}} &\le 2\Delta^{k} \exp(-k \du) \; .
      \label{eq:u-bound}
      \\
      \prp{\abs{R_k(X, t_q) \cap R_k(Y, t_u)} \ge 1} &\ge 7^{-8} k^{-14} \Delta^{k} \exp(-k \D{ T_1 }{ P_1 })
      \quad&\text{if } \abs{X\cap Y} \ge w_1 \abs{U}.
      \label{eq:1-bound}
      \\
      \ep{\abs{R_k(X, t_q) \cap R_k(Y, t_u)}} &\le 2\Delta^{k} \exp(-k \D{ T_2 }{ P_2 })
      \quad&\text{if } \abs{X\cap Y} \le w_2 \abs{U}.
      \label{eq:2-bound}
   \end{align}
   where
   $P_j = \left(\begin{smallmatrix}w_j & w_q - w_j \\ w_u - w_j & 1 - w_q - w_u + w_j\end{smallmatrix}\right)$,
   $T_j = \left(\begin{smallmatrix}t_j & t_q - t_j \\ t_u - t_j & 1 - t_q - t_u + t_j\end{smallmatrix}\right)$,
   $t_j = \arg\inf \D{ T_j }{ P_j }$ for $j \in \set{1, 2}$.

   Finally the expected running times, $\mathcal T_q$ and $\mathcal T_u$, it takes to compute $R_k(X, t_q)$ and $R_k(Y, t_u)$
   respectively are bounded by
   \begin{equation}
      \begin{array}{r@{\,}l}
         \ep{\mathcal T_q} &\le O\big(k \abs{X} + k(k + \log(n)) \Delta^k \exp(-k \dq) \left( \tfrac{t_q(1 - w_q)}{(1-t_q)w_q } \right)^{(c_k)_1}\big)\; . \\
         \ep{\mathcal T_u} &\le O\big(k \abs{Y} + k(k + \log(n)) \Delta^k \exp(-k \du) \left( \tfrac{t_u(1 - w_u)}{(1-t_u)w_u } \right)^{(c_k)_2}\big)\; .
      \end{array}
      \label{eq:t-bound}
   \end{equation}
\end{lemma}

We define $\Delta = \exp(\D{ T_1 }{ P_1 })$, and let $k$ be the smallest even integer at least $\frac{\log n}{\D{ T_2 }{ P_2 } - \dq}$.
Define the sequence $\Delta_i=2^{l_i}$ for some $l_i\in\mathbb Z_{\ge0}$ such that $\prod_{j=1}^i\Delta_j \le \Delta^i < 2\prod_{j=1}^i \Delta_j$ for all $i\in[k]$.

We make 2 initiations of \Cref{alg:treesample},
$M_1, M_2$, with height $k/2$.
$\Delta$ and $c_\ell$ are adjusted correspondingly.
In we have $c_{k/2} = c_k = \svec{0\\0}$.

For each instance we have
\begin{align}
   \ep{\abs{R_{k/2}(X, t_q)}}
   &\le 2\exp(k/2\, \D{ T_1 }{ P_1 } -k/2\, \dq)
 \\&\le 2\exp\left(\left(\frac{\log n}{\D{ T_2 }{ P_2 } - \dq}+2\right) \frac{(\D{ T_1 }{ P_1 } - \dq)}{2}\right)
 \\&= 2n^{\frac12 \frac{\D{ T_1 }{ P_1 } - \dq}{\D{ T_2 }{ P_2 } - \dq}}(\D{ T_1 }{ P_1 } - \dq).
 \\
 \text{similarly we get}
 \\
   \ep{\abs{R_{k/2}(X, t_q)}}
   &\le 2n^{\frac12 \frac{\D{ T_1 }{ P_1 } - \du}{\D{ T_2 }{ P_2 } - \dq}}(\D{ T_1 }{ P_1 } - \du).
\end{align}

We combine the two data instances $M_1$ and $M_2$ by taking as representative sets returned the product of the sets returned by each of them.
In particular, this means we successfully find a near set, if
$\abs{R_{k/2}(X, t_q) \cap R_{k/2}(Y, t_u)} \ge 1$ for both instances, which happens with probability at least
\begin{align}
   (7^{-8} k^{-14} \Delta^{k} \exp(-k \D{ T_1 }{ P_1 }))^2
   =
   (7^{-8} k^{-14})^2.
\end{align}
hence, repeating the algorithm $C k^{28}$ times, for some $C$, we can boost this probability to $99\%$.

Putting it all together now yields the full version of \Cref{thm:main} contingent on~\Cref{lem:tree-sample-properties}.

\end{proof}

\subsection{Bounds on Branching}
It now remains to prove~\Cref{lem:tree-sample-properties}.
The inequalities \eqref{eq:q-bound}, \eqref{eq:u-bound} and \eqref{eq:2-bound} are all simple calculations based on linearity of expectation.
The time bound~\eqref{eq:t-bound} is also fairly simple, but we have to take the decoding optimization described above into account.
We also need to bound the number of paths alive at some point during the decoding process, which requires being more careful about the trimming conditions.

Finally the proof of the probability lower bound~\eqref{eq:1-bound} is the main star of the section.
We do this using essentially a second-moment method, but a number of tricks are needed in order to squeeze out acceptable bounds, taking into account that any of $w_q,w_u,w_1,w_2$ may be $o(1)$, which among other things forbid the use of many Central Limit Theorem type results.

\begin{proof}[Proof of \eqref{eq:q-bound} and \eqref{eq:u-bound}]
   We only provide the proof for \eqref{eq:q-bound} since the proof \eqref{eq:u-bound}
   is analogous.

   Let $r \in R_k$ be a representative string and define the random variables
   $\mathcal X^{(i)} = [r_i \in X]$ for $i \in [k]$, because the hash functions $h_i$ used at each level of the tree are independent, so are the $(\mathcal X^{(i)})_{i \in [k]}$ independent.

   We use linearity of expectation, and completely throw away
   the fact that some branches may have been cut early.
   Throwing away extra cuts of course only increases the probability of survival.
   Meanwhile, we do not expect to gain more than factors of $k$ this way, compared to a sharp analysis, since the whole point of the algorithm is to efficiently approximate cuts done only at the leaf level.
   \begin{align}
      \ep{\abs{R_k(X, t_q)}}
         &\le \abs{R_k} \prp{\forall \ell \le k : \sum_{i \in [\ell]} \mathcal X^{(i)} \ge t_q \ell - (c_\ell)_1}
      \\&\le \abs{R_k} \prp{\sum_{i \in [k]} \mathcal X^{(i)} \ge t_q k}
      \\& \le \abs{R_k} \exp(-k \dq)
      .
   \end{align}
   The final bound is the entropy Chernoff bound we use everywhere.
   Since $|R_k| = \prod_{i=1}^k \Delta_i \le 2\Delta^k$ we get the bound.
\end{proof}

\begin{proof}[Proof of \eqref{eq:2-bound}]
   This is similar to the proof of \eqref{eq:q-bound} and \eqref{eq:u-bound}, but two dimensional.
   Like in the those proofs we consider a single representative string $r \in R_k$ and define the random variables $\mathcal X^{(i)} = \smat{[r_i \in X] \\ [r_i \in Y] }$ for $i \in [k]$.
   By definition of \Cref{alg:treesample} $(\mathcal X^{(i)})_{i \in [k]}$ are independent. 

   We then bound using linearity of expectation:
   \begin{align}
      \ep{\abs{R_k(X, t_q) \cap R_k(Y, t_u)}}
         &\le \abs{R_k}
         \prp{\forall \ell \le k : \sum_{i \in [\ell]} \mathcal X^{(i)} \ge \svec{t_q \\ t_u} \ell - c_\ell}
         \\&\le 2\Delta^k \prp{\sum_{i \in [k]} \mathcal X^{(i)} \ge \svec{t_q \\ t_u} k}
         \\&\le 2\Delta^k \exp(- \D{ T_2 }{ P_2 })
   \end{align}
\end{proof}

\begin{proof}[Proof of \eqref{eq:t-bound}]
   As a preprocessing stage we make $k$ sorted lists of $(a_i x)_{x \in X}$ where
   $a_i$ is the coefficient in $h_i(p \circ x) = h'_i(p) + a_i x \mod q$, this takes
   $O(k\abs{X})$ time.

   We will argue that at each level of tree that we only use $O(k + \log \abs{X}) = O(k + \log n)$ 
   amortized time per active path. More precisely, at level $l$ we use $O((k + \log n) \abs{R_\ell(X, t_q)})$
   amortized time.
   
   Let $\ell \in [k]$ be fixed and consider an active path $r \in R_\ell(X, t_q)$.
   If $\sum_{i \in [\ell]} [r_i \in X] \ge t_q(l + 1) - (c_{l - 1})_1$ then every one of its children
   will be active. So we need to find $\setbuilder{x \in U}{h_\ell(p \circ x) < \Delta_\ell}
   = h^{-1}_\ell([\Delta_\ell])$.
   Now $h_\ell(p \circ x) = h'_\ell(p) + ax \mod q$ where $a \neq 0 \mod q$
   and $s = h'_\ell(p)$ can be computed in $O(k)$ time. We then get that
   $h^{-1}_\ell([\Delta_\ell]) = \setbuilder{a^{-1}(i - s)}{i \in [\Delta_\ell]}$,
   this we can find in time proportional with the number of active children,
   so charging the cost to them gives the result.
   
   If $\sum_{i \in [\ell]} [r_i \in X] < t_q(l + 1) - (c_{l - 1})_1$ then only the children
   $r \circ x \in R_{l + 1}$ where $x \in X$ will be active. So we need to find
   $\setbuilder{x \in X}{h_\ell(p \circ x) < \Delta_\ell}$. Again using that
   $h_\ell(p \circ x) = h'_\ell(p) + ax \mod q$ where $a \neq 0 \mod q$ and $s = h'_\ell(p)$
   can be computed in $O(k)$ time, we have reduced the problem to finding
   $h^{-1}_\ell([\Delta_\ell]) = \setbuilder{x \in X}{s + ax \mod q < \Delta_\ell}$.
   This we note we can rewrite as $h^{-1}_\ell([\Delta_\ell])
   = \setbuilder{x \in X}{s \le ax \vee ax < \Delta + s - q}$,
   so using our sorted list this can be done in $O(\log n)$ time plus time proportional with
   the number of active children, so charging this cost to them gives the result.

   \vspace{.5em}

   We bound the expected number of active paths on a level $\ell \in [k]$.
   Let $r \in R_\ell$ be a representative string and define the random
   variables $\mathcal X^{(i)} = [r_i \in X]$ for $i \in [k]$, by definition of
   \Cref{alg:treesample} $(\mathcal X^{(i)})_{i \in [k]}$ are independent. 
   We then bound
   \begin{align}
      \prp{\sum_{i \in [l]} \mathcal X^{(i)} \ge t_q \ell - (c_\ell)_1}
         &\le
      \prp{\forall j \le \ell : \sum_{i \in [j]} \mathcal X^{(i)} \ge t_q j - (c_j)_1}
      \\&\le \prp{\sum_{i \in [\ell]} \mathcal X^{(i)} \ge t_q \ell - (c_\ell)_1}
      \\&\le \exp(-l \di{t_q - c_\ell/l}{w_q})
      \\&\le  \exp(-l \dq) \left( \tfrac{t_q (1 - w_q)}{w_q (1 - t_q)} \right)^{(c_\ell)_1}
      \; .
   \end{align}
   The crucial step here was using the identity
   \begin{align}
      \di{t_q-\eps}{w_q} = \di{t_q}{w_q} - \eps\log \tfrac{t_q (1 - w_q)}{w_q (1 - t_q)} + \di{t_q-\eps}{t_q}
   \end{align}
   from which we can ignore the $\di{t_q-\eps}{t_q}$ term, since it is positive.

   Using linearity of expectation we get that
   \[\begin{split}
      \ep{\abs{R_\ell(X, t_q)}}
         &\le \abs{R_\ell}\exp(-\ell \dq) \left( \tfrac{t_q (1 - w_q)}{w_q (1 - t_q)} \right)^{(c_\ell)_1}
       \\&\le 2 \Delta^\ell \exp(-\ell \dq) \left( \tfrac{t_q (1 - w_q)}{w_q (1 - t_q)} \right)^{(c_\ell)_1}
      \; .
   \end{split}\]

   Now the expected cost of the tree becomes
   \[\begin{split}
      \ep{\sum_{\ell \in [k]} O((k + \log(n)) \abs{R_\ell(X, t_q)})}
         &= O((k + \log(n)) \sum_{\ell \in [k]} \ep{\abs{R_\ell(X, t_q)}})
       \\&\le O(k(k + \log(n)) \Delta^k \exp(-k \dq) \left( \tfrac{t_q (1 - w_q)}{w_q (1 - t_q)} \right)^{(c_\ell)_1})
      \; .
   \end{split}\]
\end{proof}

Note that we throw away some leverage here by bounding the size of each level by the final level.
We might have defined $c_\ell$ such that $\ell\Delta-\ell\dq + c_\ell\log \tfrac{t_q (1 - w_q)}{w_q (1 - t_q)} - \ell\di{t_q-c_\ell/\ell}{t_q} = k\Delta-k\dq$ and still used the same bound.
The only later requirement we set the $c_\ell$ is that
$\sum_{\ell\in[k]} \exp(-\ell\di{t_q-c_\ell/\ell}{t_q})$ sum to $1/\poly(k)$.

Making this change could potentially kill the
$\left( \tfrac{t_q (1 - w_q)}{w_q (1 - t_q)} \right)^{(c_\ell)_1}$ factor, which is a bit of an eye sore.
However in the near-constant query time case, which is really when this factor (or term once we using the tensoring trick) is relevant, this trick wouldn't work, since we then have exactly $\Delta = \dq$.

\vspace{.5em}

For the final proof we need the following lemma,
which bounds the probability that an unbiased Bernoulli $2d$ random walk stays entirely in the negative quadrant.
A lemma like this is an exercise to show using the Central Limit Theorem and convergence to Brownian motion.
However, our bound is non-asymptotic, making no assumptions about the relationship between the probability distribution of $X_i$ and the size of $n$.
There are non-asymptotic CLT bounds, like Berry Esseen, but unfortunately multivariate Berry Esseen bounds for random walks are not very developed.

\begin{lemma}[The probability that a random walk stays in a quadrant]\label{lem:quadrant}
   Let $X_1, \dots, X_k\in\{0,1\}^2$ be iid. Bernoulli $2d$-random variables with probability matrix $\smat{p & p_{1}-p\\p_{2}-p&1-p_1-p_2+p}$.
   Assume that the coordinates are correlated, that is $p\ge p_1p_2$,
   and assume $p_q k$ and $p_2 k$ are integers.

   Let $S_\ell = \sum_{i\in[\ell]} X_i$ be the associated random walk.
   Then
   \begin{align}
      \Pr[\forall \ell\in[k] : S_\ell \le 0] \ge \frac{1}{400\, k^{6.5}}
      .
   \end{align}
\end{lemma}
The proof of this is in \Cref{app:technical}.

\begin{proof}[Proof of \eqref{eq:1-bound}]
   We will prove this bound using the second moment method.
   For this to work, it is critical that we restrict our representative strings further and consider
   \[
      S = \setbuilder{r \in R_k}{
         \forall \ell \le k : \smat{ [r_i \in X] \\ [r_i \in Y] } \ell - c_\ell \le \sum_{i \in [\ell]} \smat{ [r_i \in X] \\ [r_i \in Y] } \le \svec{t_q \\ t_u} \ell
      }
      \; ,
   \]
   It is easy to check that $S \subseteq R_k(X, t_q) \cap R_k(Y, t_u)$, thus we have that
   \[
      \prp{\abs{R_k(X, t_q) \cap R_k(Y, t_u)} \ge 1}
         \ge \prp{\abs{S} \ge 1}
         \ge \ep{\abs{S}}^2\big/\ep{\abs{S}^2}
      \; ,
   \]
   where the last bound is Paley-Zygmund's inequality.
   We then need to do two things: 1) Lower bound $\ep{\abs{S}}$, and 2)
   Upper bound $\ep{\abs{S}^2}$.

\paragraph{Lower bounding $\ep{\abs{S}}$.}
   Let $r \in R_k$ be a representative string and define the random variables
   $\mathcal X^{(i)} = \smat{ [r_i \in X] \\ [r_i \in Y] }$ for $i \in [k]$.
   Each one has distribution $P = \smat{w_1&w_q-w_1\\w_u-w_q&1-w_q-w_u+w_1}$.
   We then introduce variables $\tilde{\mathcal X}^{(i)}$ with law $T=\smat{t_1&t_q-t_1\\t_u-t_q&1-t_q-t_u+t_1}$, where $t_1$ minimizes $\D TP$ as defined in the algorithm.

   We then use the following variation on Sanov's theorem:
   \begin{lemma}
      For any set $A\subseteq \R^{2\times n}$ we have
      \[\begin{split}
         \prp{(\mathcal X^{(i)})_{i \in [k]} \in A}
            = \exp(-k \D{ T }{ P }) \prp{(\tilde{\mathcal{X}}^{(i)})_{i \in [k]} \in A}
      \end{split}\]
   \end{lemma}
   \begin{proof}
      Define the logarithmic moment generating function 
      $
         \Lambda(\lambda) = \log \ep{\exp(\langle \lambda, \mathcal X \rangle)}
      $,
      and let
      $z = (\nabla_x \Lambda^{*})(t)$.
      By a standard correspondence, (see e.g.~\cite{polyanskiy2014lecture} Chapter 14 or~\cite{DBLP:journals/siamrev/Dinwoodie94} Chapter 6.2), we have that
      \begin{align}
         \dd T(x) = \exp(\langle z,x\rangle - \Lambda(z))\dd P(x)
      \end{align}
      for Radon–Nikodym derivates $\dd T$ and $\dd P$.
      Now using the exponential change of measure, we get that
      \[\begin{split}
         \prp{(\mathcal X^{(i)})_{i \in [k]} \in A}
            &= \int_{(x^{(i)})_{i \in [k]} \in A} \, \dd P^{\otimes k}
            \\&= \int_{(\tilde{x}^{(i)})_{i \in [k]} \in A} \exp\left(k \Lambda(z) - \left\langle z, \sum_{i \in [k]} \tilde{x}^{(i)} \right\rangle\right) \, \dd T^{\otimes k}
            \\&= \exp(-k \D{ T }{ P }) \int_{(\tilde{x}^{(i)})_{i \in [k]} \in A} \exp\left(- \left\langle z, \sum_{i \in [k]} \left( \tilde{x}^{(i)} - \svec{t_q\\ t_u} \right) \right\rangle \right) \, \dd T^{\otimes k}
            \\&= \exp(-k \D{ T }{ P }) \prp{(\tilde{\mathcal{X}}^{(i)})_{i \in [k]} \in A}
            ,
      \end{split}\]
      where the last inequality follows from the fact that if $(\tilde{x}^{(i)})_{i \in [k]} \in A$ then
      $\sum_{i \in [k]} \tilde{x}^{(i)} = \svec{t_q\\t_u}k$.
   \end{proof}

   For convenience we will sometimes write $T=\smat{t_{11} & t_{12} \\ t_{21} & t_{22}}$.
   Note that by assumption $t_qk = (t_{12}+t_{11})k$ and $t_uk=(t_{21}+t_{11})k$ are integers, but values such as $t_{11}k$ and $t_{22}k$ need not be.
   This will 

   We define the sets
   \begin{align}
      U &= \setbuilder{(x^{(i)})_{i \in [k]} \in \R^{2\times k}}{\forall \ell \le k :
      \sum_{i \in [\ell]} x^{(i)} \le \svec{t_q\\t_u}\ell}
      \\
      \text{and}\quad
         L &= \setbuilder{(x^{(i)})_{i \in [k]} \in \R^{2\times k}}{\forall \ell \le k :
        \sum_{i \in [\ell]} x^{(i)} \ge \svec{t_q\\t_u}\ell - c_\ell}
   \end{align}
   such that $U\cap L$ are all sequences satisfying our path requirement.
   In other words
      $E|S| = \exp(-k\D{T}{P})\prp{(\tilde{\mathcal{X}}^{(i)})_{i \in [k]} \in U\cap L}$.
   Using a union bound we split up:
   \[\begin{split}
      \prp{(\tilde{\mathcal X}^{(i)})_{i \in [k]} \in U \cap L}
   \ge
      \prp{(\tilde{\mathcal X}^{(i)})_{i \in [k]} \in U}
      - \prp{(\tilde{\mathcal X}^{(i)})_{i \in [k]} \in L}
      .
   \end{split}\]

   The term is bounded by \Cref{lem:quadrant} from the Appendix.
   Once we notice that $w_1\ge w_qw_u$ implies that $t_1\ge t_qt_u$.
   One way to see this is that $t_1$ minimizing $\D TP$ gives rise to the equation $\frac{w_1(1 - w_q - w_u + w_1)}{(w_q - w_1)(w_u - w_1)}=\frac{t_1(1 - t_q - t_u + t_1)}{(t_q - t_1)(t_u - t_1)} = \frac{t_1 + t_1(t_1 - t_q - t_u)}{t_q t_u + t_1(t_1 - t_q - t_u)}$.
   If $w_1\ge w_qw_u$ the left hand side is $\ge 1$, and so we must have $t_1\ge t_qt_u$.

   \Cref{lem:quadrant} then gives us
   \begin{align}
      \prp{(\tilde{\mathcal X}^{(i)})_{i \in [k]} \in U}
      \ge \frac{1}{400} k^{-3.5} \; .
   \end{align}

   This is a pretty small value, so for the union bound to work we need an even smaller probability for the lower bound.

   We bound each coordinate individually.
   The cases are symmetric, so we only consider the first coordinate.
   Using another union bound and Bernstein's inequality we get
   \begin{align}
      \prp{\exists \ell \le k : \sum_{i = 1}^\ell \tilde{\mathcal X}^{(i)}_1  \le t_q \ell - (c_\ell)_1}
      &= \sum_{l \le k} \prp{\sum_{i = 1}^\ell \tilde{\mathcal X}^{(i)}_1  \le t_q \ell - (c_\ell)_1}
      \\&\le \sum_{l\le k}\exp\left(\frac{-(c_\ell)_1^2/2}{(1-t_q)t_q\ell+(1-2t_q)(c_\ell)_1/3}\right)
      \\&\le\frac{1}{1200} k^{-6.5}
      .
   \end{align}
   since $(c_\ell)_1 = \Omega(\sqrt{t_q(1-t_q)l\log l} + \abs{1-2t_q}\log l)$.

   Similarly, we upper bound
   $\prp{\exists \ell \le k : \sum_{i = 1}^\ell \tilde{\mathcal X}^{(i)}_2  \le t_u \ell - (c_\ell)_2} \le \frac{1}{1200} k^{-6.5}$.
   Putting it all together we get
   \[
      \prp{(\mathcal X^{(i)})_{i \in [k]} \in A}
         \ge \frac{1}{1200} \exp(-k \D{ T_1 }{ P_1 }) k^{-6.5} \; ,
   \]
   so by linearity of expectation we get that
   \[
      \ep{S}
         \ge \abs{R_k} \frac{1}{1200} k^{-6.5} \exp(-k\D{ T }{ P })
         \ge \frac{1}{1200} k^{-6.5} \Delta^k \exp(-k \D{ T }{ P }) \; .
   \]
   
   \subsubsection*{Upper bounding $\ep{\abs{S}^2}$}
   Consider two representative strings $r, r' \in R_k$ and let $q \in R_\ell$ be their common prefix,
   hence $l$ is the length of their common prefix. Define the random variables
   $\mathcal X^{(i)} = \smat{ [r_i \in X] \\ [r_i \in Y] }$,
   $\mathcal Y^{(j)} = \smat{ [r'_j \in X] \\ [r'_j \in Y] }$,
   and $\mathcal Z^{(h)} = \smat{[q_h \in X] \\ [q_h \in Y] }$ for
   $i, j \in [k] \setminus [\ell]$ and $h \in [l]$. We then get that
   \[\begin{split}
      \prp{p, p' \in S}
         &\le \prp{\sum_{h \in [\ell]} \mathcal Z^{(h)} + \sum_{i \in [k] \setminus [l]} \mathcal X^{(i)} \ge tk
            \wedge \sum_{h \in [\ell]} \mathcal Z^{(h)} + \sum_{j \in [k] \setminus [l]} \mathcal Y^{(j)} \ge tk
            \wedge \sum_{h \in [\ell]} \mathcal Z^{(h)} \le tl}
         \\&\le \prp{\sum_{h \in [\ell]} \mathcal Z^{(h)} + \sum_{i \in [k] \setminus [l]} \mathcal X^{(i)} + \sum_{j \in [k] \setminus [l]} \mathcal Y^{(j)} \ge (2k - \ell)t}
      \; .
   \end{split}\]
   Now $\sum_{h \in [\ell]} \mathcal Z^{(h)} + \sum_{i \in [k] \setminus [l]} \mathcal X^{(i)} + \sum_{j \in [k] \setminus [l]} \mathcal Y^{(j)}$
   is almost a sum of independent random variable. We have that $\mathcal X^{(k - \ell + 1)}$
   and $\mathcal Y^{(k - \ell + 1)}$ are correlated since they are chosen by sampling without replacement,
   but this implies that \[\ep{\exp(\langle \lambda, \mathcal X^{(k - \ell + 1)} + \mathcal Y^{(k - \ell + 1)} \rangle)}
   \le \ep{\exp(\langle \lambda, \mathcal X^{(k - \ell + 1)} \rangle)}\ep{\exp(\langle \lambda, \mathcal Y^{(k - \ell + 1)} \rangle)}\]
   We can then use a 2-dimensional Entropy-Chernoff bound and get that
   \[\begin{split}
      \prp{\sum_{h \in [\ell]} \mathcal Z^{(h)} + \sum_{i \in [k] \setminus [l]} \mathcal X^{(i)} + \sum_{j \in [k] \setminus [l]} \mathcal Y^{(j)} \ge (2k - \ell)t}
         \le \exp(-(2k - \ell)\D{ T }{ P })
      \; ,
   \end{split}\]

   Using this we can upper bound $\ep{\abs{S}^2} = \ep{\sum_{r, r' \in R_k} [r, r' \in S]}$ by
   splitting the sum by the length of their common prefix.
   \[\begin{split}
      \ep{\abs{S}^2}
         &= \ep{\sum_{r, r' \in S_k} [r, r' \in S]}
         \\&\le \sum_{i = 1}^{k} \left(\prod_{j = 1}^{i} \Delta_j \right) \binom{\Delta_{i + 1}}{2} \left(\prod_{j = i + 2}^k \Delta_j \right) \exp(-(2k - i)\D{ T }{ P })
         \\&\le \left( \prod_{j = 1}^{k} \Delta_j \right)^2 \exp(-2k \D{ T }{ P }) \sum_{i = 1}^{k} \exp(i \D{ T }{ P }) \left( \prod_{j = 1}^i \Delta_j \right)^{-1}
         \\&\le \left( \prod_{j = 1}^{k} \Delta_j \right)^2 \cdot \exp(-2k \D{ T }{ P }) \cdot k \cdot \exp(k \D{ T }{ P }) \cdot \Delta^{-k}
         \\&\le 4 k \Delta^k \exp(- k \D{ T }{ P })
   \end{split}\]

   \subsubsection*{Finishing the proof}
   Having lower bounded $\ep{\abs{S}}$ and upper bounded $\ep{\abs{S}^2}$ we can finish the proof.
   \[\begin{split}
      \prp{\abs{R_k(X, t_q) \cap R_k(Y, t_u)} \ge 1}
         &\ge \prp{\abs{S} \ge 1}
         \\&\ge \frac{\ep{\abs{S}}^2}{\ep{\abs{S}^2}}
         \\&\ge \frac{\frac{1}{1200^2} k^{-13} \Delta^{2k} \exp(-2k \D{ T }{ P })}{4 k \Delta^k \exp(- k \D{ T }{ P })}
         \\&= \frac{1}{2400^2} k^{-14} \Delta^\ell \exp(-k \D{ T }{ P })
      \; .
   \end{split}\]
\end{proof}

\subsection{Central Random Walks}\label{app:technical}

The main goal of this section is to prove \Cref{lem:quadrant},
which polynomially in $k$ lower bounds the probability that a biased random walk on $\mathbb Z^2$ always stays below its means.
Asymptotically, this can be done in various ways using the Central Limit Theorem for Brownian Motion, but as far as we know there are no standard ways to prove such a result in a quantitative way.

What we would really want is a Multidimensional Berry Esseen for Random Walks.
Instead we prove something specifically for walks where each iid. step 
   $X_1, \dots, X_k\in\{0,1\}^2$ be is a Bernoulli $2d$-random variables with probability matrix $\smat{p & p_{1}-p\\p_{2}-p&1-p_1-p_2+p}$.
   We need the further restrictions that the coordinates are correlated ($p\ge p_1p_2$), and that $p_1k$ and $p_2k$ are integers.

We will start by proving some partial results, simply bounding the probability that the final position of the random walk hits a specific value.
We then prove the lemma conditioned on hitting those values, and finally put it all together.

\begin{lemma}\label{lem:hitting-the-mean}
   Let $k \in \mathbb Z_+$ and $p_1, p_2 \in [0, 1]$, such that, both $p_1 k$ and $p_2 k$ are integers.
   Choose $p \in [0, 1]$, such that, $p \ge p_1 p_2$. Let $X^{(i)} \in \R^2$ be independent
   identically distributed 2-dimensional Bernoulli variables, where their probability matrix is
   $P = \smat{
     p        & p_1 - p \\
     p_2 - p  & 1 - p_1 - p_2 + p
   }$. We then get that
   \[
      \prp{\sum_{i \in [k]} X^{(i)} = \left(\begin{smallmatrix}p_1 \\ p_2\end{smallmatrix}\right) k \wedge \sum_{i \in [k]} X^{(i)}_1 X^{(i)}_2 = \ceil{p k}}
          \ge \frac{1}{400} k^{-3.5} \; .
   \]
\end{lemma}
In the proof we will be using the Stirling's approximation
\begin{align}\label{eq:stirling}
   \sqrt{2\pi n}n^n e^{-n} \le n! \le e \sqrt{n} n^n e^{-n} \; .   
\end{align}
This implies the following useful bounds on the binomial and multinomial coefficients.
\begin{align}
   \label{eq:stirling-binomial}
   \binom{n}{an} 
      &\ge \frac{\sqrt{2 \pi}}{e^2} \cdot \frac{\sqrt{n}}{\sqrt{an (1-a)n}} \frac{n^n}{(an)^{an} ((1 - a)n)^{(1 - a)n}}
      \\&\ge \frac{\sqrt{2 \pi}}{e^2} n^{-0.5} a^{-an} (1 - a)^{-(1 - a)n}
       \; . \\
   \label{eq:stirling-multinomial}
   \binom{n}{an, bn, cn}
      &\ge \frac{\sqrt{2 \pi}}{e^4} \cdot \frac{\sqrt{n}}{\sqrt{an bn cn (1 - a - b - c)n}}  \frac{n^n}{(an)^{an} (bn)^{bn} (cn)^{cn} ((1 - a - b -c)n)^{(1 - a - b - c)n}} 
      \\&\ge \frac{\sqrt{2 \pi}}{e^4} n^{-1.5} a^{-an} b^{-bn} c^{-cn} (1 - a - b - c)^{-(1 - a - b - c)n} 
       \; .
\end{align}
\begin{proof}
   If $p_1 = 1$ then $p = p_2$ and we get that
   \[
     \prp{\sum_{i \in [k]} X^{(i)}_2 = p_2 k}
       = \binom{k}{p_2 k} p_2^{p_2 k} (1 - p_2)^{(1 - p_2)k}
       \ge \frac{\sqrt{2 \pi}}{e^2} k^{-\frac{1}{2}} \; ,
   \]
   where we have used \cref{eq:stirling-binomial}.
   We get the same bound when $p_1 = 0$, $p_2 = 1$, or $p_2 = 0$.
 
   Now assume that $p_1, p_2 \not\in \set{0, 1}$, we then have that
   $\frac{1}{k} \le p_1 \le 1 - \frac{1}{k}$ and $\frac{1}{k} \le p_2 \le 1 - \frac{1}{k}$.
   We first note that
   \[\begin{split}
     &\prp{\sum_{i \in [k]} X^{(i)} = (p_1 k, p_2 k) \wedge \sum_{i \in [k]} X^{(i)}_1 X^{(i)}_2 = \ceil{p k}}
       \\&\phantom{==}= \binom{k}{\ceil{p k}, p_1 k - \ceil{p k}, p_2 k - \ceil{p k}}
         p^{\ceil{p k}}(p_1 - p)^{p_1 k - \ceil{p k}}(p_2 - p)^{p_2 k - \ceil{p k}}(1 - p_1 - p_2 + p)^{k - p_1 k - p_2 k + \ceil{p k}}
       \\&\phantom{==}\ge
         \frac{\sqrt{2 \pi}}{e^4} \cdot \frac{1}{k^{3/2}} \exp\Big(-\ceil{p k}\log \frac{\ceil{p k}}{p k} - (p_1 k - \ceil{p k})\log \frac{p_1 k - \ceil{p k}}{p_1 k - p k}
           \\&\phantom{====}- (p_2 k - \ceil{p k})\log \frac{p_2 k - \ceil{p k}}{p_2 k - p k} - (k - p_1 k - p_2 k + \ceil{p k})\log \frac{k - p_1 k - p_2 k + \ceil{p k}}{k - p_1 k - p_2 k + p k} \Big)
   \end{split}\]
   where we have used \cref{eq:stirling-multinomial}. We will bound each of the terms 
   $\ceil{p k}\log \frac{\ceil{p k}}{p k}$, $(p_1 k - \ceil{p k})\log \frac{p_1 k - \ceil{p k}}{p_1 k - p k}$,
   $(p_2 k - \ceil{p k})\log \frac{p_2 k - \ceil{p k}}{p_2 k - p k}$, and
   $(k - p_1 k - p_1 k + \ceil{p k})\log \frac{k - p_1 k - p_2 k + \ceil{p k}}{k - p_1 k - p_2 k + p k}$ individually.
 
   Using that $p \ge p_1 p_2 \ge \frac{1}{k^2}$
   we get that
   \[
     \ceil{p k}\log \frac{\ceil{p k}}{p k}
     = (1 + p k)\log\left(1 + \frac{1}{p k}\right)
     \le 1 + \log(1 + k) \; .
   \]
 
   Now using that $1 - p_1 - p_2 + p \ge (1 - p_1)(1 - p_2) \ge \frac{1}{k^2}$ we get that
   \[\begin{split}
     (k - p_1 k - p_1 k + \ceil{p k})\log \frac{k - p_1 k - p_2 k + \ceil{p k}}{k - p_1 k - p_2 k + p k}
       &\le (k(1 - p_1 - p_2 + p) + 1)\log\left(1 + \frac{1}{k(1 - p_1 - p_2 + p)}\right)
       \\&\le 1 + \log(1 + k) \; .
   \end{split}\]
 
   We easily get that
   \[
     (p_1 k - \ceil{p k})\log \frac{p_1 k - \ceil{p k}}{p_1 k - p k}
       \le (p_1 k - \ceil{p k})\log \frac{p_1 k - p k}{p_1 k - p k}
       = 0 \; .
   \]
   Similarly, we get that $(p_2 k - \ceil{p k})\log \frac{p_2 k - \ceil{p k}}{p_2 k - p k} = 0$.
 
   Combining all this we get that
   \[\begin{split}
     \prp{\sum_{i \in [k]} X^{(i)} = (p_1 k, p_2 k) \wedge \sum_{i \in [k]} X^{(i)}_1 X^{(i)}_2 = \ceil{p k}}
       &\ge \frac{\sqrt{2 \pi}}{e^4} k^{-1.5} \exp(-(1 + \log(1 + k)) - (1 + \log(1 + k)))
       \\&\ge \frac{1}{400} k^{-3.5}\; .
   \end{split}\]
\end{proof}

We now prove a result for the random walk, conditioned on the final position.
In the last result of this section, we will remove those restrictions.
\begin{lemma}\label{lem:rearrangement}
   Let $k \in \mathbb Z_+$ and $p, p_1, p_2 \in [0, 1]$, such that, 
   $p k, p_1 k$, and $p_2 k$ are integers and $p \ge p_1 p_2$. Let $X^{(i)} \in \set{0, 1}^2$ be independent
   identically distributed variables. We then get that
   \[
     \prpcond{\forall l \le k : \sum_{i \in [k]} X^{(i)} \ge \begin{pmatrix}p_1 \\ p_2 \end{pmatrix} l}{\sum_{i \in [k]} X^{(i)} = \begin{pmatrix}p_1 \\ p_2 \end{pmatrix} k \wedge \sum_{i \in [k]} X^{(i)}_1 X^{(i)}_2 = p k}
       \ge k^{-3} \; .
   \]
\end{lemma}

In the proof we will use the folklore result.
\begin{lemma}\label{lem:1-dimensional-rearrangement}
   Let $k \in Z_+$ and $(a_i)_{i \in [k]}$ numbers such that
   $\sum_{i \in [k]} a_i \ge 0$ then there exists a $s \in [k]$ such that
   $\sum_{i \in [l]} a_{(s + i) \mod k} \ge 0$ for every $l \le k$.
\end{lemma}

\begin{proof}[Proof of \Cref{lem:rearrangement}.]
   Using \Cref{lem:1-dimensional-rearrangement} we get that
   $\sum_{i \in [l]} X^{(i)}_1 \ge p_1 l$ for every $l \le k$ with probability
   at least $k^{-1}$ since every variable identically distributed. Fixing $(X^{(i)}_1)_{i \in [k]}$
   and using \Cref{lem:1-dimensional-rearrangement} 2 times we get that
   $\sum_{i \in [l]} X^{(i)}_1 X^{(i)}_2 \ge \frac{p}{p_1} \sum_{i \in [l]} X^{(i)}_1$ and
   $\sum_{i \in [l]} (1 - X^{(i)}_1) X^{(i)} \ge \frac{p_2 - p}{1 - p_1} \sum_{i \in [l]} X^{(i)}_1$
   for every $l \le k$ with probability at least $k^{-2}$. If all these three events happens then for
   every $l \le k$ we get that
   \[\begin{split}
      \sum_{i \in [l]}X^{(i)}_2
         &= \sum_{i \in [l]} X^{(i)}_1 X^{(i)}_2 + \sum_{i \in [l]} (1 - X^{(i)}_1) X^{(i)}_2
         \\&\ge \frac{p}{p_1} \sum_{i \in [l]} X^{(i)}_1 + \frac{p_2 - p}{1 - p_1} \sum_{i \in [l]} X^{(i)}_1
         \\&= \frac{p - p_1 p_2}{p_1(1 - p_1)} \sum_{i \in [l]} X^{(i)}_1 + \frac{p_2 - p}{1 - p_1}l
         \\&\ge \frac{p - p_1 p_2}{p_1(1 - p_1)} p_1 l + \frac{p_2 - p}{1 - p_1}l
         \\&= p_2 l \; .
   \end{split}\]
   So we conclude that with probability at least $k^{-3}$ then
   $\sum_{i \in [l]} X^{(i)} \ge \begin{pmatrix}p_1 \\ p_2 \end{pmatrix} l$ for every $l \le k$
   which finishes the proof.
\end{proof}

All that remains is proving \Cref{lem:quadrant}.
We restate it and then prove it.
\begingroup
\def\thelemma{\ref{lem:quadrant}}
\begin{lemma}
   Let $X_1, \dots, X_k\in\{0,1\}^2$ be iid. Bernoulli $2d$-random variables with probability matrix $\smat{p & p_{1}-p\\p_{2}-p&1-p_1-p_2+p}$.
   Assume that the coordinates are correlated, that is $p\ge p_1p_2$,
   and assume $p_q k$ and $p_2 k$ are integers.

   Let $S_\ell = \sum_{i\in[\ell]} X_i$ be the associated random walk.
   Then
   \begin{align}
      \Pr[\forall \ell\in[k] : S_\ell \le 0] \ge \frac{1}{400\, k^{6.5}}
      .
   \end{align}
\end{lemma}
\addtocounter{lemma}{-1}
\endgroup
\begin{proof}
   We define the set
   $ U = \setbuilder{(x^{(i)})_{i \in [k]} \in \R^{2\times k}} {\forall \ell \le k : \sum_{i \in [\ell]} x^{(i)} \le \svec{p_1\\p_2}\ell}$
   of all sequences satisfying our path requirement.
   In other words
      $\Pr[\forall k\in[n] : S_k \le 0] = \prp{(\mathcal X^{(i)})_{i \in [k]} \in U}$.
   We then add even more restrictions by defining
   \begin{align}
      A' = \setbuilder{(x^{(i)})_{i \in [k]} \in \R^{2\times k}}{
         \sum_{i \in [k]} x^{(i)} = \svec{p_1\\p_2}k
         \wedge \sum_{i \in [k]} (1 - x^{(i)}_1) (1 - x^{(i)}_2) = \ceil{p_{22} k}
      }.
   \end{align}
   That is, we require the last final value of the path to completely match its expectation, rounded up.
   By monotonicity we have
   $\prp{(\mathcal X^{(i)})_{i \in [k]} \in U} \ge \prp{(\tilde{\mathcal X}^{(i)})_{i \in [k]} \in U \cap A'}$.

   We want to use \Cref{lem:hitting-the-mean} and \Cref{lem:rearrangement} and to ease the notation
   we introduce the negated random variables $\mathcal Y^{(i)} = 1 - \tilde{\mathcal X}^{(i)}$.
   Define $p_{22} = 1-p_1-p_2+p$.
   We then have that $\ep{\mathcal Y^{(i)}} = \svec{1-p_1\\1-p_2}$
   and $\prp{\mathcal Y^{(i)} = \left(\begin{smallmatrix}1 \\ 1 \end{smallmatrix}\right)} = p_{22} = 1 - p_1 - p_2 + p \ge (1 - p_1)(1 - p_2)$ by the assumption of correlation.

   We can then rewrite using $\mathcal Y^{(i)}$:
   \[\begin{split}
      \prp{(\tilde{\mathcal X}^{(i)})_{i \in [k]} \in U \cap A'}
      = \prp{\forall \ell \le k : \sum_{i \in [k]} \mathcal Y^{(i)} \ge \svec{1-p_1\\1-p_2}\ell 
         \wedge \sum_{i \in [k]} \mathcal Y^{(i)} = \svec{1-p_1\\1-p_2}k
         \wedge \sum_{i \in [k]} \mathcal Y^{(i)}_1 \mathcal Y^{(i)}_2 = \ceil{p_{22} k}
      }
   \end{split}\]
   Now using \Cref{lem:hitting-the-mean} we have that
   \[\begin{split}
      \prp{\sum_{i \in [k]} \mathcal{Y}^{(i)} = \svec{1-p_1\\1-p_2} k
            \wedge \sum_{i \in [k]} \mathcal{Y}^{(i)}_1 \mathcal{Y}^{(i)}_2 = \ceil{p_{22} k}}
         \ge \frac{1}{400} k^{-3.5}
      \; .
   \end{split}\]
   Combining this with \Cref{lem:rearrangement} we get that
   \[\begin{split}
      &\prp{\forall \ell \le k : \sum_{i \in [k]} \mathcal Y^{(i)} \ge \svec{1-p_1\\1-p_2}\ell 
         \wedge \sum_{i \in [k]} \mathcal Y^{(i)} = \svec{1-p_1\\1-p_2}k
         \wedge \sum_{i \in [k]} \mathcal Y^{(i)}_1 \mathcal Y^{(i)}_2 = \ceil{p_{22} k}
      }
         \\&\phantom{==}= \prp{\sum_{i \in [k]} \mathcal Y^{(i)} = \svec{1-p_1\\1-p_2}k \wedge \sum_{i \in [k]} \mathcal Y^{(i)}_1 \mathcal Y^{(i)}_2 = \ceil{p_{22} k}}
            \\&\phantom{====}\cdot \prpcond{\forall \ell \le k : \sum_{i \in [k]} \mathcal Y^{(i)} \ge \svec{1-p_1\\1-p_2}\ell}
               {\sum_{i \in [k]} \mathcal Y^{(i)} = \svec{1-p_1\\1-p_2}k \wedge \sum_{i \in [k]} \mathcal Y^{(i)}_1 \mathcal Y^{(i)}_2 = \ceil{p_{22} k}}
         \\&\phantom{==}\ge \frac{1}{400} k^{-6.5} \; .
   \end{split}\]

\end{proof}

\section{Lower Bounds}

Our lower bounds all assume that $w_2 d = \omega(\log n)$, where $d$ is the size of the universe.
As discussed in the introduction is both standard and necessary.

We proceed to define the hard distributions for all further lower bounds.
\begin{enumerate}
   \item A query $x \in \zo^d$ is created by sampling $d$ random independent bits with Bernoulli($w_q$) distribution.
   \item A dataset $P \subseteq \zo^d$ is constructed by sampling $n-1$ vectors with random independent bits from such that $y_i \sim \text{Bernoulli}(w_2/w_q)$ if $x_i=1$ and $y_i \sim \text{Bernoulli}((w_u-w_2)/(1-w_q))$ otherwise, for all $y \in P$.
   \item A `close point', $y' \in \zo^d$, is created by $y'_i \sim \text{Bernoulli}(w_1/w_q)$ if $x_i=1$ and $y'_i \sim \text{Bernoulli}((w_u-w_1)/(1-w_q))$ otherwise.
      This point is also added to $P$.
\end{enumerate}
The values are chosen such that $\ep{\abs{x}} = w_q d$, $\ep{\abs{z}} = w_u d$ for all $z \in P$,
$\ep{\abs{x \cap y'}} = w_1 d$, and $\ep{\abs{x \cap y}} = w_2 d$ for all $y \in P \setminus\set{y'}$.
By a union bound over $P$, the actual values are within factors $1 + o(1)$ of their expectations with high probability.
Changing at most $o(\log n)$ coordinates we ensure the weights of queries/database points is exactly their expected value, while only changing the inner products by factors $1 + o(1)$.
Since the changes do not contain any new information, we can assume for lower bounds that entries are independent.
Thus any $(w_q, w_u, w_1(1-o(1)), w_2(1+o(1)))$-GapSS data structure on $P$ must thus be able to return $y'$ with at least constant probability when given the query $x$.

\paragraph{Model}
Our lower bounds are shown in slightly different models.
The first lower bound follows the framework of O'Donnell et al.~\cite{o2014optimal} and Christiani~\cite{christiani2016framework}
and directly lower bound the quantity $\frac{\log(p_1/\min\{p_u,p_q\})}{\log(p_2/\min\{p_u,p_q\})}$ which lower bounds $\rho_u$ and $\rho_q$ in \Cref{defn:lsf}.
This lower bound holds for all $w_2 \neq w_q w_u$, i.e., it gives a lower bound when we are not considering a random instance, and it only
gives a lower bound in the case where $\rho_q = \rho_u$. 

For the second lower bound we follow the framework of Andoni et al.~\cite{andoni2016optimal} and give a lower bound in the ``list-of-points''-model (see \Cref{def:list-of-points}).
This is a slightly more general model, though it is believed that all bounds for the first model can be shown in the list-of-points model as well.
Our lower bound shows that our upper bound is tight in the full time/space trade-off when $w_2 = w_q w_u$, i.e., when we are given a random instance.

The second bound can be extended to show cell probe lower bounds by the arguments in \cite{panigrahy2008geometric}.

\subsection{$p$-biased Analysis}\label{sec:pbias}
We first give some preliminaries on $b$-biased Boolean analysis, and then introduce the directed noise operator.

\subsubsection{Preliminaries}
We want analyse Boolean functions $f : \set{0, 1}^d \to \set{0, 1}$ but as is common, it turns out to be beneficial to consider a more general class of functions $f : \set{0, 1}^d \to \R$.

The probability distribution $\pi_p$ is defined on $\set{0, 1}$ by $\pi_p(1) = p$ and $\pi_p(0) = 1 - p$,
and we define $\pi_p^{\otimes d}$ to be the product probability distribution on $\set{0, 1}^d$. We
write $L_2(\set{0, 1}^d, \pi_p^{\otimes d})$ for the inner product space of functions $f : \set{0, 1}^d \to \R$
with inner product
\begin{align}
   \langle f, g \rangle_p = \ep[x \sim \pi_p^{\otimes d}]{f(x) g(x)} \; .
\end{align}
We will define the norm $\norm{f}{L_q(p)} = \left(\ep[x \sim \pi_p^{\otimes d}]{f(x)^q}\right)^{1/q}$.

We define the $p$-biased Fourier coefficients for a function $f : L_2(\set{0, 1}^d, \pi_p^{\otimes d})$ by
\begin{align}
   \hat{f}^{(p)}(S) = \ep[x \sim \pi_p^{\otimes d}]{f(x) \phi^{(p)}_S(x)} \; ,
\end{align}
for every $S \subseteq [d]$ and where we define
\begin{align}
   \phi^{(p)}(x) &= \frac{x - p}{\sqrt{p(1 - p)}} && \phi^{(p)}_S(x) = \prod_{i \in S} \phi^{(p)}(x_i) \; .
\end{align}
The Fourier coefficients have the nice property that
\begin{align}\label{eq:fourier-expansion}
   f(x) = \sum_{S \subseteq [d]}\hat{f}^{(p)}(S)\phi^{(p)}_S(x) \; .
\end{align}
The Fourier coefficients satisfy the Parseval-Plancherel identity, which says that for any $f, g \in L_2(\set{0, 1}^d, \pi_p^d)$
we have that
\begin{align}
   \langle f, g \rangle_p = \sum_{S \subseteq [d]} \hat{f}^{(p)}(S)\hat{g}^{(p)}(S) \; .
\end{align}
In particular we have that $\ep[x \sim \pi_p^{\otimes d}]{f(x)^2} = \norm{f}{L_2(p)}^2 = \sum_{S \subseteq [d]} \hat{f}^{(p)}(S)^2$.
For Boolean functions $f : \set{0, 1}^d \to \set{0, 1}$ this is particularly useful since we get that
\begin{align}
   \prp[x \sim \pi_p^{\otimes d}]{f(x) = 1}
      &= \ep[x \sim \pi_p^{\otimes d}]{f(x)}
      = \ep[x \sim \pi_p^{\otimes d}]{\sum_{S\subseteq[n]} \hat{f}^{(p)}(S) \phi^{(p)}_S(x)}
      = \hat{f}^{(p)}(\emptyset) \\
   \prp[x \sim \pi_p^{\otimes d}]{f(x) = 1}
      &= \ep[x \sim \pi_p^{\otimes d}]{f(x)}
      = \ep[x \sim \pi_p^{\otimes d}]{f(x)^2}
      = \sum_{S \subseteq [d]} \hat{f}^{(p)}(S)^2
      .
\end{align} 

If we think of $f$ as a filter in a Locality Sensitive data structure, $\prp[x \sim \pi_p^{\otimes d}]{f(x) = 1}$ is the probability that the filter accepts
a random point with expected weight $p$ ($d\cdot p$ of the coordinates being $1$).

\subsubsection{Noise}

For $\rho \in [-1, 1]$, $p_1, p_2 \in (0, 1)$, and $x \in \set{0, 1}^d$ we write $y \sim N_\rho^{p_1 \to p_2}(x)$ when $y \in \set{0, 1}^d$
is randomly chosen such that for each $i \in [d]$ independently, we have that if $x_i \sim \pi_{p_2}$ then $y_i \sim \pi_{p_1}$
and $(x_i, y_i)$ are $\rho$-correlated. We note that if $x \sim \pi_{p_2}^{\otimes d}$ and $y \sim N_\rho^{p_1 \to p_2}$ then
we also have that $y \sim \pi_{p_1}^{\otimes d}$ and $x \sim N_\rho^{p_2 \to p_1}(y)$.

For $\rho \in [-1, 1]$ and $p_1, p_2 \in (0, 1)$ we define the \emph{directed noise operator} $T_\rho^{p_1 \to p_2} : L_2(\set{0, 1}^d, \pi_{p_1}^{\otimes d}) \to L_2(\set{0, 1}^d, \pi_{p_2}^{\otimes d})$
by
\begin{align}
   T_{\rho}^{p_1 \to p_2}f(x) = \ep[y \sim N_\rho^{p_1 \to p_2}]{f(y)} \; .
\end{align}
When $p_1 = p_2 = p$ then $T_{\rho}^{p \to p}$ is the \emph{usual noise operator} on $p$-biased spaces and we denote
it $T_{\rho}^{(p)}$. $T_{\rho}^{(p)}$ has the nice property that $\widehat{T_{\rho}^{(p)} f}^{(p)}(S) = \rho^{\abs{S}}\hat{f}^{(p)}(S)$
for any $S \subseteq [d]$, and hence $T_{\rho}^{(p)}$ satisfies the semigroup property $T_{\rho}^{(p)}T_{\sigma}^{(p)} = T_{\rho \sigma}^{(p)}$.
The following lemma shows that we have similar properties for $T_\rho^{p_1 \to p_2}$.
\begin{lemma}\label{lem:directed-noise}
   For $\rho \in [-1, 1]$, $p_1, p_2 \in (0, 1)$ and $f \in L_2(\set{0, 1}^d, \pi_{p_1}^{\otimes})$ we have that
   \begin{align}
      \widehat{T_{\rho}^{p_1 \to p_2} f}^{(p_2)}(S) = \rho^{\abs{S}}\hat{f}^{(p_1)}(S) \; ,
   \end{align}
   for any $S \subseteq [d]$. Furthermore, for any $\sigma \in [-1, 1]$ and $p_3 \in [0,1]$ we have that
   $T_{\sigma}^{p_2 \to p_3} T_{\rho}^{p_1 \to p_2} = T_{\rho \sigma}^{p_1 \to p_3}$ and $T_{\rho}^{p_2 \to p_1}$ is
   the adjoint of $T_{\rho}^{p_1 \to p_2}$.
\end{lemma}
\begin{proof}
   We fix $S \subseteq [d]$ and get that
   \begin{align}
      \widehat{T_{\rho}^{p_1 \to p_2} f}^{(p_2)}(S)
         &= \ep[x \sim \pi_{p_2}^{\otimes d}]{T_{\rho}^{p_1 \to p_2} f(x) \phi^{(p_2)}_S(x)}
         \\&= \ep[x \sim \pi_{p_2}^{\otimes d}]{\ep[y \sim N_{\rho}^{p_1 \to p_2}(x)]{f(y)}\phi^{(p_2))}_S(x)}
         \\&= \ep[x \sim \pi_{p_2}^{\otimes d}]{\ep[y \sim N_{\rho}^{p_1 \to p_2}(x)]{\sum_{T \subseteq [d]} \widehat{f}^{(p_1)}(T) \phi^{(p_1)}_T(y)}\phi^{(p_2)}_S(x)}
         \\&= \ep[x \sim \pi_{p_2}^{\otimes d}]{\ep[y \sim N_{\rho}^{p_1 \to p_2}(x)]{\widehat{f}^{(p_1)}(S) \phi^{(p_1)}_S(y)\phi^{(p_2)}_S(x)}}
         \\&= \widehat{f}^{(p_1)}(S) \prod_{i \in S} \ep[x_i \sim \pi_{p_2}]{\ep[y_i \sim N_{\rho}^{p_1 \to p_2}]{\phi^{(p_1)}_i(y_i)\phi^{(p_2)}_i(x_i)}}
         \\&= \rho^{\abs{S}} \widehat{f}^{(p_1)}(S) \; ,
   \end{align}
   where the last line uses that $\phi_i^{(p)}(x) = \frac{x - p}{\sqrt{p(1 - p)}}$, which proves the first claim.
   For the second claim we note that
   \begin{align}
   \widehat{( T_{\sigma}^{p_2 \to p_3} T_{\rho}^{p_1 \to p_2}f)}^{(p_3)}(S)
      = \sigma^{\abs{S}}\widehat{T_\rho^{p_1\to p_2}f}^{(p_2)}(S)
      = (\rho \sigma)^{\abs{S}}\widehat{f}^{(p_1)}(S)
      = \widehat{T_{\rho \sigma}^{p_1 \to p_2}f}^{(p_3)}(S) \; ,
   \end{align}
   for any $f \in f \in L_2(\set{0, 1}^d, \pi_{p_1}^{\otimes})$ and any $S \subseteq [d]$
   which proves the second claim.
   For the last claim we use the Plancherel-Parseval identity and get that
   \begin{align}
      \langle T_\rho^{p_1 \to p_2} f, g \rangle_{L_2(p_2)}
         = \sum_{S \in [d]} \rho^{\abs{S}} \widehat{f}^{(p_1)} \widehat{g}^{(p_2)}
         = \langle f, T_\rho^{p_2 \to p_1} g \rangle_{L_2(p_1)} \; ,
   \end{align}
   for any $f \in L_2(\set{0, 1}^d, \pi_{p_1}^{\otimes d})$ and any $g \in L_2(\set{0, 1}^d, \pi_{p_2}^{\otimes d})$
   which shows that $T_\rho^{p_2 \to p_1}$ is the adjoint of $T_\rho^{p_1 \to p_2}$.
\end{proof}

We say that $(T^{p_1 \to p_2}_\rho)_{\rho > 0}$ is $(s, r)$-hypercontractive if there exists
$\rho_0 > 0$ such that for every $\rho \ge \rho_0$ and every $f \in L_r(\set{0, 1}^d, \pi_{p_1}^d)$
\begin{align}
   \norm{T^{p_1 \to p_2}_\rho f}{L_s(p_2)} \le \norm{f}{L_r(p_1)} \; .
\end{align}
We define $\sigma_{s, r}(p_1, p_2)$ to be the smallest possible $\rho_0$
We are interested in the hypercontractivity of $T^{p_1 \to p_2}$

\subsection{Symmetric Lower bound}

The most general, but sadly least tractable, approach to our lower bounds, is to bound the noise operator $T_\alpha$ in terms of a different level of noise, $T_\beta$.
We do however manage to show one bound on this type, following an spectral approach first used by O'Donnell et al.~\cite{o2014optimal} to prove the first optimal LSH lower bounds of $\rho\ge1/c$ for data-independent hashing.
Besides handling the case of set similarity with filters rather than hash functions, we slightly generalize the approach a big by using the power-means inequality rather than log-concavity.
\footnote{This widens the range in which the bound is applicable -- the O'Donnell bound is only asymptotic for $r\to 0$. However the values we obtain outside this range, when applied to Hamming space LSH, aren't sharp against the upper bounds.}

We will show the following inequality
\begin{align}
   \left(\frac{\Pr_{x,y',f}[f(x)=1,f(y')=1]}{\Pr_{x,f}[f(x)=1] }\right)^{1/\log\alpha}
   \le
   \left(\log\frac{\Pr_{x,y,f}[f(x)=1,f(y)=1]}{\Pr_{x,f}[f(x)=1]}\right)^{1/\log\beta}
\end{align}
where $\alpha = \frac{w_1-w^2}{w(1-w)}$
and $\beta = \frac{w_2-w^2}{w(1-w)}$,
and $y'$ and $y$ are sampled as respectively a close and a far point (see the top of the section).
By rearrangement, this directly implies a lower bound in the LSF model as defined in \Cref{defn:lsf}.

First we prove a general lemma about Boolean functions, which contains the most important arguments.
\begin{lemma}\label{lem:compare-correlation}
   Let $f : \set{0, 1}^n \to \R$ be a function and $p \in (0, 1)$. Then for any $1 > \alpha \ge \beta > 0$ we have that
   \begin{align}
      \left( \frac{\langle T^{(p)}_\alpha f, f \rangle_{L_2(p)}}{\norm{f}{L_2(p)}^2} \right)^{1/\log(1/\alpha)} 
         \le \left( \frac{\langle T^{(p)}_\beta f, f \rangle_{L_2(p)}}{\norm{f}{L_2(p)}^2} \right)^{1/\log(1/\beta)} \; .
   \end{align}
\end{lemma}
\begin{proof}
   We use the Parseval-Plancherel identity and the power-mean inequality to get that 
   \begin{align}
      \left( \frac{\langle T^{(p)}_\alpha f, f \rangle_{L_2(p)}}{\norm{f}{L_2(p)}^2} \right)^{1/\log(1/\alpha)}
         &= \left( \frac{\sum_{S \subseteq [n]} \alpha^{\abs{S}} \hat{f}^{(p)}(S)^2 }{\sum_{S \subseteq [n]} \hat{f}^{(p)}(S)^2} \right)^{1/\log(1/\alpha)}
         \\&= \left( \sum_{k = 0}^n \frac{\sum_{\substack{S \subseteq [n] \\ \abs{S} = k}} \hat{f}^{(p)}(S)^2}{\sum_{S \subseteq [n]} \hat{f}^{(p)}(S)^2} \left(e^{-k} \right)^{\log(1/\alpha)} \right)^{1/\log(1/\alpha)}
         \\&\le \left( \sum_{k = 0}^n \frac{\sum_{\substack{S \subseteq [n] \\ \abs{S} = k}} \hat{f}^{(p)}(S)^2}{\sum_{S \subseteq [n]} \hat{f}^{(p)}(S)^2} \left(e^{-k} \right)^{\log(1/\beta)} \right)^{1/\log(1/\beta)}
         \\&= \left( \frac{\sum_{S \subseteq [n]} \beta^{\abs{S}} \hat{f}^{(p)}(S)^2 }{\sum_{S \subseteq [n]} \hat{f}^{(p)}(S)^2} \right)^{1/\log(1/\beta)}
         \\&= \left( \frac{\langle T^{(p)}_\beta f, f \rangle_{L_2(p)}}{\norm{f}{L_2(p)}^2} \right)^{1/\log(1/\beta)} \; .
   \end{align}
   The first and the last equality follows from the Parseval-Plancherel identity and the inequality follows from the power-mean
   inequality since $\log(1/\alpha) \le \log(1/\beta)$.
\end{proof}
The proof of \Cref{thm:donnell} is then simply a few a rearrangements such that we can use \Cref{lem:compare-correlation}.
\begin{corollary}
   \label{lemma:donnell}
   Any data-independent LSF data structure for the $(w, w, w_1, w_2)$-GapSS problem with expected query time $n^{\rho_q}$
   and expected space usage $n^{1 + \rho_u}$ where $\rho_q=\rho_u=\rho$ must have
   \begin{align}
      \rho \ge \log\left(\frac{w_1-w^2}{w(1-w)}\right)
      \bigg/ \log\left(\frac{w_2-w^2}{w(1-w)}\right)
      \; .
   \end{align}
\end{corollary}
\begin{proof}
   Let $\mathcal F$ be any fixed LSF-family and let $f: \set{0, 1}^n \to \set{0, 1}$ be a random function
   such that $f^{-1}(1) = Q$ for $Q \sim \mathcal F$. Now we define the deterministic function
   $\overline{f}: \set{0, 1}^n \to \R$ by $\overline{f}(x) = \sum_{S \subseteq [d]} \sqrt{\ep[f]{\hat{f}^{(w)}(S)^2}} \phi_S(x)$.
   Using the Parseval-Plancherel identity we get that
   \begin{align}
      \ep[f]{\langle T^{(w)}_\rho f, f \rangle_{L_2(w)}}
         = \sum_{S \subseteq [d]} \rho^{\abs{S}} \ep[f]{\hat{f}^{(w)}(S)^2}
         = \langle T^{(w)}_\rho \overline{f}, \overline{f} \rangle_{L_2(w)} \; .
   \end{align}
   for every $\rho$. We set $\alpha = \frac{w_1-w^2}{w(1-w)}$
   and $\beta = \frac{w_2-w^2}{w(1-w)}$ and note that
   $\Pr_{x,y',f}[f(x)=1,f(y')=1] = \ep[f]{\langle T^{(w)}_\alpha f, f \rangle_{L_2(w)}}$,
   $\Pr_{x,y,f}[f(x)=1,f(y)=1] = \ep[f]{\langle T^{(w)}_\beta f, f \rangle_{L_2(w)}}$, and
   $\Pr_{x,f}[f(x)=1] = \ep[f]{\norm{f}{L_2(w)}^2}$.
   Then using \Cref{lem:compare-correlation} we get that
   \begin{align}
      \left(\frac{\Pr_{x,y',f}[f(x)=1,f(y')=1]}{\Pr_{x,f}[f(x)=1] }\right)^{1/\log 1/\alpha}
         &= \left( \frac{\ep[f]{\langle T^{(w)}_\alpha f, f \rangle_{L_2(w)}}}{\ep[f]{\norm{f}{L_2(w)}^2 }} \right)^{1 / \log 1/\alpha}
         \\&= \left( \frac{\langle T^{(w)}_\alpha \overline{f}, \overline{f} \rangle_{L_2(w)}}{\norm{\overline{f}}{L_2(w)}^2} \right)^{1 / \log 1/\alpha}
         \\&\le \left( \frac{\langle T^{(w)}_\beta \overline{f}, \overline{f} \rangle_{L_2(w)}}{\norm{\overline{f}}{L_2(w)}^2} \right)^{1 / \log 1/\beta}
         \\&= \left(\frac{\Pr_{x,y,f}[f(x)=1, f(y)=1]}{\Pr_{x,f}[f(x)=1] }\right)^{1/\log 1/\beta}
            \; .
   \end{align}
   By rearrangement this implies that
   \begin{align}
      \rho = \frac{\log \frac{\Pr_{x,y',f}[f(x)=1,f(y')=1]}{\Pr_{x,f}[f(x)=1] }}{\log \frac{\Pr_{x,y,f}[f(x)=1,f(y)=1]}{\Pr_{x,f}[f(x)=1] }}
         \ge \frac{\log\alpha}{\log\beta} \; .
   \end{align}
\end{proof}

As noted the bound is sharp against our upper bound when $w_u,w_q,w_1,w_2$ are all small.
Also notice that $\log\alpha/\log\beta\le\frac{1-\alpha}{1+\alpha}\frac{1-\beta}{1+\beta}$ is a rather good approximation for $\alpha$ and $\beta$ close to 1.
Here the right hand side is the $\rho$ value of Spherical LSH with the batch-normalization embedding discussed in \Cref{sec:embed}.

Note that the lower bound becomes 0 when we get close to the random instance, $w_2\to w_qw_u$.
In the next sections we will remedy this, by showing a lower bound tight exactly when $w_2=w_qw_u$.

\subsection{General Lower Bound}

Our second lower bound will be proven in the ``list-of-points'' model. We follow and expand upon the
approach by Andoni et al.~\cite{andoni2016optimal}. The main idea is to lower bound random instances
with planted points. If the random instances correspond to a Similarity Search problem with high
probability then we have a lower bound for the Similarity Search problem. We formalize the notion
of random instances in the following general definition.

\begin{definition}[Random instance]
   For spaces $Q$ and $U$ we describe a distribution 
   of dataset-query pairs $(P, q)$ where $P \subseteq U$ and $q \in Q$.
   Let $\mathcal{P}_{QU}$ be a probability distribution on $Q \times U$, a $\mathcal{P}_{QU}$-random instance is
   a dataset-query pair drawn from the following distribution.
   \begin{enumerate}
      \item A dataset $P \subseteq U$ is constructed by sampling $n$ points where $p \sim \mathcal{P}_U$ for all $p \in P$.
      \item A dataset point $p' \in P$ is fixed and a $q \in Q$ is sampled such that $(q, p') \sim \mathcal{P}_{QU}$.
      \item The goal of the data structure is to preprocess $P$ such that it recovers $p'$ when given the query point $q$.
   \end{enumerate}
\end{definition}

We can then generalize the result by Andoni et al.~\cite{andoni2016optimal}, who proved a result
specifically for random Hamming instances, to general random instances. We defer the proof to \Cref{app:proof-hypercontractive-to-lower-bound}.
\begin{lemma}\label{lem:hypercontractive-to-lower-bound}
   Let $Q$ and $U$ be some spaces and $\mathcal{P}_{QU}$ a probability distribution on $Q \times U$.
   Consider any list-of-points data structure for $\mathcal{P}_{QU}$-random instances of $n$ points,
   which uses expected space $n^{1 + \rho_u}$, has expected query time $n^{\rho_q - o_n(1)}$, 
   and succeeds with probability at least $0.99$. Let $r, s \in [1, \infty]$
   satisfy  
   \begin{align}
      \ep[(X, Y) \sim \mathcal{P}_{QU}]{f(X) g(Y)} \le \norm{f(X)}{L_r(\mathcal{P}_Q)} \norm{g(Y)}{L_s(\mathcal{P}_U)} \; ,
   \end{align}
   for all functions $f : Q \to \R$ and $g : U \to \R$. Then
   \begin{align}
      \frac{1}{r}\rho_q + \frac{1}{r'}\rho_u \ge \frac{1}{r} + \frac{1}{s} - 1 \; ,
   \end{align}
   where $r' = \frac{r}{r - 1}$ is the convex conjugate of $r$.
\end{lemma}

This gives a good way to lower bound random instances when one has tight hypercontractive inequalities.
Unfortunately, for most probability distributions this is not the case but we can amplify the power of \Cref{lem:hypercontractive-to-lower-bound}
by combining it with \Cref{lem:hypercontractive-to-divergence} which we recall from the introduction.
\begingroup
\def\thelemma{\ref*{lem:correspondence-hypercontractive-divergence}}
\begin{lemma}\label{lem:hypercontractive-to-divergence}
   Let $\mathcal{P}_{XY}$ be a probability distribution on a space $\Omega_X \times \Omega_Y$ and let $\mathcal{P}_X$ and $\mathcal{P}_Y$ be the
   marginal distributions on the spaces $\Omega_X$ and $\Omega_Y$ respectively. Let $s, r \in [1, \infty)$, then the following is equivalent
   \begin{enumerate}
      \item For all functions $f : \Omega_X \to \R$ and $g : \Omega_Y \to \R$ we have
      \begin{align}\label{eq:general-hypercontractive}
         \ep[(X, Y) \sim \mathcal{P}_{XY}]{f(X) g(Y)} \le \norm{f(X)}{L_r(\mathcal{P}_X)} \norm{g(Y)}{L_s(\mathcal{P}_Y)} \; .
      \end{align}
      \item For all probability distributions $\mathcal{Q}_{XY}$ which are absolutely continuous with respect to $\mathcal{P}_{XY}$ we have
      \begin{align}\label{eq:general-divergence}
         \D{\mathcal{Q}_{XY}}{\mathcal{P}_{XY}} \ge \frac{\D{\mathcal{Q}_X}{\mathcal{P}_X}}{r} + \frac{\D{\mathcal{Q}_Y}{\mathcal{P}_Y}}{s} \; .
      \end{align}
   \end{enumerate}
\end{lemma}
\addtocounter{lemma}{-1}
\endgroup

We defer the proof to the end of the \hyperlink{prf:hypercontractive-to-divergence}{section}
and instead start by focusing on the effects of combining \Cref{lem:hypercontractive-to-lower-bound}
and \Cref{lem:hypercontractive-to-divergence}. First of all we can prove the following general lower bound
for random instances.
\begin{theorem}\label{thm:lower-bound-random-instance}
   Let $Q$ and $U$ be some spaces and $\mathcal{P}_{QU}$ a probability distribution on $Q \times U$.
   Consider any list-of-points data structure for $\mathcal{P}_{QU}$-random instances of $n$ points,
   which uses expected space $n^{1 + \rho_u}$, has expected query time $n^{\rho_q - o_n(1)}$, 
   and succeeds with probability at least $0.99$. Then for every
   $r \in [1, \infty]$ we have that
   \begin{align}
      \frac{1}{r}\rho_q + \frac{1}{r'}\rho_u \ge \inf_{\mathcal{Q}_{QU}}\left(
         \frac{1}{r}\frac{\D{\mathcal{Q}_{QU}}{\mathcal{P}_{QU}} - \D{\mathcal{Q}_Q}{\mathcal{P}_Q}}{\D{\mathcal{Q}_U}{\mathcal{P}_U}}
         + \frac{1}{r'}\frac{\D{\mathcal{Q}_{QU}}{\mathcal{P}_{QU}} - \D{\mathcal{Q}_U}{\mathcal{P}_U}}{\D{\mathcal{Q}_U}{\mathcal{P}_U}}
      \right) \; ,
   \end{align}
   where $r' = \frac{r}{r - 1}$ is the convex conjugate of $r$ and the infimum is over every probability distribution $\mathcal{Q}_{QU}$
   with $\mathcal{Q}_U \neq \mathcal{P}_U$ and which is absolutely continuous with respect to $\mathcal{P}_{QU}$.
\end{theorem}
\begin{proof}
   Let $r \in [1, \infty]$ and choose $s = \arginf\setbuilder{s \in [1, \infty]}{\text{$\mathcal{P}_{QU}$ is $(r, s)$-hypercontractive}}$.
   \Cref{lem:hypercontractive-to-lower-bound} give us that
   \begin{align}\label{eq:simple-lower-bound}
      \frac{1}{r}\rho_q + \frac{1}{r'}\rho_u \ge \frac{1}{r} + \frac{1}{s} - 1 \; .
   \end{align}
   \Cref{lem:hypercontractive-to-divergence} give us that
   \begin{align}
      \D{\mathcal{Q}_{XY}}{\mathcal{P}_{XY}} 
         &\ge \frac{\D{\mathcal{Q}_X}{\mathcal{P}_X}}{r} + \frac{\D{\mathcal{Q}_Y}{\mathcal{P}_Y}}{s}
   \end{align}
   for every $\mathcal{Q}_{QU}$  with $\mathcal{Q}_U \neq \mathcal{P}_U$ and which is absolutely continuous with respect to $\mathcal{P}_{QU}$.
   We can rewrite this as
   \begin{align}
      \frac{1}{r}\D{\mathcal{Q}_{QU}}{\mathcal{P}_{QU}} - \D{\mathcal{Q}_Q}{\mathcal{P}_Q}
         + \frac{1}{r'}\D{\mathcal{Q}_{QU}}{\mathcal{P}_{QU}} - \D{\mathcal{Q}_U}{\mathcal{P}_U}
         &\ge \left(\frac{1}{s} - \frac{1}{r'}\right)\D{\mathcal{Q}_U}{\mathcal{P}_U}
            &&\Leftrightarrow\\
      \frac{1}{r}\frac{\D{\mathcal{Q}_{QU}}{\mathcal{P}_{QU}} - \D{\mathcal{Q}_Q}{\mathcal{P}_Q}}{\D{\mathcal{Q}_U}{\mathcal{P}_U}}
         + \frac{1}{r'}\frac{\D{\mathcal{Q}_{QU}}{\mathcal{P}_{QU}} - \D{\mathcal{Q}_U}{\mathcal{P}_U}}{\D{\mathcal{Q}_U}{\mathcal{P}_U}}
         &\ge \frac{1}{s} - \frac{1}{r'} = \frac{1}{s} + \frac{1}{r} - 1
   \end{align}
   Now the minimality of $s$ give us that
   \begin{align}\label{eq:divergence-infimum}
      \inf_{\mathcal{Q}_{QU}}\left( 
         \frac{1}{r}\frac{\D{\mathcal{Q}_{QU}}{\mathcal{P}_{QU}} - \D{\mathcal{Q}_Q}{\mathcal{P}_Q}}{\D{\mathcal{Q}_U}{\mathcal{P}_U}}
         + \frac{1}{r'}\frac{\D{\mathcal{Q}_{QU}}{\mathcal{P}_{QU}} - \D{\mathcal{Q}_U}{\mathcal{P}_U}}{\D{\mathcal{Q}_U}{\mathcal{P}_U}}   
      \right) = \frac{1}{s} + \frac{1}{r} - 1 \; ,
   \end{align}
   where the infimum is over every probability distribution $\mathcal{Q}_{QU}$
   with $\mathcal{Q}_U \neq \mathcal{P}_U$ and which is absolutely continuous with respect to $\mathcal{P}_{QU}$.
   Now combining \eqref{eq:simple-lower-bound} and \eqref{eq:divergence-infimum} give us the result.
\end{proof}

Combining the lemma with the ``Hypercontractive Induction Theorem''~\cite{o2014analysis} we can prove \Cref{thm:lower_main}.
\begin{lemma}\label{lem:hypercontractive-induction}
   Let $\mathcal{P}_{XY}$ be a probability distribution on a space $\Omega_X \times \Omega_Y$
   and $\mathcal{P}_{XY}^{\otimes n}$ be a probability distribution consisting $n$ independent copies 
   of $\mathcal{P}_{XY}$. Then $\mathcal{P}_{XY}$ is $(r, s)$-hypercontractive if and only if
   $\mathcal{P}_{XY}^{\otimes n}$ is $(r, s)$-hypercontractive.
\end{lemma}

We restate \Cref{thm:lower_main} and prove it.
\begingroup
\def\thetheorem{\ref*{thm:lower_main}}
\begin{theorem}
   Consider any list-of-point data structure for the $(w_q, w_u, w_1, w_q w_u)$-GapSS problem over a universe
   of size $d$ of $n$ points with $w_q w_u d = \omega(\log n)$, which uses expected space $n^{1 + \rho_u}$,
   has expected query time $n^{\rho_q - o_n(1)}$, 
   and succeeds with probability at least $0.99$. Then for every
   $\alpha \in [0, 1]$ we have that
   \begin{align}
        \alpha \rho_q + (1 - \alpha)\rho_u 
        \ge \inf_{\substack{
           t_q, t_u \in [0,1]
           \\ t_u \neq w_u
        }}
        \left(
            \alpha \frac{\D{T}{P}-\dq}{\du}
            + (1 - \alpha) \frac{\D{T}{P}-\du}{\du}
        \right) \;,
   \end{align}
   where
   $P = \smat{w_1 && w_q-w_1 \\ w_u-w_1 && 1-w_q-w_u+w_1}$
   and
   $T = 
        \underset{\substack{
           T \ll P
        , \underset{X\sim T}{E}[X] = \svec{t_q\\t_u}
  }}{\arginf}\D{T}{P}$.
\end{theorem}
\addtocounter{theorem}{-1}
\endgroup

\begin{proof}
   From the discussion at the beginning of the section it is enough to lower bound the
   $P^{\otimes d}$-random instance where $P = \text{Bernoulli}([\begin{smallmatrix}
      w_1 & w_q - w_1 \\
      w_u & 1 - w_q - w_u + w_1
   \end{smallmatrix}])$, since this will imply a lower bound for the $(w_q, w_u, w_1, w_qw_u)$-GapSS problem.
   Combining \Cref{lem:hypercontractive-induction} and \Cref{lem:hypercontractive-to-divergence} we get
   that $P^{\otimes d}$ is $(r, s)$-hypercontractive if and only if
   $\D{T}{P} \ge \frac{\di{t_q}{w_q}}{r} + \frac{\di{t_u}{w_u}}{s}$
   where $T = \underset{\substack{
         T \ll P
      , \underset{X\sim T}{E}[X] = \svec{t_q\\t_u}
   }}{\arginf}\D{T}{P}$. Now repeating the proof of \Cref{thm:lower-bound-random-instance} give us the result.
\end{proof}

\hypertarget{prf:hypercontractive-to-divergence}{}
\paragraph{Proof of \Cref{lem:hypercontractive-to-divergence}} We now turn to the proof \Cref{lem:hypercontractive-to-divergence}.
The main argument needed in the proof of is contained in the following lemma, which can be seen as a variation of Fenchel's inequality.
\begin{lemma}\label{lem:divergence-legendre}
   Let $P$ be a probability distribution on a space $\Omega$, $Q$ a probability which is absolutely continuous with respect to $P$,
   and $\phi : \Omega \to \R$ a function such that $\ep[P]{\exp(\phi(X))} \le \infty$. Then
   \begin{align}
       \D{Q}{P} + \log \ep[X \sim P]{\exp(\phi(X))}
           \ge \ep[X \sim Q]{\phi(X)} \; .
   \end{align}
   and we have equality if and only if $\frac{d Q}{d P}(x) = \frac{\exp(\phi(x))}{\ep[X \sim P]{\exp(\phi(X))}}$.
\end{lemma}
\begin{proof}
   To ease notation we write $p = \frac{d Q}{d P}$.
   We note that
   \begin{align}
       \D{Q}{P} 
           = \ep[X \sim Q]{\log p(X)}
           = \ep[X \sim Q]{\log \frac{p(X)}{\exp(\phi(X))}} + \ep[X \sim Q]{\phi(X)}
           = \ep[X \sim Q]{\phi(X)} - \ep[X \sim Q]{\log \frac{\exp(\phi(X))}{p(X)}} \; .
   \end{align}
   Using Jensen's inequality we get that
   \begin{align}
       \ep[X \sim Q]{\log \frac{\exp(\phi(x))}{p(X)}}
           \le \log \ep[X \sim Q]{\exp(\phi(x)) \frac{d P}{d Q}(x)}
           = \log \ep[X \sim P]{\exp(\phi(x))} \; .
   \end{align}
   Combining these two equations give us the inequality. Now we note that we have equality if and only
   if $e^{\phi(x)} \frac{d P}{d Q}(x)$ is constant, and since $Q$ is a probability distribution this
   is equivalent with $\frac{d Q}{d P}(x) = \frac{\exp(\phi(x))}{\ep[X \sim P]{\exp(\phi(X))}}$.
\end{proof}

We are now ready to prove \Cref{lem:hypercontractive-to-divergence}.
\begin{proof}[Proof of \Cref{lem:hypercontractive-to-divergence}]
   $\eqref{eq:general-hypercontractive} \Rightarrow \eqref{eq:general-divergence}$. Let $\mathcal{Q}_{XY}$ be a probability distribution which is
   absolutely continuous with respect to $\mathcal{P}_{XY}$. We set $\exp(\phi_X(x)) = \frac{d \mathcal{Q}_X}{d \mathcal{P}_X}(x)$ and $\exp(\phi_Y(y)) = \frac{d \mathcal{Q}_Y}{d \mathcal{P}_Y}(y)$.
   From this we see that $\ep[X \sim \mathcal{P}_X]{\exp(\phi_X(X))} = \ep[X \sim \mathcal{P}_x]{\frac{d \mathcal{Q}_X}{d \mathcal{P}_X}(X)} = \ep[X \sim \mathcal{Q}_X]{1} = 1$
   and similarly that $\ep[Y \sim \mathcal{P}_Y]{\exp(\phi_X(X))} = 1$, hence we have that
   $\frac{d \mathcal{Q}_X}{d \mathcal{P}_X}(x) = \frac{\exp(\phi_X(x))}{\ep[X \sim \mathcal{P}_X]{\exp(\phi_X(X))}}$ and
   $\frac{d \mathcal{Q}_Y}{d \mathcal{P}_Y}(y) = \frac{\exp(\phi_Y(y))}{\ep[Y \sim \mathcal{P}_Y]{\exp(\phi_X(Y))}}$.
   Using \eqref{eq:general-hypercontractive} we get that
   \begin{align}
       \ep[(X, Y) \sim \mathcal{P}_{XY}]{\exp(\phi_X(X) + \phi_Y(Y)} &\le \ep[X \sim \mathcal{P}_X]{\exp(r\phi_X(X))}^{1/r} \ep[Y \sim \mathcal{P}_Y]{\exp(s\phi_Y(Y)}^{1/s}
           &&\Leftrightarrow\\
       \log \ep[(X, Y) \sim \mathcal{P}_{XY}]{\exp(\phi_X(X) + \phi_Y(Y)} &\le \frac{\log \ep[X \sim \mathcal{P}_X]{\exp(r\phi_X(X))}}{r} + \frac{\log \ep[Y \sim \mathcal{P}_Y]{\exp(s\phi_Y(Y)}}{s} \; .
   \end{align}
   Using \Cref{lem:divergence-legendre} 3 times we have that 
   \begin{align}
       \log \ep[(X, Y) \sim \mathcal{P}_{XY}]{\exp(\phi_X(X) + \phi_Y(Y))} &\ge  \ep[(X, Y) \sim \mathcal{Q}_{XY}]{\phi_X(X) + \phi_Y(Y)} - \D{\mathcal{Q}_{XY}}{\mathcal{P}_{XY}} \\
       \log \ep[X \sim \mathcal{P}_{X}]{\exp(\phi_X(X))} &= \ep[X \sim \mathcal{Q}_{X}]{\phi_X(X)} - \D{\mathcal{Q}_{X}}{\mathcal{P}_{X}} \\
       \log \ep[X \sim \mathcal{P}_{Y}]{\exp(\phi_Y(Y))} &= \ep[Y \sim \mathcal{Q}_{Y}]{\phi_Y(Y)} - \D{\mathcal{Q}_{Y}}{\mathcal{P}_{Y}} \; ,
   \end{align}
   where the equalities hold since $\frac{d \mathcal{Q}_X}{d \mathcal{P}_X}(x) = \frac{\exp(\phi_X(x))}{\ep[X \sim \mathcal{P}_X]{\exp(\phi_X(X))}}$ and $\frac{d \mathcal{Q}_Y}{d \mathcal{P}_Y}(y) = \frac{\exp(\phi_Y(y))}{\ep[Y \sim \mathcal{P}_Y]{\exp(\phi_X(Y))}}$.
   We then get that
   \begin{align}
       &\log \ep[(X, Y) \sim \mathcal{P}_{XY}]{\exp(\phi_X(X) + \phi_Y(Y)} \le \frac{\log \ep[X \sim \mathcal{P}_X]{\exp(r\phi_X(X))}}{r} + \frac{\log \ep[Y \sim \mathcal{P}_Y]{\exp(s\phi_Y(Y)}}{s} &&\Rightarrow\\
       &\ep[(X, Y) \sim \mathcal{Q}_{XY}]{\phi_X(X) + \phi_Y(Y)} - \D{\mathcal{Q}_{XY}}{\mathcal{P}_{XY}}
           \\&\quad\le \ep[X \sim \mathcal{Q}_{X}]{\phi_X(X)} - \D{\mathcal{Q}_{X}}{\mathcal{P}_{X}} + \ep[Y \sim \mathcal{Q}_{Y}]{\phi_Y(Y)} - \D{\mathcal{Q}_{Y}}{\mathcal{P}_{Y}} &&\Leftrightarrow\\
       &\D{\mathcal{Q}_{XY}}{\mathcal{P}_{XY}} \ge \frac{\D{\mathcal{Q}_X}{\mathcal{P}_X}}{r} + \frac{\D{\mathcal{Q}_Y}{\mathcal{P}_Y}}{s} \; ,
   \end{align}
   which proves that $\eqref{eq:general-hypercontractive} \Rightarrow \eqref{eq:general-divergence}$.

   $\eqref{eq:general-divergence} \Rightarrow \eqref{eq:general-hypercontractive}$. Fix the functions $f : \Omega_X \to \R$
   and $g : \Omega_Y \to \R$. We note that $\ep[(X, Y) \sim \mathcal{P}_{XY}]{f(X) g(Y)} \le \ep[(X, Y) \sim \mathcal{P}_{XY}]{\abs{f}(X) \abs{g}(Y)}$
   hence we can assume that $f$ and $g$ are non-negative. We define $\phi_X(x) = \log(f(x))$ and $\phi_Y(x) = \log(g(x))$\footnote{We define $\log(0) = -\infty$ and $\exp(-\infty) = 0$.}.
   Then \eqref{eq:general-hypercontractive} is equivalent with
   \begin{align}
       \ep[(X, Y) \sim \mathcal{P}_{XY}]{\exp(\phi_X(X) + \phi_Y(Y)} &\le \ep[X \sim \mathcal{P}_X]{\exp(r\phi_X(X))}^{1/r} \ep[Y \sim \mathcal{P}_Y]{\exp(s\phi_Y(Y)}^{1/s}
           &&\Leftrightarrow\\
       \log \ep[(X, Y) \sim \mathcal{P}_{XY}]{\exp(\phi_X(X) + \phi_Y(Y)} &\le \frac{\log \ep[X \sim \mathcal{P}_X]{\exp(r\phi_X(X))}}{r} + \frac{\log \ep[Y \sim \mathcal{P}_Y]{\exp(s\phi_Y(Y)}}{s} \; .
   \end{align}
   We define the probability distribution $\mathcal{Q}_{XY}$ by $\frac{d \mathcal{Q}_{XY}}{d \mathcal{P}_{XY}}(x, y) = \frac{\exp(\phi_X(X) + \phi_Y(Y)}{\ep[(X, Y) \sim \mathcal{P}_{XY}]{\exp(\phi_X(X) + \phi_Y(Y)}}$.
   It is easy to see that $\mathcal{Q}_{XY}$ is indeed a probability distribution. Using \eqref{eq:general-divergence} we get that
   \begin{align}
       \D{\mathcal{Q}_{XY}}{\mathcal{P}_{XY}} \ge \frac{\D{\mathcal{Q}_X}{\mathcal{P}_X}}{r} + \frac{\D{\mathcal{Q}_Y}{\mathcal{P}_Y}}{s} \; .
   \end{align}
   Using \Cref{lem:divergence-legendre} 3 times we have that 
   \begin{align}
       \D{\mathcal{Q}_{XY}}{\mathcal{P}_{XY}} &= \ep[(X, Y) \sim \mathcal{Q}_{XY}]{\phi_X(X) + \phi_Y(Y)} - \log \ep[(X, Y) \sim \mathcal{P}_{XY}]{\exp(\phi_X(X) + \phi_Y(Y))} \\
       \D{\mathcal{Q}_{X}}{\mathcal{P}_{X}} &\ge \ep[X \sim \mathcal{Q}_{X}]{\phi_X(X)} - \log \ep[X \sim \mathcal{P}_{X}]{\exp(\phi_X(X))} \\
       \D{\mathcal{Q}_{Y}}{\mathcal{P}_{Y}} &\ge \ep[Y \sim \mathcal{Q}_{Y}]{\phi_Y(Y)} - \log \ep[X \sim \mathcal{P}_{Y}]{\exp(\phi_Y(Y))} \; ,
   \end{align}
   where the equality holds since $\frac{d \mathcal{Q}_{XY}}{d \mathcal{P}_{XY}}(x, y) = \frac{\exp(\phi_X(X) + \phi_Y(Y)}{\ep[(X, Y) \sim \mathcal{P}_{XY}]{\exp(\phi_X(X) + \phi_Y(Y)}}$.
   We then get that
   \begin{align}
       &\D{\mathcal{Q}_{XY}}{\mathcal{P}_{XY}} \ge \frac{\D{\mathcal{Q}_X}{\mathcal{P}_X}}{r} + \frac{\D{\mathcal{Q}_Y}{\mathcal{P}_Y}}{s} &&\Rightarrow\\
       &\ep[(X, Y) \sim \mathcal{Q}_{XY}]{\phi_X(X) + \phi_Y(Y)} - \log \ep[(X, Y) \sim \mathcal{P}_{XY}]{\exp(\phi_X(X) + \phi_Y(Y))}
           \\&\quad\ge \ep[X \sim \mathcal{Q}_{X}]{\phi_X(X)} - \log \ep[X \sim \mathcal{P}_{X}]{\exp(\phi_X(X))} + \ep[Y \sim \mathcal{Q}_{Y}]{\phi_Y(Y)} - \log \ep[X \sim \mathcal{P}_{Y}]{\exp(\phi_Y(Y))} &&\Leftrightarrow\\
       &\log \ep[(X, Y) \sim \mathcal{P}_{XY}]{\exp(\phi_X(X) + \phi_Y(Y)} \le \frac{\log \ep[X \sim \mathcal{P}_X]{\exp(r\phi_X(X))}}{r} + \frac{\log \ep[Y \sim \mathcal{P}_Y]{\exp(s\phi_Y(Y)}}{s} \; ,
   \end{align}
   which proves that $\eqref{eq:general-divergence} \Rightarrow \eqref{eq:general-hypercontractive}$.
\end{proof}

\subsection{Explicit Hypercontractive Bounds}\label{sec:explicit_hyp}
In this section we show how to relate the directed noise operator to the lower bounds of Oleszkiewicz~\cite{oleszkiewicz2003nonsymmetric}, thereby giving direct lower bounds for a number of cases for $s$ and $r$.
By \Cref{thm:lower_main} and \Cref{lem:hypercontractive-to-lower-bound} this is the dual to proving optimal values $(t_q,t_u)$ in our upper bound.

We start by with a standard lemma which shows that hypercontractivity of an operator
implies hypercontractivity of its adjoint.
\begin{lemma}\label{lem:dual-equivalence}
   Let $T : L_2(\Omega, \pi) \to L_2(\Omega, \pi')$ be an operator with $T^{*} : L_2(\Omega, \pi') \to L_2(\Omega, \pi)$ being its
   adjoint, and let $1 \le r, s < \infty$ with $r', s'$ being their convex conjugates.
   Then
   \begin{align}
      \norm{T f}{L_{s'}(\pi')} \le \norm{f}{L_r(\pi)}
   \end{align}
   holds for all $f \in L_2(\Omega, \pi)$, if and only if
   \begin{align}
      \langle T f, g \rangle_{L_2(\pi')} = \langle f, T^* g \rangle_{L_2(\pi)} \le \norm{f}{L_r(\pi)}\norm{g}{L_{s}(\pi')}
   \end{align}
   holds for all $f \in L_2(\Omega, \pi)$ and all $g \in L_2(\Omega, \pi')$, if and only if
   \begin{align}
      \norm{T^* g}{L_{r'}(\pi)} \le \norm{g}{L_{s}(\pi')}
   \end{align}
   holds for all $g \in L_2(\Omega, \pi')$.
\end{lemma}
\begin{proof}
   We assume that $\norm{T f}{L_{s'}(\pi')} \le \norm{f}{L_r(\pi)}$ holds for all $f \in L_2(\Omega, \pi)$.
   Let $f \in L_2(\Omega, \pi)$ and $g \in L_2(\Omega, \pi')$ then by Hölder's inequality we have that
   \begin{align}
      \langle T f, g \rangle_{L_2(\pi')} \le \norm{T f}{L_{s'}(\pi')}\norm{g}{L_{s}(\pi')} \le \norm{f}{L_r(\pi)}\norm{g}{L_s(\pi')} \; .
   \end{align}

   Similarly, we assume that $\norm{T^* g}{L_{r'}(\pi)} \le \norm{g}{L_{s}(\pi')}$ holds for all $g \in L_2(\Omega, \pi')$.
   Let $f \in L_2(\Omega, \pi)$ and $g \in L_2(\Omega, \pi')$ then by Hölder's inequality we have that
   \begin{align}
      \langle f, T^* g \rangle_{L_2(\pi)} \le \norm{f}{L_{r}(\pi')}\norm{T^* g}{L_{r'}(\pi')} \le \norm{f}{L_r(\pi)}\norm{g}{L_{s'}(\pi')} \; .
   \end{align}
 
   Finally, we assume that $\langle T f, g \rangle_{L_2(\pi')} \le \norm{f}{L_r(\pi)}\norm{g}{L_{s}(\pi')}$ holds for
   all $f \in L_2(\Omega, \pi)$ and all $g \in L_2(\Omega, \pi')$. Let $f \in L_2(\Omega, \pi)$ then using that
   $L_{s}(\pi')$ is the dual norm of $L_{s'}(\pi')$ we get that
   \begin{align}
      \norm{T f}{L_{s'}(\pi')}
         = \sup_{\norm{g}{L_{s}(\pi')} = 1} \langle T f, g \rangle_{L_2(\pi')}
         \le \sup_{\norm{g}{L_{s}(\pi')} = 1} \norm{f}{L_s(\pi)}\norm{g}{L_{r}(\pi')}
         = \norm{f}{L_r(\pi)} \; .
   \end{align}
   Similarly, let $g \in L_2(\Omega, \pi')$ then using that $L_{r}(\pi)$ is the dual norm of $L_{r'}(\pi)$ we get that
   \begin{align}
      \norm{T^* g}{L_{r'}(\pi)}
         = \sup_{\norm{f}{L_{r}(\pi)} = 1} \langle f, T^* g \rangle_{L_2(\pi)}
         \le \sup_{\norm{f}{L_{r}(\pi')} = 1} \norm{f}{L_r(\pi)}\norm{g}{L_{s}(\pi')}
         = \norm{g}{L_{s}(\pi')} \; ,
   \end{align}
   which finishes the proof.
\end{proof}
 
Our hypercontractive results will be based on the tight hypercontractive inequality by
Oleszkiewicz~\cite{oleszkiewicz2003nonsymmetric}.
\begin{theorem}[\cite{oleszkiewicz2003nonsymmetric}]\label{thm:Oleszkiewicz}
   Let $p \in (0, \tfrac{1}{2}) \cup (\tfrac{1}{2}, 1)$ and $1 \le r \le 2$ then for any
   function $f \in L_2(\set{0, 1}^d, \pi_p^{\otimes d})$ we have that
   \begin{align}
      \norm{T^{(p)}_\rho f}{L_2(p)} \le \norm{f}{L_r(p)} \; ,F
   \end{align}
   where $\rho = p^{-1/2}(1 - p)^{-1/2}\sqrt{\frac{(1 - p)^{2 - 2/r} - p^{2 - 2/r}}{p^{-2/r} - (1 - p)^{-2/r}}}$
   which is best possible.
\end{theorem}

From this we get following tight hypercontractive inequalities for $T^{p_1 \to p_2}_\rho$.
\begin{corollary}\label{cor:directed-2-norm}
   Let $p_1, p_2 \in (0, \tfrac{1}{2}) \cup (\tfrac{1}{2}, 1)$ and $1 \le r \le 2$ then for any
   function $f \in L_2(\set{0, 1}^d, \pi_{p_1}^{\otimes d})$ we have that
   \begin{align}
      \norm{T^{p_1 \to p_2}_{\rho} f}{L_2(p_2)} \le \norm{f}{L_{r}(p_1)} \; ,
   \end{align}
   where $\rho = p_1^{-1/2}(1 - p_1)^{-1/2}\sqrt{\frac{(1 - p_1)^{2 - 2/r} - p_1^{2 - 2/r}}{p_1^{-2/r} - (1 - p_1)^{-2/r}}}$
   which is best possible.
\end{corollary}
\begin{proof}
   Using the Parseval-Plancherel identity and \Cref{lem:directed-noise} we get that
   \begin{align}
      \norm{T^{p_1 \to p_2}_{\rho} f}{L_2(p_2)}^2
         = \sum_{S \subseteq [d]}\widehat{T^{p_1 \to p_2}_{\rho} f}^{(p_2)}(S)^2
         = \sum_{S \subseteq [d]} \rho^{2\abs{S}} \widehat{T^{p_1 \to p_2}_{\rho} f}^{(p_2)}(S)^2
         = \norm{T^{(p_1)} f}{L_2(p_1)}^2 \; ,
   \end{align}
   hence the result follows from \Cref{thm:Oleszkiewicz}.
\end{proof}
\begin{corollary}\label{cor:directed-p-to-2}
   Let $p_1, p_2 \in (0, \tfrac{1}{2}) \cup (\tfrac{1}{2}, 1)$ and $1 \le s \le 2 $ with $s'$ the convex conjugate
   of $s$, then for any function $f \in L_2(\set{0, 1}^d, \pi_{p_1}^{\otimes d})$ we have that
   \begin{align}
      \norm{T^{p_1 \to p_2}_{\rho} f}{L_{s'}(p_2)} \le \norm{f}{L_2(p_1)} \; ,
   \end{align}
   where $\rho = p_2^{-1/2}(1 - p_2)^{-1/2}\sqrt{\frac{(1 - p_2)^{2 - 2/s} - p_2^{2 - 2/s}}{p_2^{-2/s} - (1 - p_2)^{-2/s}}}$
   which is best possible.
\end{corollary}
\begin{proof}
   By \Cref{lem:dual-equivalence} we get that the result is true if and only if
   \begin{align}
      \norm{T^{p_2 \to p_1}_{\rho} g}{L_2(p_1)} \le \norm{g}{L_s(p_2)} \; ,
   \end{align}
   for all functions $g \in L_2(\set{0, 1}, \pi_{p_2})$. Now the result follows by
   using \Cref{cor:directed-2-norm}.
\end{proof}

We also get a hypercontractive inequality for the standard noise operator $T^{(p)}_\rho$.
\begin{corollary}\label{cor:hypercontractive-balanced}
   Let $p \in (0, \tfrac{1}{2}) \cup (\tfrac{1}{2}, 1)$ and $1 \le r \le 2$ with convex conjugate $r'$,
   then for any
   function $f \in L_2(\set{0, 1}^d, \pi_p^{\otimes d})$ we have that
   \begin{align}
      \norm{T^{(p)}_\rho f}{L_{r'}(p)} \le \norm{f}{L_r(p)} \; ,
   \end{align}
   where $\rho = p (1 - p) \frac{(1 - p)^{2 - 2/r} - p^{2 - 2/r}}{p^{-2/r} - (1 - p)^{-2/r}}$
   which is best possible.
\end{corollary}
\begin{proof}
   Using \Cref{lem:dual-equivalence} we get the result holds if and only if
   \begin{align}
      \langle T^{(p)}_\rho f, g \rangle_{L_2(p)} \le \norm{f}{L_r(p)}\norm{g}{L_r(p)} \; ,
   \end{align}
   holds for all $f, g \in L_2(\set{0, 1}^d, \pi_p^{\otimes d})$. First we note that the result
   is true by using Cauchy-Schwartz and \Cref{thm:Oleszkiewicz}
   \begin{align}
      \langle T^{(p)}_\rho f, g \rangle_{L_2(p)}
         = \langle T^{(p)}_{\sqrt{\rho}} f, T^{(p)}_{\sqrt{\rho}} g \rangle_{L_2(p)}
         \le \norm{T^{(p)}_{\sqrt{\rho}} f}{L_2{p}} \norm{T^{(p)}_{\sqrt{\rho}} g}{L_2{p}}
         \le \norm{f}{L_r{p}} \norm{f}{L_r{p}}
   \end{align}
   Now to see that $\rho$ is best possible we set $g = f$ which give us that
   \begin{align}
      \norm{T^{(p)}_{\sqrt{\rho}} f}{L_2{p}}^2
         = \langle T^{(p)}_\rho f, f \rangle_{L_2(p)}
         \le \norm{f}{L_r{p}}^2 \; ,
   \end{align}
   so \Cref{thm:Oleszkiewicz} gives that $\rho$ is best possible.
\end{proof}

We will use \Cref{cor:hypercontractive-balanced} to show that setting $(t_q, t_u) = (1 - w, 1 - w)$ is an
optimal threshold for the $(w, w, w_1, w^2)$-GapSS problem. First of we note that
$\rho = p (1 - p) \frac{(1 - p)^{2 - 2/r} - p^{2 - 2/r}}{p^{-2/r} - (1 - p)^{-2/r}}$
can be rewritten as $r = \frac{2 \log\tau}{\log \frac{\rho + \tau}{\rho + \tau^{-1}}}$
where $\tau = \frac{1 - p}{p}$. Using \Cref{lem:hypercontractive-to-lower-bound} we get that
$\rho = \frac{w_1 - w^2}{w(1 - w)}$, $\tau = \frac{1 - w}{w}$, and that
\begin{align}
   \frac{1}{r}\rho_q + \frac{1}{r'}\rho_u \ge \frac{2}{r} - 1 \; .
\end{align}
Now we have that $\rho_q = \rho_u = \frac{\log \frac{w (w_1 - 2w + 1)}{w_1(1 - w)}}{\log \frac{w}{1 - w}}$,
and we find that
\begin{align}
   \frac{2}{r} - 1
      = \frac{\log \frac{\rho + \tau}{\rho + \tau^{-1}}}{\log \tau} - 1
      = \frac{\log  \tau^{-1}\frac{\rho + \tau}{\rho + \tau^{-1}}}{\log \tau} \; .
\end{align}
It is then easy to check that $\tau^{-1}\frac{\rho + \tau}{\rho + \tau^{-1}} = \frac{w_1(1 - w)}{w(w_1 - 2w + 1)}$,
which then shows that $(t_q, t_u) = (1 - w, 1 - w)$ is an

\section{Other Algorithms}

We show two results that, while orthogonal to Supermajorities, help us understand them and how they fit within the space of Similarity Search algorithms.

The first result is an optimal affine embedding of sets onto the sphere.
This result is interesting in its own right, as it results in an algorithm that is in many cases better than the state of the art, and which can be implemented very easily in systems that can already solve Euclidean or Spherical Nearest Neighbours.
The result gives a simple, general condition a Spherical LSH scheme must meet for the embedding to be optimal, and we show that both SimHash and Spherical LSH meets it.

The second result is also a new algorithm.
In particular, it is a mix between Chosen Path and MinHash, which always achieves $\rho$ values lower than both of them.
It is in a sense a simple answer to the open problem in~\cite{tobias2016} about how to beat MinHash consistently.
More interesting though, is that it sheds light on what makes Supermajorities work:
It balances the amount of information pulled from sets vs. their complements.
The proof is also conceptually interesting, since it proves that it is never advantageous to combine multiple Locality Sensitive Filter families.

\subsection{Embedding onto the Sphere}\label{sec:embed}
\begin{figure}
   \centering
   \begin{subfigure}[t]{.47\textwidth}
      \centering
      \includegraphics[scale=0.5]{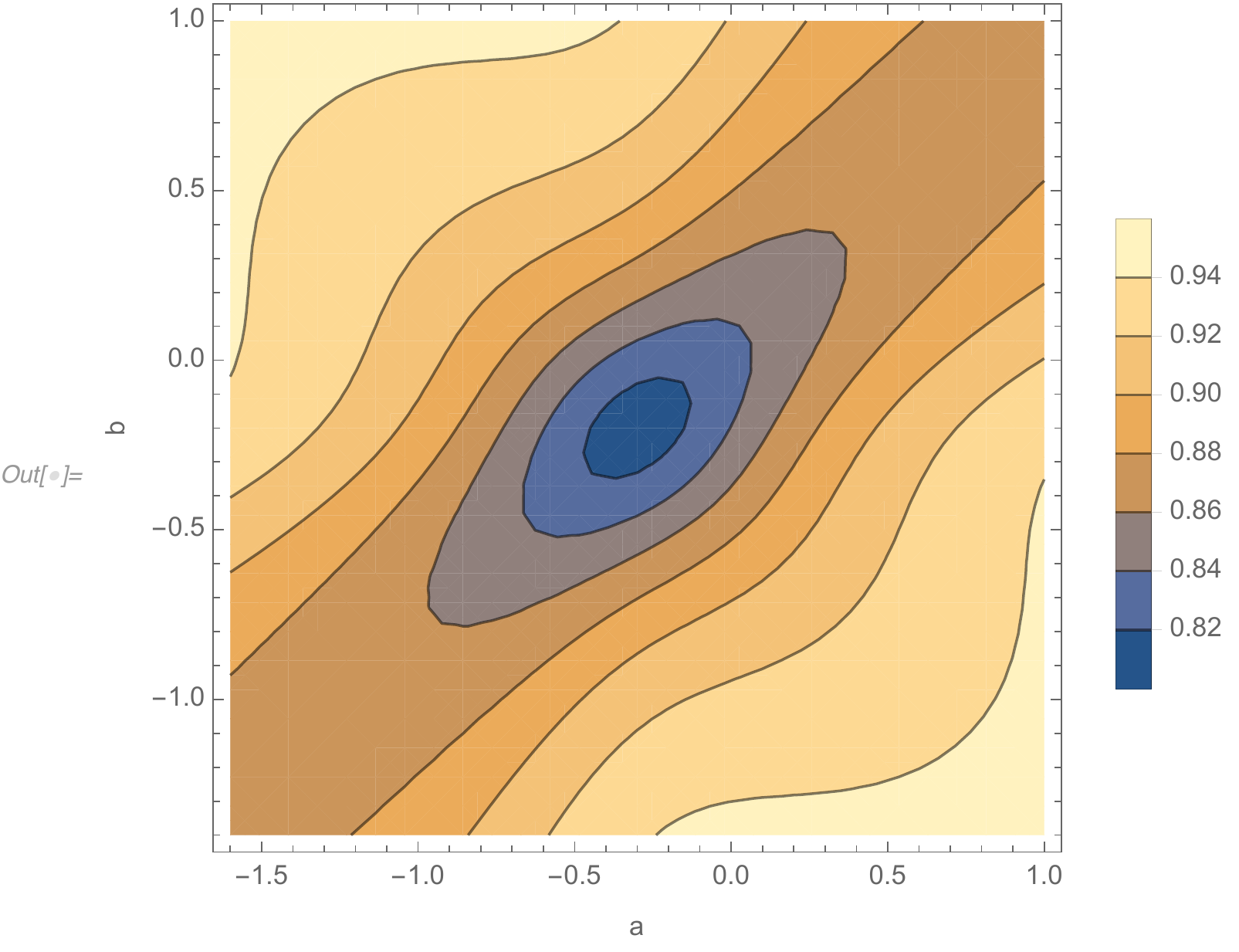}
      \caption{Query time/space exponent, $\rho$, for the SimHash algorithm~\cite{charikar2002new}.}
    \end{subfigure}
    \hspace{.5em}
   \begin{subfigure}[t]{.47\textwidth}
      \centering
      \includegraphics[scale=0.5]{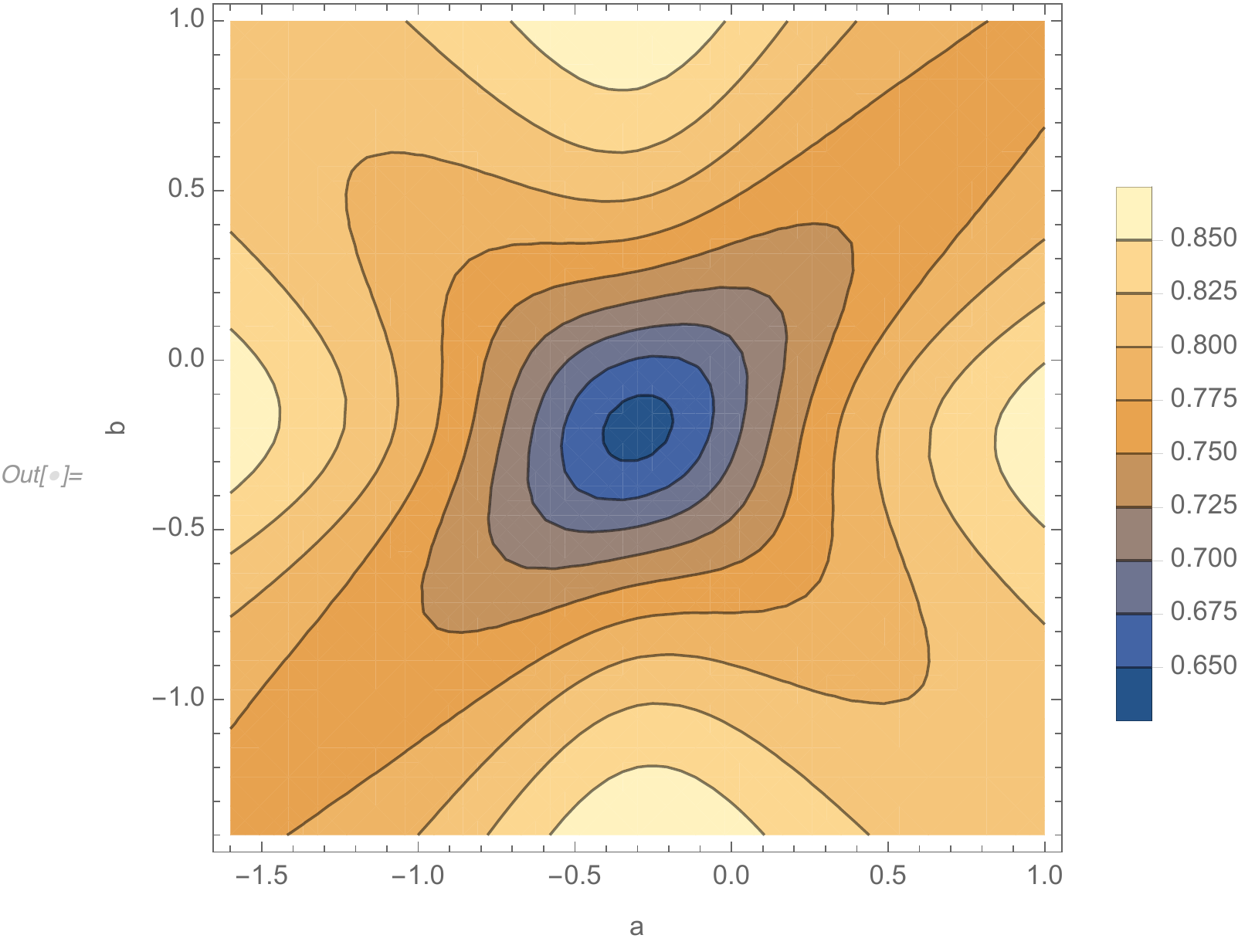}
      \caption{Query time/space exponent, $\rho$, for the Spherical LSH algorithm~\cite{terasawa2007spherical}.}
    \end{subfigure}
    \caption{Given a GapSS instance with $w_q=.3$ and $w_u=.2$, the optimal affine embedding of the data (represented as vectors $x\in\{0,1\}^{|U|}$) onto the sphere, turns out to be normalizing the ``mean'' and ``variance''.
    That is, before scaling down to $\|x\|_2=1$, we subtract respectively $w_q$ and $w_u$ from all coordinates.
    The plot shows the ``$\rho$-value'' achieved by different spherical algorithms as the among subtracted is varied: The x-axis, $a$, is the amount subtracted from queries and the y-axis, $b$, is the amount subtracted from datasets.
 }
   \label{fig:batchnorm}
\end{figure}
We show, that if an algorithm has exponent $\rho(\alpha, \beta)=f(\alpha)/f(\beta)$ where $\alpha$ is the cosine similarity between good points and $\beta$ is the similarity between bad points on the sphere;
then assuming some light properties on $f$, which contain both Spherical and Hyperplane LSH,
two affine embedding of sets $x\in\{0,1\}^d$ to $S^{d-1}$ that minimizes $\rho$ once the new cosine similarities are calculated, is $x\mapsto (x - w)/\sqrt{w(1-w)}$ where $w=|x|/d$.
While the mapping is allowed to depend on any of the GapSS parameters, it curiously only cares about the weight of the set itself.
For fairness, all our plots, such as \Cref{fig:sp_comparison}, uses this embedding when comparing Supermajorities to Spherical LSH.

\begin{lemma}[Embedding Lemma]\label{lem:embedding}
   Let $g,h: \{0,1\}^d\to\R^d$ be function on the form $g(x) = a_1 x + b_1$ and $h(y) = a_2 y + b_2$.
   Let $\rho(x,y,y')=f(\alpha(x,y))/f(\alpha(x,y'))$ where $\alpha(x,y) = \langle x,y\rangle/\|x\|\|y\|$ be such that
   \begin{align}
      f(z)\ge0
      ,\quad
      \tfrac{d}{dz}\left( (\pm1-z)\tfrac{d}{dz}\log f(z)\right)\ge0
      \quad\text{and}\quad
      \tfrac{d^3}{dz^3}\log f(z)\le0
   \end{align}
   for all $z\in[-1,1]$.
   Assume we know that $\|x\|^2_2=w_q d$, $\|y\|^2_2=w_u d$, $\langle x,y'\rangle=w_1 d$ and $\langle x,y\rangle=w_2 d$, then
   $ \argmin_{a_1,a_2,b_1,b_2}\rho(g(x),h(y),h(y')) = (1,1,-w_q,-w_u) $.
\end{lemma}

In this section we will show that Hyperplane~\cite{charikar2002similarity} and Spherical~\cite{andoni2016optimal} LSH both satisfy the requirements of the lemma.
Hence we get two algorithms with $\rho$-values:
\begin{align}
   \rho_\text{hp} = \frac{\log(1-\arccos(\alpha)/\pi)}{\log(1-\arccos(\beta)/\pi)}
   ,
   \quad
   \rho_\text{sp} = \frac{1-\alpha}{1+\alpha}\frac{1+\beta}{1-\beta}
   .
\end{align}
where
$\alpha=\frac{w_1-w_qw_u}{\sqrt{w_q(1-w_q)w_u(1-w_u)}}$
and $\beta=\frac{w_2-w_qw_u}{\sqrt{w_q(1-w_q)w_u(1-w_u)}}$,
and space/time trade-offs using the $\rho_q, \rho_u$ values in~\cite{christiani2016framework}.
\footnote{Unfortunately the space/time aren't on a form applicable to \cref{lem:embedding}.
From numerical experiments we however still conjecture that the embedding is optimal for those as well.}
\Cref{fig:batchnorm} shows how $\rho$ varies with different translations $a,b$.

Taking $t_q = w_q(1+o(1))$ and $t_u = w_u(1+o(1))$ in \cref{thm:main} recovers $\rho_\text{sp}$ by standard arguments.
This implies that~\cref{thm:main} dominates Spherical LSH (for binary data).

\begin{lemma}
   The functions
   $f(z)=(1-z)/(1+z)$ for Spherical LSH
   and $f(z)=-\log(1-\arccos(z)/\pi)$ for Hyperplane LSH
   satisfy \cref{lem:embedding}.
\end{lemma}
\begin{proof}
   For Spherical LSH we have $f(z)=(1-z)/(1+z)$ and get
   \begin{align}
      \tfrac{d}{dz}\left( (\pm1-z)\tfrac{d}{dz}\log f(z)\right)
      &= 2/(1\pm z)^2 \ge0,
      \\
      \tfrac{d^3}{dz^3}\log f(z)
      &= -4(1+3z^2)/(1-z^2)^3\le 0.
   \end{align}

   For Hyperplane LSH we have $f(z)=-\log(1-\arccos(z)/\pi)$ and get
   \begin{align}
      \tfrac{d}{dz}\left( (\pm1-z)\tfrac{d}{dz}\log f(z)\right)&=
      \frac{
      (\arccos(z)\mp\sqrt{1-z^2}-\pi)\log(1-\arccos(z)/\pi)\mp\sqrt{1-z^2}}
      {(1\pm z)\sqrt{1-z^2}(\pi-\arccos(z))^2\log(1-\arccos(z)/\pi)^2}.
   \end{align}
   In both cases the denominator is positive, and the numerator
   can be shown to be likewise by applying the inequalities
   $\sqrt{1-z^2}\le\arccos(z)$, $\sqrt{1-z^2}+\arccos(z)\le\pi$ and $x\le\log(1+x)$.

   The $\tfrac{d^3}{dz^3}\log f(z)\le0$ requirement is a bit trickier, but a numerical optimization shows that it's in fact less than $-1.53$.
\end{proof}

Finally we prove the embedding lemma:

\begin{proof}[Proof of \cref{lem:embedding}]
   We have
   \begin{align}
      \alpha
      =\frac{\langle x+a, y+b\rangle}{\|x+a\|\|y+b\|}
      =\frac{w_1+w_q b+w_u a+ab}{
         \sqrt{(w_q(1+a)^2+(1-w_q)a^2)(w_u(1+b)^2+(1-w_u)b^2)}}
   \end{align}
   and equivalent with $w_2$ for $\beta$.
   We'd like to show that $a=-w_q$, $b=-w_u$ is a minimum for
   $\rho=f(\alpha)/f(\beta)$.

   Unfortunately the $f$'s we are interested in are usually not convex, so it is not even clear that there is just one minimum.
   To proceed, we make the following substitution
   $a\to (c+d)\sqrt{w_q(1-w_q)}-w_q$,
   $b\to (c-d)\sqrt{w_u(1-w_u)}-w_u$
   to get
   \begin{align}
      \alpha(c,d) = \frac{cd+\frac{w_1-w_qw_u}{\sqrt{w_q(1-w_q)w_u(1-w_u)}}}{\sqrt{(1+c^2)(1+d^2)}}.
   \end{align}
   We can further substitute
   $cd\mapsto rs$ and $\sqrt{(1+c^2)(1+d^2)}\mapsto r+1$
   or $r\ge0$, $-1\le s\le 1$,
   since $1+cd\le\sqrt{(1+c^2)(1+d^2)}$ by Cauchy Schwartz, and $(cd,\sqrt{(1+c^2)(1+d^2)})$ can take all values in this region.

   The goal is now to show that $h=f\left(\frac{rs+x}{r+1}\right)\big/f\left(\frac{rs+y}{r+1}\right)$,
   where $1\ge x\ge y\ge -1$,
   is increasing in $r$.
   This will imply that the optimal value for $c$ and $d$ is 0, which further implies that $a=-w_q$, $b=-w_u$ for the lemma.

   We first show that $h$ is quasi-concave in $s$, so we may limit ourselves to $s=\pm1$.
   Note that $\log h = \log f\left(\frac{rs+x}{r+1}\right)
   - \log f\left(\frac{rs+y}{r+1}\right)$, and that
   $
      \frac{d^2}{ds^2}\log f\left(\frac{rs+x}{r+1}\right)
      =
      \left(\frac{r}{1+r}\right)^2
         \frac{d^2}{dz^2}\log f(z)
   $ by the chain rule.
   Hence it follows from the assumptions that $h$ is log-concave, which implies quasi-concavity as needed.

   We now consider $s=\pm1$ to be a constant.
   We need to show that
   $\frac{d}{dr} h \ge 0$.
   Calculating,
   \begin{align}
   \frac{d}{dr} f\left(\frac{rs+x}{r+1}\right)\big/f\left(\frac{rs+y}{r+1}\right)
   =
   \frac{
         (s-x)f\left(\frac{rs+y}{r+1}\right)f'\left(\frac{rs+x}{r+1}\right)
         -(s-y)f\left(\frac{rs+x}{r+1}\right)f'\left(\frac{rs+y}{r+1}\right)
      }{(1+r)^2f\left(\frac{rs+y}{r+1}\right)^2}.
   \end{align}
   Since $f\ge 0$ it suffices to show
   $ \frac{d}{dx} (s-x)f'\left(\frac{rs+x}{r+1}\right)\big/f\left(\frac{rs+x}{r+1}\right)\ge 0 $.
   If we substitute $z=\frac{rs+x}{r+1}$
   , $z\in[-1,1]$,
   we can write the requirement as
   $\frac{d}{dz} (s-z) f'(z)/f(z)\ge 0$
   or
   $\frac{d}{dz}\left( (\pm1-z)\frac{d}{dz}\log f(z)\right)\ge0$.
\end{proof}

\subsection{A MinHash Dominating Family}\label{sec:minhashdom}

Consider the classical MinHash scheme:
A permutation $h : [d]\to[d]$ is sampled at random,
and $y\subseteq \{0,1\}^d$ is placed in bucket $i\in[m]$ if $h(i)\in y$ and $\forall_{j<i} h(j)\not\in y$.
The probability for a collision between two sets $q,y$ is then $|q\cap y|/(|q|+|y|-|q\cap y|)$ by a standard argument which implies an exponent of $\rho_\text{mh}=\log\frac{w_1}{w_q+w_u-w_1}\big/\log\frac{w_2}{w_q+w_u-w_2}$.

Now consider building multiple independent such MinHash tables, but \emph{keeping only the $k$th bucket in each one.}
That gives a Locality Sensitive Filter family, which we will analyse in this section.

The Locality Sensitive Filter approach to similarity search is an extension by Becker et al.~\cite{becker2016new} to the Locality Sensitive Hashing framework by Indyk and Motwani~\cite{indyk1998approximate}.
We will use the following definition by Christiani~\cite{christiani2016framework}, which we have slightly extended to support separate universes for query and data points:
\begin{definition}[LSF]\label{defn:lsf}
   Let $X$ and $Y$ be some universes, let $S : X\times Y\to \mathbb R$ be a similarity function, and let $\mathcal F$ be a probability distribution over $\{(Q, U) \mid Q \subseteq X, U \subseteq Y\}$.
   We say that F is $(s_1, s_2, p_1, p_2, p_q, p_u)$-sensitive if for all points $x\in X, y \in Y$ and $(Q, U)$ sampled randomly from $\mathcal F$ the following holds:
   \begin{enumerate}
      \item If $S(x,y) \ge s_1$ then $\Pr[x\in Q,y\in U]\ge p_1$.
      \item If $S(x,y) \le s_2$ then $\Pr[x\in Q,y\in U]\le p_2$.
      \item $\Pr[x\in Q]\le p_q$ and $\Pr[x\in U]\le p_u$.
   \end{enumerate}
   We refer to $(Q,U)$ as a filter and to $Q$ as the query filter and $U$ as the update filter.
\end{definition}

We first state the LSF-Symmetrization lemma implicit in~\cite{tobias2016}:
\begin{lemma}[LSF-Symmetrization]\label{lem:symmetrization}
   Given a $(p_1,p_2,p_q,p_u)$-sensitive LSF-family, we can create a new family that is
   $(p_1 \frac{q}{p}, p_2 \frac{q}{p}, q, q)$-sensitive,
   where $p=\max\{p_q,p_u\}$ and $q=\min\{p_q,p_u\}$.
\end{lemma}
For some values of $p_1,p_2,p_q,p_u$ this will be better than simply taking $\max(\rho_u, \rho_q)$.
In particular when symmetrization may reduce $\rho_u$ by a lot by reducing its denominator.
\begin{proof}
   W.l.o.g. assume $p_q\ge p_u$.
   When sampling a query filter, $Q\subseteq U$, pick a random number $\varrho\in[0,1]$.
   If $\varrho>p_u/p_q$ use $\emptyset$ instead of $Q$.
   The new family then has $p_q' = p_q \cdot p_u/p_q$ and so on giving the lemma.
\end{proof}

Getting back to MinHash, we note that the ``keeping only the $i$th bucket'' family discussed above, corresponds sampling a permutation $s$ of $Y$ and taking the filter
\begin{align}
U = \{x \mid s_i \in x \wedge s_0 \not\in x\wedge\dots\wedge s_{i-1}\not\in x\}
.
\end{align}
That is, the collection of $x$ such that the first $i-1$ values of $s$ are not in $x$ (since then $x$ would have been put in that earlier bucket), but the $i$th element of $s$ is in $x$ (since otherwise $x$ would have been put in a later bucket.)

Using just one of these families, combined with symmetrization, gives the $\rho$ value:\begin{align}
   \rho_i =
   \log\frac{(1-w_q-w_u+w_1)^i w_1}{\max\{(1-w_q)^i w_q,\, (1-w_u)^i w_u\}}\bigg/
   \log\frac{(1-w_q-w_u+w_2)^i w_2}{\max\{(1-w_q)^i w_q,\, (1-w_u)^i w_u\}}
   .
\end{align}

This scheme is a generalization of Chosen Path, since taking $i=0$ recovers exactly that algorithm.
However, as we increase $i$, we see that the weight gradually shifts from the \emph{present} elements (symbolized by $w_1$, $w_2$, $w_q$ and $w_u$) to the \emph{absent} elements (symbolized by $(1-w_q-w_u+w_q)$, etc.).

We will now show that for a given set of $(w_q, w_u, w_1, w_2)$ there is always an optimal $i$ which is better than using all of the $i$, which is what MinHash does.
The exact goal is to show
\begin{align}
   \rho_\text{mh} = \log\frac{w_1}{w_q+w_u-w_1}\big/\log\frac{w_2}{w_q+w_u-w_2} \ge \min_{i\ge0} \rho_i
   .
\end{align}

For this we show the following lemma, which intuitively says that it is never advantageous to combine multiple filter families:
\begin{lemma}\label{lem:quasi-concave}
   The function $f(x,y,z,t) = \log(\max\{x,y\}/z)/\log(\max\{x,y\}/t)$, defined for
   $\min\{x, y\} \ge z \ge t > 0$, is quasi-concave.
\end{lemma}
This means in particular that
\begin{align}
   \frac{\log(\max\{x+x',y+y'\}/(z+z'))}{\log(\max\{x+x',y+y'\}/(t+t'))}
   \ge \min\left\{
      \frac{\log(\max\{x,y\}/z)}{\log(\max\{x,y\}/t)},
      \frac{\log(\max\{x',y'\}/z')}{\log(\max\{x',y'\}/t')}
   \right\},
\end{align}
when the variables are in the range of the lemma.
\begin{proof}
   We need to show that the set
   \begin{align}
      \{(x,y,z,t) : \log(\max\{x,y\}/z)/\log(\max\{x,y\}/t) \ge \alpha\}
      =\{(x,y,z,t) : \max\{x,y\}^{1-\alpha}t^{\alpha} \ge z\}
   \end{align}
   is convex for all $\alpha\in[0,1]$ (since $z \ge t$ so $f(x,y,z,t)\in[0,1]$).
   This would follow if $g(x,y,t)=\max\{x,y\}^{1-\alpha}t^{\alpha}$ would be quasi-concave itself,
   and the eigenvalues of the Hessian of $g$ are exactly 0, 0 and
   $-(1-\alpha) \alpha t^{\alpha-2} \max \{x,y\}^{-\alpha-1} \left(\max \{x,y\}^2+t^2\right)$
   so $g$ is even concave!
\end{proof}
We can then show that MinHash is always dominated by one of the filters described, as
\begin{align}
\rho_\text{mh}
= \frac{\log\frac{w_1}{w_q+w_u-w_1}}{\log\frac{w_2}{w_q+w_u-w_2}}
&=
\frac{\log\frac
   {\sum_{i\ge 0} (1-w_q-w_u+w_1)^i w_1}
      {\max\{\sum_{i\ge 0} (1-w_q)^i w_q, \,\sum_{i\ge 0} (1-w_u)^i w_u\}}}
     {\log\frac
        {\sum_{i\ge 0} (1-w_q-w_u+w_2)^i w_2}
         {\max\{\sum_{i\ge 0} (1-w_q)^i w_q, \,\sum_{i\ge 0} (1-w_u)^i w_u\}}}
\ge \min_{i\ge 0}\frac
   {\log\frac{(1-w_q-w_u+w_1)^i w_1}{\max\{(1-w_q)^i w_q, (1-w_u)^i w_u\}}}
   {\log\frac{(1-w_q-w_u+w_2)^i w_2}{\max\{(1-w_q)^i w_q, (1-w_u)^i w_u\}}}
   ,
\end{align}
where the right hand side is exactly the symmetrization of the ``only bucket $i$'' filters.
By monotonicity of $(1-w_q)^iw_q$ and $(1-w_u)^iw_u$ we can further argue that it is even possible to limit ourselves to one of $i\in\{0,\infty,\log(w_q/w_u)/\log((1-w_q)/(1-w_u))\}$, where the first gives Chosen Path, the second gives Chosen Path on the complemented sets, and the last gives a balanced trade-off where $(1-w_q)^iw_q=(1-w_u)^iw_u$.

\section{Conclusion and Open Problems}

For a long time there was a debate~\cite{shrivastava2014defense} about why MinHash worked so well for sets, compared to other more general methods, like SimHash.
It was a mystery why this method, so foreign to the frameworks of Spherical LSF and Chosen Path could still do so much better.
For asymmetric problems like Subset Search, it was entirely open how far $\rho$ could be reduced.
\begin{center}
   \emph{This paper finally solves the mystery of MinHash and unifies the ideas and frameworks of Euclidean and Set Similarity Search.}
\end{center}

By showing that supermajorities indeed solve the general problem optimally, we not only unify and explain the performance of the previous literature, but also recover major performance improvements, space/time trade-offs, and the ability to solve Set Similarity Search for any similarity measure.

We propose the following open problems for future research:
\begin{description}
   \item[LSH with polylog time]
      When parametrized accordingly, we get a data structure with
      $e^{\tilde O(\sqrt{\log n})}$ query time and $n^{O(1)}$ space.
      Using Spherical LSH one can get similar runtime, though with a higher polynomial space usage.
      Employing a tighter analysis of our algorithm, the query time can be reduced to $e^{\tilde O((\log n)^{1/3})}$, which by comparison with 
      we conjecture is tight for the approach.
      A major open question is whether one can get $\tilde O(1)$?

   \item[Data-dependent]
      Data-dependent LSH is able to reduce approximate similarity search problems to the case where far points are as far away as had they been random.
      For $(w_q,w_u,w_1,w_2)$-GapSS this corresponds to the case $w_2=w_qw_u$.
      This would finally give the ``optimal'' algorithm for GapSS without any ``non-data-dependent'' disclaimers.
   \item[Sparse, non-binary data]
      Our lower bounds really hold for a much larger class of problems, including cosine similarity search on sparse data in $\R^d$.
      However, our upper bounds currently focus on binary data only.
      It would be interesting to generalize our algorithm to this and other types of data for which Supermajorities are also optimal.
   \item[Sketching]
      We have shown that Supermajorities can shave large polynomial factors of space and query time in LSH.
      Can they be used to give similar gains in the field of sketching sets under various similarity measures?
      Can one expand the work of~\cite{pagh2014min} and show optimality of some intersection sketching scheme?
\end{description}

\subsection{Acknowledgements}
We would like to thank Rasmus Pagh and Tobias Christiani on suggesting the problem discussed in this paper, as well as for reviews of previous manuscripts.
Much of the initial work in the paper was done while the first author was visiting Eric Price at the University of Texas, and we would like to thank people there for encouragement and discussions on Boolean functions.
Many people have helped reading versions of the manuscript, and we would like to thank Morgan Mingle, Morten Stöckel, John Kallaugher, Ninh Pham, Evangelos Kipouridis, Peter Rasmussen, Anders Aamand, Mikkel Thorup and everyone else who has given feedback.
Finally, we really appreciate the comprehensive comments from the anonymous reviewers!

Thomas D.\ Ahle and Jakob B.\ T.\ Knudsen are partly supported by Thorup's Investigator Grant 16582, Basic Algorithms Research Copenhagen (BARC), from the VILLUM Foundation.

\bibliographystyle{plainurl}
\bibliography{supermajority}

\appendix
\section{Proof of \Cref{lem:hypercontractive-to-lower-bound}}\label{app:proof-hypercontractive-to-lower-bound}
This proof in this section mostly follows~\cite{andoni2017optimal}, with a few changes to work with separate spaces $Q$ and $U$.

\begingroup
\def\thelemma{\ref{lem:hypercontractive-to-lower-bound}}
\begin{lemma}
   Let $Q$ and $U$ be some spaces and $\mathcal{P}_{QU}$ a probability distribution on $Q \times U$.
   Consider any list-of-points data structure for $\mathcal{P}_{QU}$-random instances of $n$ points,
   which uses expected space $n^{1 + \rho_u}$, has expected query time $n^{\rho_q - o_n(1)}$, 
   and succeeds with probability at least $0.99$. Let $r, s \in [1, \infty]$
   satisfy  
   \begin{align}
      \ep[(X, Y) \sim \mathcal{P}_{QU}]{f(X) g(Y)} \le \norm{f(X)}{L_r(\mathcal{P}_Q)} \norm{f(Y)}{L_s(\mathcal{P}_U)} \; ,
   \end{align}
   for all functions $f : Q \to \R$ and $g : U \to \R$. Then
   \begin{align}
      \frac{1}{r}\rho_q + \frac{1}{r'}\rho_u \ge \frac{1}{r} + \frac{1}{s} - 1 \; ,
   \end{align}
   where $r' = \frac{r}{r - 1}$ is the convex conjugate of $r$.
\end{lemma}
\addtocounter{lemma}{-1}
\endgroup
\begin{proof}
    Fix a data structure $D$, where $A_i$ specifies which dataset points are placed in $L_i$.
    Additionally, we define $B_i = \setbuilder{v}{i \in I(v)}$ to the set of query points
    which scan $L_i$. We sample a random dataset point $u$ and then a random query point $v$
    from the neighborhood of $u$. Let
    \begin{align}
       \gamma_i = \prpcond{v \in B_i}{u \in A_i}
    \end{align}
    represent the probability that query $v$ scans the list $L_i$ conditioned on $u$ being
    in $L_i$. The query time for $D$ is given by the following expression
    \begin{align}
       T &= \sum_{i \in [m]}[v \in B_i]\left(1 + \sum_{j \in [n]}[u_j \in A_i] \right) \\
       \ep{T} &= \sum_{i \in [m]} \prp{v \in B_i}
          + \sum_{i \in [m]} \gamma_i \prp{u \in A_i}
          + (n - 1)\sum_{i \in [m]} \prp{u \in A_i} \prp{v \in B_i} \; .
    \end{align}
    We want to lower bound $\prp{v \in B_i}$, so let $1 \le r, s$ be any values such that
    $\mathcal{P}_{QU}$ is $(r, s)$-hypercontractive. We then get that
    \begin{align}
       \gamma_i \prp{u \in A_i}
          &= \prp{u \in A_i \wedge v \in B_i}
          \\&= \ep{[u \in A_i][v \in B_i]}
          \\&\le \norm{[u \in A_i]}{L_s(p_1)} \norm{[v \in B_i]}{L_r(p_2)}
          \\&= \prp{u \in A_i}^{1/s} \prp{v \in B_i}^{1/r}
    \end{align}
    Hence we get that $\prp{v \in B_i} \ge \gamma_i^r \prp{u \in A_i}^{r/s'}$. We define
    $\tau_i = \prp{u \in A_i}$ and get that
    \begin{align}
       \ep{T}
          \ge \sum_{i \in [m]} \gamma_i^r \tau_i^{r/s'}
          + \sum_{i \in [m]} \gamma_i \tau_i
          + (n - 1)\sum_{i \in [m]} \gamma_i^r \tau_i^{1 + r/ s'} \; .
    \end{align}
    Since the data structure succeeds with probability $\gamma$ we have that
    \begin{align}
       \sum_{i \in [m]} \tau_i \gamma_i
          \ge \prp{\exists i \in [m] : v \in B_i, u \in A_i}
          = \gamma \; .
    \end{align}
    Since $D$ uses at most $S$ space we have that
    \begin{align}
       m + \sum_{i \in [m]} \abs{A_i} \le S \;\Rightarrow\;
       \sum_{i \in [m]} \tau_i \le \frac{S}{n} \; .
    \end{align}
    We then get that we want to minimize
    \begin{align}
       \ep{T}
          \ge \sum_{i \in [m]} \gamma_i^r \tau_i^{r/s'}
          + \sum_{i \in [m]} \gamma_i \tau_i
          + (n - 1)\sum_{i \in [m]} \gamma_i^r \tau_i^{1 + r/s'}
          \ge \sum_{i \in [m]} \gamma_i^{r}\tau_i^r (\tau_i^{-r/s} + (n - 1)\tau_i^{1 - r/s}) \; ,
    \end{align}
    given the constraints
    \begin{align}
       \sum_{i \in [m]} \tau_i \gamma_i &\ge \gamma \\
       \sum_{i \in [m]} \tau_i &\le \frac{S}{n} \; .
    \end{align}
    First we fix $(\tau_i)_{i \in [m]}$ and minimize the function with
    respect to $(\gamma_i)_{i \in [m]}$.
    Using Lagrange multipliers this is equivalent to minimizing the function
    \begin{align}
       f((\gamma_i)_{i \in [m]}, \lambda, \nu)
          = \sum_{i \in [m]} \gamma_i^{r}\tau_i^r (\tau_i^{-r/s} + (n - 1)\tau_i^{1 - r/s})
          - \lambda (\sum_{i \in [m]} \tau_i \gamma_i - \gamma - \nu^2)
    \end{align}
    We find the critical points $\nabla f = 0$:
    \begin{align}
       r \gamma_i^{r - 1} \tau_i^{r} (\tau_i^{-r/s} + (n - 1)\tau_i^{1 - r/s})
          &= \lambda \tau_i \\
       \sum_{i \in [m]} \tau_i \gamma_i &= \gamma + \nu^2 \\ 
       2\lambda \nu &= 0
    \end{align}
    for all $i \in [m]$. We note that since $\gamma > 0$ then $\lambda > 0$
    and hence $\nu = 0$.
    The first inequality can be rewritten as
    \begin{align}
       \gamma_i^{r - 1} \tau_i^{r - 1} &= \frac{\lambda}{r (\tau_i^{-r/s} + (n - 1)\tau_i^{1 - r/s})}
          &&\Leftrightarrow\\
       \gamma_i \tau_i &= \left(\frac{\lambda}{r (\tau_i^{-r/s} + (n - 1)\tau_i^{1 - r/s})}\right)^{r'/r}
    \end{align}
    Combining this with $\sum_{i \in [m]} \tau_i \gamma_i = \gamma$ give us that
    \begin{align}
       \sum_{i \in [m]} \left(\frac{\lambda}{r (\tau_i^{-r/s} + (n - 1)\tau_i^{1 - r/s})}\right)^{r'/r}
          &= \gamma &&\Leftrightarrow\\
       \lambda^{r'/r} &= \frac{\gamma}{\sum_{i \in [m]}\left( \frac{1}{r (\tau_i^{-r/s} + (n - 1)\tau_i^{1 - r/s})} \right)^{r'/r} }
    \end{align}
    We define $t_i = \left( \frac{1}{\tau_i^{-r/s} + (n - 1)\tau_i^{1 - r/s}} \right)^{r'/r}$
    and get that
    \begin{align}
       \gamma_i \tau_i &= \gamma \frac{t_i}{\sum_{i \in [m]} t_i} &&\Leftrightarrow\\
       \gamma_i^r \tau_i^r &= \gamma^r \frac{t_i^r}{(\sum_{i \in [m]} t_i)^r}
    \end{align}
    We then get that our original function becomes
    \begin{align}
       \gamma^r \sum_{i \in [m]} \frac{t_i^r}{(\sum_{i \in [m]} t_i)^r} t_i^{-r/r'}
          = \gamma^r \sum_{i \in [m]} \frac{t_i}{(\sum_{i \in [m]} t_i)^r}
          = \gamma^r (\sum_{i \in [m]} t_i)^{-(r - 1)}
          = \gamma^r (\sum_{i \in [m]} t_i)^{-r/r'}
    \end{align}
    So we want to maximize
    \begin{align}
       \sum_{i \in [m]} t_i
          = \sum_{i \in [m]} \left( \frac{1}{\tau_i^{-r/s} + (n - 1)\tau_i^{1 - r/s}} \right)^{r'/r}
          = \sum_{i \in [m]} \frac{\tau_i^{r'/s}}{(1 + (n - 1)\tau_i)^{r'/r}}
    \end{align}
    We now consider two different cases.
 
    \paragraph{Case 1. $r > s$.}
    We know that $\tau_i \le 1$ so we get that
    \begin{align}
       \sum_{i \in [m]} \frac{\tau_i^{r'/s}}{(1 + (n - 1)\tau_i)^{r'/r}}
          \le \sum_{i \in [m]} \frac{\tau_i^{r'(1/s - 1/r)}}{n^{r'/r}}
    \end{align}
    Since $r'(1/s - 1/r) > 0$ then we can use the power-mean inequality to get
    \begin{align}
       \sum_{i \in [m]} \frac{\tau_i^{r'(1/s - 1/r)}}{n^{r'/r}}
          &\le \frac{m}{n^{r'/r}} \left(\frac{\sum_{i \in [m]} \tau_i}{m} \right)^{r'(1/s - 1/r)}
          \\&\le \frac{m^{r' - r'/s}}{n^{r'/r}} \left(\frac{S}{n}\right)^{r'(1/s - 1/r)}
          \\&= \frac{m^{r'/s'}}{n^{r'/s}} S^{r'(1/s - 1/r)}
          \\&\le \frac{S^{r'/s' + r'(1/s - 1/r)}}{n^{r'/s}}
          \\&= \frac{S}{n^{r'/s}}
    \end{align}
    where we have used that $\max\set{m, n\sum_{i \in [m]} \tau_i} \le S$.
 
    \paragraph{Case 2. $r \le s$}
    We find the derivatives
    \begin{align}
       &\frac{\frac{r'}{s}t_i^{r'/s - 1}(1 + (n - 1)\tau_i)^{r'/r} - (n - 1)\frac{r'}{r}(1 + (n - 1)\tau_i)^{r'/r - 1}t_i^{r'/s}}{(1 + (n - 1)\tau_i)^{2r'/r}}
          \\&= \frac{r' t_i^{r'/s - 1}}{(1 + (n - 1)\tau_i)^{r'/r + 1}}\left(\frac{1}{s}(1 + (n - 1)\tau_i) - (n - 1)\frac{1}{r}\tau_i \right)
    \end{align}
    \paragraph{Case 2.1. $r < s$}
    We note that the function is maximized when we 
    set $\tau_i = \frac{\frac{1}{s}}{(n - 1)(\frac{1}{r} - \frac{1}{s})} = \frac{r}{(n - 1)(s - r)}$.
    This give us that
    \begin{align}
       \sum_{i \in [m]} \frac{\tau_i^{r'/s}}{(1 + (n - 1)\tau_i)^{r'/r}}
          \le m \frac{\left(\frac{r}{(n - 1)(s - r)}\right)^{r'/s}}{\left(1 + \frac{r}{s - r}\right)^{r'/r}}
          \le \frac{S}{n^{r'/s}} \frac{\left(\frac{2r}{s - r} \right)^{r'/s}}{\left(\frac{s}{s - r}\right)^{r'/r}}
    \end{align}
    where we have used that $m \le S$ and $n \ge 2$.
 
    \begin{align}
       m\frac{r}{(n - 1)(s - r)} \le \frac{S}{n} \Rightarrow
       m \le S \frac{(n - 1)(s - r)}{n r}
    \end{align}
 
    \paragraph{Case 2.2. $r = s$}
    We note that the function is increasing in $\tau_i$ so it is maximized when $\tau_i = 1$.
    Then we get that
    \begin{align}
       \sum_{i \in [m]} \frac{\tau_i^{r'/s}}{(1 + (n - 1)\tau_i)^{r'/r}}
          \le \frac{m}{n^{r'/r}}
          = \frac{S}{n^{r'/r}}
          = \frac{S}{n^{r'/s}}
    \end{align}
    where we have used that $m \le S$ and $r = s$.
 
    From this we note that if we set $K = \max\set{1, \frac{\left(\frac{2r}{s - r} \right)^{r'/s}}{\left(\frac{s}{s - r}\right)^{r'/r}}}$
    then $\sum_{i \in [m]} \frac{\tau_i^{r'/s}}{(1 + (n - 1)\tau_i)^{r'/r}} \le \frac{S}{n^{r'/s}} K$.
    Now we can give the final lower bound on $\ep{T}$:
    \begin{align}
       \ep{T}
          \ge \gamma^r (\sum_{i \in [m]} t_i)^{-r/r'}
          \ge \gamma^r \left( \frac{S}{n^{r'/s}} K \right)^{-r/r'}
          = \gamma^r K^{-r/r'} S^{-r/r'} n^{r/s}
    \end{align}
    From this we get the result we want
    \begin{align}
       \rho_q &\ge -\frac{r}{r'}(1 + \rho_u) + \frac{r}{s} - o_n(1) \Leftrightarrow\\
       \frac{1}{r}\rho_q + \frac{1}{r'}\rho_u &\ge \frac{1}{s} - \frac{1}{r'} - o_n(1)
    \end{align}
 \end{proof}

\end{document}